\documentclass[12pt]{article}
\date{March 9, 2016}
\usepackage{amsmath,amsthm,amsfonts,amssymb,multicol}
\usepackage{graphicx}
\usepackage{booktabs}
\usepackage{array}
\usepackage{subfigure}
\usepackage{enumerate}
\usepackage{url}

\usepackage[english]{babel}
\usepackage{verbatim} %multi-line comments
\usepackage{amsthm} %theoreme & definitionen
\usepackage{enumitem}
\usepackage{enumerate}
\usepackage{bbm} %math symbols
\usepackage[normalem]{ulem} %dashed & pointed underline
\usepackage{bm}
\usepackage{mathtools}
\usepackage{authblk}
\usepackage{fullpage}
\numberwithin{equation}{section}

\newcommand{\be}{\begin}
\newcommand{\e}{\end}
\newcommand{\beq}{\begin{equation}}
\newcommand{\eeq}{\end{equation}}
\newcommand{\beqs}{\begin{equation*}}
\newcommand{\eeqs}{\end{equation*}}

\renewcommand{\l}{\left}
\renewcommand{\r}{\right}
\renewcommand{\d}{\mathrm{d}} %straight d for integrals
\newcommand{\vt}{\vartheta}
\newcommand{\vp}{\varphi}
\newcommand{\vr}{\varrho}
\newcommand{\Arg}{\textnormal{Arg}}
\newcommand{\p}{\mathbf{p}}
\newcommand{\x}{\mathbf{x}}
\newcommand{\y}{\mathbf{y}}

\newcommand{\set}[1]{\mathbb{#1}}
\newcommand{\rt}{\set{R}^3}

\newcommand{\curly}[1]{\mathcal{#1}}
\newcommand{\goth}[1]{\mathfrak{#1}}
\newcommand{\setof}[2]{\left\{ #1\; : \;#2 \right\}}

\newcommand{\om}{\omega}
\newcommand{\eps}{\varepsilon}
\newcommand{\lam}{\lambda}

\newcommand{\gam}{\gamma}
\newcommand{\Gam}{\Gamma}

\newcommand{\al}{\alpha}
\newcommand{\de}{\delta}
\newcommand{\Del}{\Delta}
\newcommand{\ttmatrix}[4]{\left(\be{array}{cc} #1&#2\\	#3&#4 \e{array}	\right)}

\newcommand{\scp}[2]{\langle#1,#2\rangle}
\newcommand{\spa}{\textnormal{span}}

\newcommand{\jap}[1]{\langle #1\rangle}
\renewcommand{\it}{\infty}

\newcommand{\del}{\partial}
\newcommand{\dpp}{\,\d \p} %integration measure

\newcommand{\dx}{\,\d \x}

\newcommand{\ol}{\overline}

\newcommand{\Tra}[2]{\textnormal{Tr}_{#1}\left[#2\right]}	%Trace
\newcommand{\Tr}[1]{\textnormal{Tr}\left[#1\right]}	%Trace
\renewcommand{\Tr}[1]{\int \Tra{\set{C}^2}{#1} \dpp}

\newcommand{\ft}[1]{\widehat{{#1}}}			%Fourier Transform
\newtheorem{thm}{Theorem}[section]
\newtheorem{lm}[thm]{Lemma}

\theoremstyle{definition}

\theoremstyle{remark}

\def\dotuline{\bgroup
  \ifdim\ULdepth=\maxdimen  % Set depth based on font, if not set already
   \settodepth\ULdepth{(j}\advance\ULdepth.4pt\fi
  \markoverwith{\begingroup
  \advance\ULdepth0.08ex
  \lower\ULdepth\hbox{\kern.15em .\kern.1em}%
  \endgroup}\ULon}

\def\dashuline{\bgroup
  \ifdim\ULdepth=\maxdimen  % Set depth based on font, if not set already
   \settodepth\ULdepth{(j}\advance\ULdepth.4pt\fi
  \markoverwith{\kern.15em
  \vtop{\kern\ULdepth \hrule width .3em}%
  \kern.15em}\ULon}

\begin{document}

\title{Multi-component Ginzburg-Landau theory: microscopic derivation and examples}
\author{Rupert L. Frank\thanks{rlfrank@caltech.edu} }
\author{Marius Lemm\thanks{mlemm@caltech.edu} }
\affil{Department of Mathematics, Caltech}

\maketitle

\abstract{This paper consists of three parts. In part I, we microscopically derive Ginzburg--Landau (GL) theory from BCS theory for trans\-lation-invariant systems in which multiple types of superconductivity may coexist. Our motivation are unconventional superconductors. We allow the ground state of the effective gap operator $K_{T_c}+V$ to be $n$-fold degenerate and the resulting GL theory then couples $n$ order parameters. 

In part II, we study examples of multi-component GL theories which arise from an isotropic BCS theory. We study the cases of (a) pure $d$-wave order parameters and (b) mixed $(s+d)$-wave order parameters, in two and three dimensions.

In part III, we present explicit choices of spherically symmetric interactions $V$ which produce the examples in part II. In fact, we find interactions $V$ which produce ground state sectors of $K_{T_c}+V$ of \emph{arbitrary} angular momentum, for open sets of of parameter values. This is in stark contrast with Schr\"odinger operators $-\nabla^2+V$, for which the ground state is always non-degenerate. Along the way, we prove the following fact about Bessel functions: At its first maximum, a half-integer Bessel function is strictly larger than all other half-integer Bessel functions. 
}

\setcounter{tocdepth}{1}
\tableofcontents

\section{Introduction}
Since its advent in 1950 \cite{GL50}, Ginzburg--Landau (GL) theory has become ubiquitous in the description of superconductors and superfluids near their critical temperature $T_c$. GL theory is a phenomenological theory that describes the superconductor on a macroscopic scale.
Apart from being a very successful physical theory, it also has a rich mathematical structure which has been extensively studied, see e.g.\ \cite{CorreggiRougerie14, FournaisHelffer, GustafsonSigal, SandierSerfaty} and references therein. Microscopically, superconductivity arises due to an effective attraction betweeen electrons, causing them to condense into Cooper pairs. In 1957 Bardeen, Cooper and Schrieffer \cite{BCS57}, were the first to explain the origin of the attractive interaction in crystalline $s$-wave superconductors. By integrating out phonon modes, they arrived at their effective ``BCS theory'', in which one restricts to a certain class of trial states now known as BCS states. In 1959, Gor'kov \cite{Gorkov59} argued how the microscopic BCS theory with a rank-one interaction gives rise to the macroscopic GL theory near $T_c$. An alternative argument is due to de Gennes \cite{deGennes}.

The first mathematically rigorous proof that Ginzburg--Landau theory arises from BCS theory,
 on macroscopic length scales and for temperatures close to $T_c$, was given in \cite{FrankHainzlSeiringerSolovej12} under the non-degeneracy assumption that there is only one type of superconductivity present in the system. The derivation there allows for local interactions and external fields and hence applies to superfluid ultracold Fermi gases, a topic of considerable current interest.

In the present paper, we use the same formalism as in \cite{FrankHainzlSeiringerSolovej12} and study microscopically derived Ginzburg--Landau theories involving multiple types of superconductivity for systems without external fields.

We first discuss the main result of part I, which forms the basis for the applications in parts II and III. Afterwards, we discuss the physical motivation for studying multi-component GL theories and the extent to which our model applies to realistic systems. The introduction closes with a description of the main results of parts II and III.

\subsection{Main result of part I}	
As in \cite{FrankHainzlSeiringerSolovej12}, we employ a variational formulation of BCS theory \cite{BachLiebSolovej94, Leggett79} with an isotropic electronic dispersion relation.  We use previous rigorous results about this theory in the absence of external fields \cite{FrankHainzlNabokoSeiringer07, HainzlHamzaSeiringerSolovej08, HainzlSeiringer08PR, HainzlSeiringer08LMP}. Particularly important is the result of \cite{HainzlHamzaSeiringerSolovej08} that the critical temperature $T_c$ can be characterized by the following linear criterion. $T_c$ is the unique value of $T\geq 0$ for which the ``effective gap operator'' 
$$
\label{eq:KT}
K_T+V(\x)=\frac{-\nabla^2-\mu}{\tanh\l(\frac{-\nabla^2-\mu}{2T}\r)}+V(\x)$$
 has zero as its lowest eigenvalue. Here $V$ is the electron-electron interaction potential. Throughout the  microscopic derivation of GL theory in \cite{FrankHainzlSeiringerSolovej12}, it is \emph{assumed that zero is a non-degenerate eigenvalue} of $K_{T_c}+V$. For radially symmetric $V$, this means that the order parameter is an $s$-wave, i.e.\ it is spherically symmetric. 
 
 The \textbf{main result of part I}, Theorem \ref{thm:mainTI}, is that for systems without external fields the microscopic derivation of GL theory also holds when the eigenvalue is \emph{degenerate} of arbitrary order $n>1$. (A general argument shows that always $n<\it$.) The arising GL theory now features precisely $n$ order parameters $\psi_1,\ldots, \psi_n$. It turns out that one can use the same general strategy as in \cite{FrankHainzlSeiringerSolovej12}. 
 
In fact, one can classify approximate minimizers of the BCS free energy via the GL theory. Given an orthonormal basis $\{a_1,\ldots, a_n\}$ of $\ker (K_{T_c}+V)$, Theorem \ref{thm:mainTI} (ii) says that, near the critical temperature, the Cooper pair wave function $\al$ of a BCS state of almost minimal free energy (i.e.\ the Cooper pair wave function realized by the physical system) is approximately given by a linear combination of the $\{a_1,\ldots, a_n\}$ of the form
\beq
\label{eq:content}
		\al \approx \sqrt{\frac{T_c-T}{T_c}} \sum_{j=1}^n \psi_j a_j,
\eeq
where the ``amplitudes'' $\psi_1,\ldots, \psi_n\in\set{C}$ almost minimize the corresponding GL function.

The results of \cite{FrankHainzlSeiringerSolovej12} allow for the presence of weak external fields which vary on the macroscopic scale. A key step is to establish semiclassical estimates under weak regularity assumptions. We emphasize that in our case the system has no external fields and is therefore translation-invariant. This simplifies several technical difficulties present in \cite{FrankHainzlSeiringerSolovej12}. In particular, the semiclassical analysis of \cite{FrankHainzlSeiringerSolovej12} reduces to an ordinary Taylor expansion. The result of the expansion is stated as Theorem \ref{thm:SC} and we give the simplified proof for the translation-invariant situation. We do this (a) to obtain optimal error bounds and (b) to hopefully make the emergence of GL theory more transparent in our technically simpler situation. 

\subsection{Physical motivation} 
\paragraph{Background.}
The degenerate case corresponds to systems which have \emph{multiple order parameters}, i.e.\ which can host multiple types of superconductivity. Physically, this situation occurs e.g.\ for \emph{unconventional superconductors}. By definition, these are materials in which an effective attractive interaction of electrons leads to the formation of Cooper pairs, but the effective attraction is not produced by the usual electron-phonon interactions. (Identifying the underlying mechanisms is a major open problem in condensed matter physics.)

Two important classes of unconventional superconductors are the layered cuprates and iron-based compounds, typically designed to have large values of $T_c$ (``high-temperature superconductors''). Many of these materials possess tetragonal lattice symmetry, though the prominent example of YBCO has orthorhombic symmetry. There is strong experimental evidence for the occurrence of $d$-wave order parameters in these materials, in contrast to the pure $s$-wave order parameter in conventional superconductors. More precisely, phase-sensitive experiments with Josephson junctions \cite{Kirtleyetal95,Kirtleyetal06,KirtleyTsuei00,Tsueietal97,Wollmanetal93} have evidenced the presence of a $d_{x^2-y^2}$-wave order parameter (for tetragonal symmetry) and of mixed $(s+d_{x^2-y^2})$-wave order parameters (for orthorhombic symmetry). 

There also exist proposals of $d$-wave super\emph{fluidity} for molecules in optical lattices \cite{Kunsetal}.

\paragraph{Multi-component Ginzburg--Landau theories.}
On the theoretical side, one of the most important tools for studying unconventional superconductors are multi-component Ginzburg-Landau theories \cite{Askerzade03, Berlinskyetal95,Joynt90,Leggett, XuRenTing95,  SoininenKallinBerlinsky94, Volovik93,  WangWangLi99,XuRenTing96, ZhuKimTingHu98}. Many of these papers study the symmetry properties near the vortex cores in two-component GL theories. A very common example is a GL theory with $(s+{d_{x^2-y^2}})$-wave order parameters; this case has also been studied mathematically in \cite{Du99, KimPhillips12}. The effect of an an anisotropic order parameter on the upper critical field was studied in \cite{Langmann92}. %In recent years, the conditions for the applicability of multi-component GL theories have also been the subject of closer investigation \cite{SilaevBabaev12}.

Another avenue where two-component GL theories have been successful is in the description of type I.5 superconductors \cite{BS05, CBS, SB11}. These are systems in which the magnetic field penetration depth lies in between the coherence lengths of the different order parameters (of course this effect only manifests itself in an external magnetic field).

\paragraph{Microscopically derived GL theories.}
In many of the papers cited above, the GL theories that are studied are first obtained microscopically by using Gor'kov's formal expansion of Green's functions. The advantage of having a microscopically derived GL theory is that it has some remaining ``microscopic content''. By this we mean: 

\be{enumerate}
\item One can directly associate each macroscopic order parameters with a certain symmetry type of the system's Cooper pair wave function. Therefore, if we can classify the minimizers of the microscopically derived GL theory, we understand exactly which Cooper pair wave functions $\al$ can occur in the physical system in configurations of almost minimal free energy. 
\item One has explicit formulae for computing the GL coefficients as integrals over microscopic quantities.
\e{enumerate}

The first point is expressed by \eqref{eq:content} above and is therefore a corollary of Theorem \ref{thm:mainTI}. The second point is represented by formulae \eqref{eq:cdefn},\eqref{eq:ddefn} in Theorem \ref{thm:mainTI}.

While the papers cited above provide important insight about the vortex structure in unconventional superconductors, they are restricted in that the GL theories are obtained using the formal Gorkov procedure and that almost exclusively two-component GL theories are studied. Our Theorem \ref{thm:mainTI} provides a rigorous microscopic derivation of $n$-component GL theories with
 $n$ arbitrary starting from a BCS theory with an isotropic electronic dispersion. 
 
 \paragraph{Physical assumptions of our model.}
We discuss the main physical assumptions of our model and the resulting limitations in its applicability to realistic systems.
\be{enumerate}
\item [(a)] \emph{Translation-invariance}. We view the degenerate translation-invariant systems as toy models for multi-component superconductivity. We believe that the examples of multi-component GL theories studied in part II are already rich enough to show that the translation-invariant case can be interesting. From a technical perspective, translation invariance yields major technical simplifications. In particular, the semiclassical analysis of \cite{FrankHainzlSeiringerSolovej12} reduces to a Taylor expansion. 

\item[(b)] \emph{BCS theory with a Fermi-Dirac normal state.} There are two assumptions here: First, we start from a BCS theory (meaning a theory in which electrons can form Cooper pairs and which restricts to BCS-type trial states). The question whether such a theory can be used to describe unconventional superconductors is unresolved \cite{Leggett}. Second, we work with a BCS theory for which the normal state is given by the usual Fermi-Dirac distribution. Most realistic unconventional superconductors are strongly interacting systems with a non-Fermi liquid normal state \cite{Leggett,Senthil08}. 
%We would like to point out, however, that we derive GL theory close to the critical temperature where also an unconventional superconductor could be in a state that is close to the normal state of a free Fermi gas. Put differently, the electron-electron attraction \emph{energy} will be small near $T_c$ because there are few electrons that attract.

\item[(c)] \emph{Isotropy.} 
We study a BCS theory in which the electrons live in the continuum and have an isotropic dispersion. Many of the known examples of unconventional superconductors are layered compounds which are effectively two-dimensional. When we say that their order parameter has $d_{x^2-y^2}$-wave symmetry, then this only means that it has a four-lobed shape similar to that of $k_x^2-k_y^2$ for $-\pi< k_x,k_y< \pi$, but its precise dependence on $k_x,k_y$ depends on the symmetry group of the two-dimensional lattice \cite{Leggett}. Order parameters of the form $k_x^2-k_y^2$ have been studied as a first approximation to unconventional superconductors, see e.g.\ \cite{XuRenTing95, WangWangLi99,XuRenTing96}. 

For the examples in part II, we consider a spherically symmetric interaction potential, resulting in a fully isotropic BCS theory. Consequently, the $d$-wave order parameters that we consider are the ``usual'' ones, known from atomic physics (see section \ref{sect:dwave}). By isotropy, all the $d$-waves (there are two in two dimensions and five in three dimensions) are energetically equal. The examples in part II show that even this isotropic microscopic theory can lead to rather rich coupling phenomena of anisotropic macroscopic order parameters, as we discuss next.

\item[(d)] \emph{Spin singlet order parameter.} We restrict to order parameters which are singlets in spin space. This is indeed the case for unconventional superconductors \cite{Leggett}, but it excludes systems with $p$-wave order parameters such as  superfluid Helium-3.
\e{enumerate}

\subsection{Main results of part II}
In part II, we compute the $n$-component GL theories that arises from the BCS theory according to Theorem \ref{thm:mainTI} for several exemplary cases. For each situation, we make some observations about the minimizers of the GL energy and their symmetries and give a physical interpretation.

\emph{Throughout part II, $V$ is assumed to be spherically symmetric}, so the BCS theory becomes fully isotropic. The order parameters can then be described by the decomposition into angular momentum sectors (see section \ref{sect:dwave}) and we consider the case of pure $d$-wave and mixed $(s+d)$-wave order parameter. Here and in the following, we write \emph{``GLn''} for ``$n$-component Ginzburg--Landau theory''. The dimension $D$ will be either two or three.
 
\be{enumerate}[label=(\roman*)]
 \item Let $D=3$. Assume the Cooper pair wave function is a linear combination of the five linearly independent $d$-waves with a given radial part. \textbf{Theorem \ref{thm:pureTI}} explicitly computes the microscopically derived GL5 energy and gives a full description of all its minimizers. Surprisingly, the GL5 energy in three dimensions exhibits the emergent symmetry group $O(5)$, see Corollary \ref{cor:symmgroup} (i), which is considerably larger
      than the original $O(3)$ symmetry group coming from the spherical symmetry and reflection symmetry of $V$. %A noteworthy consequence of the characterization of minimizers is that no single $d$-wave can minimize the GL energy, a result we refer to as ``non-trivial coupling of $d$-channels'', see Corollary \ref{cor:symmgroup} (iii). 
      
      \item Let $D=2$. Assume the Cooper pair wave function is a linear combination of the two linearly independent $d$-waves with a given radial part. \textbf{Theorem \ref{thm:pureTI2D}} explicitly computes the microscopically derived GL2 energy and gives a full description of all its minimizers. We find that the $(d_{x^2-y^2},d_{xy})$ order parameter must be of the form $(\psi,\pm i\psi)$ with $|\psi|$ minimizing an appropriate GL1. In particular, the minimizers of this GL2 form a double cover of the minimizer of a GL1. 

 \item Let $D=3$. Assume the Cooper pair wave function is a linear combination of the five linearly independent $d$-waves with a given radial part \emph{and} the $s$-wave with another given radial part. \textbf{Theorem \ref{thm:mixedTI}} explicitly computes the microscopically derived GL6 energy. It also gives a simple characterization of the parameter values for which the pure $d$-wave minimum is always unstable under $s$-wave perturbations and of the parameter values for which, vice-versa, the pure $s$-wave minimum is unstable under $d$-wave perturbations. As a consequence, we give parameter values for which $s$- and $d$-waves must \emph{couple non-trivially} to be energy-minimizing.
\e{enumerate}

We also consider the mixed $(s+d)$-wave case in $D=2$ dimensions. The result is presented in Remark \ref{rmk:mixed} (v) for brevity.

%We repeat that isotropy implies that all $d$-waves are energetically equal. Hence, the most realistic of our examples from above is the two-dimensional situation of Theorem \ref{thm:pureTI2D} which describes three-dimensional systems with cylindrical symmetry. Theorems \ref{thm:pureTI} and \ref{thm:mainTI}, while displaying some mathematically interesting behavior, are only pertinent to hypothetical three-dimensional, isotropic $d$-wave superconductors. 

\subsection{Main results of part III}
Recall from the discussion of part I above, that the candidate Cooper pair wave functions are the ground states of the effective gap operator $K_{T_c}+V$. A priori, it is not at all clear that the fully isotropic BCS theory can produce ground state sectors of $K_{T_c}+V$ which are not spherically symmetric. In particular, it is not clear that the examples considered in part II actually exist.

In fact, if $K_{T_c}$ is replaced by the Laplacian $-\nabla^2$ we have a Schr\"odinger operator and under very general conditions on the potential $V$, the Perron-Frobenius theorem implies that the ground state is in fact \emph{non-degenerate}, see e.g.\ Theorem 11.8 in \cite{LiebLoss}. For spherically symmetric $V$, this means the ground state is also spherically symmetric (``$s$-wave''). 

In part III, we remedy this by exhibiting examples of spherically symmetry potentials $V$ such that the ground state sector of $K_{T_c}+V$ can in fact have \emph{arbitrary} angular momentum. These potentials will be of the form
\beqs
V_{\lam,R}(\x)=-\lam \de(|\x|-R)
\eeqs
in three dimensions. Here $\lam$ and $R$ are positive parameters. The result holds for open intervals of the parameters values, so it is ``not un-generic''.

\section{Part I: Microscopic derivation of GL theory in the degenerate case}
\subsection{BCS theory}
We consider a gas of fermions in $\set{R}^D$ with $1\leq D\leq 3$ at temperature $T> 0$ and chemical potential $\mu\in\set{R}$, interacting via the two-body potential $V(\x)$. We assume that $V(\x)=V(-\x)$ is reflection symmetric. We do not consider external fields, so the system is translation-invariant. A BCS state $\Gam$ can then be conveniently represented as a $2\times 2$ matrix-valued Fourier multiplier on $L^2(\rt)\oplus L^2(\rt)$ of the form
\beq
\label{eq:BCSstatedefn}
 \ft \Gam(\p) = \ttmatrix{ \ft{\gam}(\p)}{ \ft{\al}(\p)}{\ol{\ft{\al}(\p)}}{1-\ft{\gam}(\p)},
\eeq
for all $\p\in \set{R}^D$. Here, $\ft{\gam}(\p)$ denotes the Fourier transform of the one particle-density matrix and $\ft{\al}(\p)$ the Fourier transform of the Cooper pair wave function. We require $\ft{\al}(\p)=\ft{\al}(-\p)$ and $0\leq \Gam(\p) \leq 1$ as a matrix, which is equivalent to $0\leq \ft{\gam}(\p) \leq 1$ and  $|\ft{\al}(\p)|^2 \leq \ft{\gam}(\p) (1-\ft{\gam}(\p))$. The \emph{BCS free energy per unit volume} reads, in suitable units
\beq
\label{eq:Fdefn}
		\curly{F}^{BCS}_{T}(\Gam)=\int_{\set{R}^D} (\p^2-\mu) \ft{\gam}(\p) \dpp - T S[\Gam] + \int_{\set{R}^D} V(\x) |\al(\x)|^2 \dx,
\eeq
where the \emph{entropy per unit volume} is given by
\beq
		S[\Gam] = -\int_{\set{R}^D} \Tra{\set{C}^2}{\ft \Gam(\p) \log\ft \Gam(\p)} \dpp.
\eeq

\be{rmk}[BCS states]
\label{rmk:BCS}
\be{enumerate}[label=(\roman*)]
\item In general \cite{BachLiebSolovej94, FrankHainzlSeiringerSolovej12}, $SU(2)$-invariant BCS sta\-tes are represented as $2\times 2$ block operators 
\beqs
\Gam=\ttmatrix{\gam}{\al}{\ol{\al}}{1-\ol{\gam}}
\eeqs
where $\gam,\al$ are operators on $L^2(\set{R}^D)$ with kernel functions $\gam(\x,\y)$ and $\al(\x,\y)$ in $L^2(\set{R}^D)\oplus L^2(\set{R}^D)$. Since $0\leq \Gam\leq 1$ is Hermitian, $\gam(\x,\y)=\ol{\gam(\y,\x)}$ and $\al(\x,\y)=\al(\y,\x)$. In the translation-invariant case considered here, these kernel functions are assumed to be of the form $\gam(\x-\y)$ and $\al(\x-\y)$. Since convolution by $\gam,\al$ becomes multiplication in Fourier space, we can equivalently describe the BCS state by its Fourier transform $\ft\Gam$ defined in \eqref{eq:BCSstatedefn} above. In the translation-invariant case, the symmetries of $\gamma, \al$ turn into the relations $\gam(\x)=\ol{\gam(-\x)}$ and $\al(\x)=\al(-\x)$ or equivalently $\ft{\gam}(\p)=\ol{\ft{\gam}(\p)}$ and $\ft \al(\p)=\ft \al(-\p)$. Finally, since we are interested in states with minimal free energy, we may also assume
\beq
\label{eq:symmass}
\ft{\gam}(\p)=\ft\gam(-\p)
\eeq
and this was already used on the bottom right element in \eqref{eq:BCSstatedefn}. To see this, let $\ft \Gam$ be a BCS state not satisfying \eqref{eq:symmass}, set $\ft\Gam_r(\p):=\ft\Gam(-\p)$ and observe that 
\beqs
\curly{F}^{BCS}_{T}\l(\frac{\Gam+\Gam_r}{2}\r)< \curly{F}^{BCS}_{T}(\Gam)
\eeqs
by strict concavity of the entropy and reflection symmetry of all terms in $\curly{F}^{BCS}_{T}$.
\item Note that $\al(\x,\y)=\al(\y,\x)$ means that the Cooper pair wave function is symmetric in its arguments. To obtain a fermionic wave function, we would eventually tensor $\al$ with an antisymmetric spin singlet. Since $\al$ is reflection-symmetric in the translation-invariant case, $\al$ must be of even angular momentum if $V$ is radial.

The restriction to symmetric $\al$ is a consequence of assuming $SU(2)$ invariance in the heuristic derivation of the BCS free energy \cite{HainzlHamzaSeiringerSolovej08, Leggett79}. This means the full Cooper pair wave function must be a spin singlet and so its spatial part $\al$ must be symmetric. Note that this excludes systems, e.g.\ superfluid Helium-3, which display a $p$-wave order parameter.
\item 
For more background on the BCS functional, in particular a heuristic derivation from the many-body quantum Hamiltonian in which one restricts to quasi-free states, assumes $SU(2)$ invariance and drops the direct and exchange terms, see \cite{Leggett79} or the appendix in \cite{HainzlHamzaSeiringerSolovej08}. Recently, \cite{BraunlichHainzlSeiringer13} justified the last step for translation-invariant systems by proving that dropping the direct and exchange terms only leads to a renormalization of the chemical potential $\mu$, for a class of short-ranged potentials.
\e{enumerate}
\e{rmk}

We make the following technical assumption on the interaction potential.
\be{ass}
\label{ass:V}
We either have $V\in L^{p_V}(\set{R}^D)$ with $p_V=1$ for $D=1$, $1<p_V<\it$ for $D=2$ and $p_V=3/2$ for $D=3$, or we have 
\beq
\label{eq:Vdeltadefns}
    V(\x)= V_{\lam,R}(|\x|):= -\lam \de(|\x|-R),
\eeq
when $D=1,2,3$ and $\lam,R>0$.
\e{ass}

We note 

\be{prop}
\label{prop:formbounded}
A potential $V$ satisfying Assumption \ref{ass:V} is infinitesimally form-bounded with respect to $-\nabla^2$.
\e{prop}

We quote a result of \cite{HainzlHamzaSeiringerSolovej08}, which provides the foundation for studying the variational problem associated with $\curly{F}^{BCS}_{T}$. Define 
\beqs
\curly{D}:=\l\{	\Gam \textnormal{ as in \eqref{eq:BCSstatedefn}}\,:\, 0\leq \ft\Gam \leq 1,\, \ft{\gam}\in L^1(\set{R}^D,(1+\p^2)\dpp),\, \al\in H^1_{sym}(\set{R}^D)\r\}
\eeqs
with $H^1_{sym}(\set{R}^D)=\setof{\al\in H^1(\set{R}^D)}{\al(\x)=\al(-\x)\, a.e.}$.

\be{prop}[Prop.\ 2 in \cite{HainzlHamzaSeiringerSolovej08}]
\label{prop:HHNS}
Under Assumption \ref{ass:V} on $V$, the BCS free energy \eqref{eq:Fdefn} is bounded below on $\curly{D}$ and attains its minimum. 
\e{prop}

The physical interpretation rests on the following
\be{defn}[Superconductivity]
\label{defn:SC}
 The system described by $\curly{F}^{BCS}_{T}$ is \emph{superconducting} (or \emph{superfluid}, depending on the context) iff any minimizer $\Gam$ of $\curly{F}^{BCS}_{T}$ has off-diagonal entry $\al \not\equiv 0$.
\e{defn}

It was shown in \cite{HainzlHamzaSeiringerSolovej08} that the question whether the system is superconducting can be reduced to the following linear criterion, which we will use heavily. (In \cite{HainzlHamzaSeiringerSolovej08}, the results are proved for $D=3$ and without the restriction to the reflection-symmetric subspace of $L^2(\set{R}^D)$, but it was already observed in \cite{FrankHainzlSeiringerSolovej12} that the statement holds as stated here.)

%Intuitively, the question whether $\al \equiv 0$ or not is decided by the competition between the attractive $V$ term and the entropy $S[\Gam]$ in \eqref{eq:Fdefn}. Since the entropy comes with a weight $T$, it is not surprising that there is a unique critical temperature $T_c$ such that the minimizer has $\al\equiv 0$ iff $T\geq T_c$. 

\be{prop}[Theorems 1 and 2 in \cite{HainzlHamzaSeiringerSolovej08}]
Define the operator
\beq
\label{eq:KTdefn}
K_T:=\frac{-\nabla^2-\mu}{\tanh\l(\frac{-\nabla^2-\mu}{2T}\r)}
\eeq
as a Fourier multiplier and consider $K_T+V$ in the Hilbert space
\beq
\label{eq:L2sym}
L^2_{sym}(\set{R}^D):=\{f\in L^2(\set{R}^D)\,:\, f(\x)=f(-\x) \text{ a.e.}\}.
\eeq
Then:
\be{enumerate}[label=(\roman*)]
\item the system is superconducting in the sense of Definition \ref{defn:SC} iff $K_T+V$ has at least one negative eigenvalue.
\item there exists a unique critical temperature $0\leq T_c<\it$ such that
\beq
\begin{aligned}
\label{eq:Tcdefn}
 K_{T_c}+V&\geq 0,\\
 \inf\textnormal{spec}(K_T+V)&<0, \quad \forall T<T_c.
\e{aligned}
\eeq
\e{enumerate}
\e{prop}

$T_c$ is unique because the quadratic form associated with $K_T$ is strictly monotone in $T$. In a nutshell, the reason why the operator $K_T+V$ appears, is that it is the Hessian of the map
\beqs
		\phi\mapsto \curly{F}^{BCS}_{T}\l(\Gam_0+\ttmatrix{0}{\phi}{\ol{\phi}}{0}\r)
\eeqs
at $\phi=0$ with $\Gam_0$ the normal state of the system, see \eqref{eq:normalstate}, and naturally, the positivity of the Hessian is related to minimality. For the details, we refer to \cite{HainzlHamzaSeiringerSolovej08}.  In the following, we make

\be{ass}
\label{ass:Tc}
$V$ is such that $T_c>0$.
\e{ass}

By Theorem 3 in \cite{HainzlHamzaSeiringerSolovej08}, $V\leq 0$ and $V\not \equiv 0$ implies $T_c>0$ in $D=3$ and this result is stable under addition of a small positive part. 

\be{defn}[Ground-state degeneracy]
We set
\beq
\label{eq:ndefn}
n:=\dim\ker (K_{T_c}+V).
\eeq
\e{defn}

\be{rmk}
\be{enumerate}[label=(\roman*)]
\item We always have $n<\it$. The reason is that, by Assumption \ref{ass:V} on $V$, the essential spectrum of $K_T+V$ is contained in $[2T,\it)$. Therefore, zero is an isolated eigenvalue of $K_{T_c}+V$ of finite multiplicity and so $n<\it$. 
\item A sufficient condition for $n=1$ is that $\ft{V}\leq 0$ and $\ft{V}\not\equiv 0$ \cite{FrankHainzlNabokoSeiringer07,HainzlSeiringer08PR}. 
\item For Schr\"odinger operators $-\nabla^2+V$, the ground state is non-degenerate by the Perron-Frobenius theorem. That is, one always has the analogue of $n=1$ in that case. One may therefore wonder if $n>1$ ever holds. In part III, we present a class of radial potentials such that for open intervals of parameter values, we have $n>1$. In fact, one can tune the parameters such that $\ker (K_{T_c}+V)$ lies in an \emph{arbitrary} angular momentum sector.
\e{enumerate}
\e{rmk}

\subsection{GL theory}
In GL theory, one aims to find ``order parameters'' that minimize the GL energy. The minimizers then describe the
macroscopic relative density of superconducting charge carriers, up to spontaneous symmetry breaking. Microscopically, they describe the center of mass coordinate of the Cooper pair wave function $\al$. In our case, translation-invariance implies that the order parameters are complex-valued constants, which are non-zero iff the system is superconducting. 

When $n=1$ (and the system is translation-invariant), there is a single order parameter $\psi\in\set{C}$ and for $T<T_c$ the GL energy is of the all-familiar ``mexican hat'' shape
\beq
\label{eq:mexican}
\curly{E}^{GL}(\psi)=c|\psi|^4-d|\psi|^2,\qquad c,d>0.
\eeq
Below, in Theorem \ref{thm:mainTI}, we show that for $n>1$, the GL energy is of the form
\beq
\curly{E}^{GL}(\goth{a})=\int f_4(p) |\goth{a}(\p)|^4\d\p -\int f_2(p) |\goth{a}(\p)|^2\d\p
\eeq
and $\goth{a}$ varies over the $n$-dimensional set $\ker(K_{T_c}+V)$. The functions $f_4$ and $f_2$ are explicit; they are radial ($p\equiv |\p|$) and positive for $T<T_c$. 

Thus, we see that \emph{the mexican hat shape is characteristic for the transla\-tion-invariant case, even in the presence of degeneracies}. However, there exists \emph{nontrivial coupling} (i.e.\ mixed terms) between the different basis elements of $\ker(K_{T_c}+V)$ in general. 

\subsection{Result}
We write $\Gam_0$ for the minimizer of the free energy  $\curly{F}^{BCS}_{T}$ as in \eqref{eq:Fdefn} but with $V\equiv 0$. That is, $\Gam_0$ describes a free Fermi gas at temperature $T$ and for this reason we call $\Gam_0$ the ``normal state'' of the system. From the Euler-Lagrange equation, one easily obtains
\beq
\label{eq:normalstate}
	\ft\Gam_0(\p) = \ttmatrix{\ft{\gam}_0(\p)}{0}{0}{1-\ft{\gam}_0(\p)},
\eeq
where
\beq
	\ft{\gam}_0(\p) = \frac{1}{1+\exp((\p^2-\mu)/T)}
\eeq
is the well-known Fermi-Dirac distribution. (Of course, $\Gam_0$ depends on $\mu$ and $T$, but for the following we implicitly assume that it has the same values of $\mu,T$ as the free energy under consideration.)

We now state our first main result. It says that an appropriate $n$-component GL theory arises from BCS theory on the macroscopic scale and for temperatures close to $T_c$. Recall that $p\equiv |\p|$.

\be{thm}
\label{thm:mainTI}
\be{enumerate}[label=(\roman*)]
Let $V$ satisfy Assumptions \ref{ass:V} and \ref{ass:Tc} and let $\mu\in\set{R}$, $T<T_c$. Recall that $n=\dim\ker(K_{T_c}+V)$. Then:
\item As $T\uparrow T_c$,
\beq
\begin{aligned}
\label{eq:thmmainTI}
 &\min_{\Gam}\curly{F}^{BCS}_{T}(\Gam)-\curly{F}^{BCS}_{T}(\Gam_0)\\ 
 &= \l(\frac{T_c-T}{T_c}\r)^2 \min_{\goth{a}\in \ker(K_{T_c}+V)} \curly{E}^{GL}(\goth{a}) + O\l((T_c-T)^3\r),
\end{aligned}
\eeq
where $\curly{E}^{GL}$ is defined by 
\beq
\label{eq:alternative}
\begin{aligned}
\curly{E}^{GL}(\goth{a})
=&\frac{1}{T_c} \int_{\set{R}^D} \frac{g_1((p^2-\mu)/T_c)}{(p^2-\mu)/T_c} \l|K_{T_c}(p)\r|^4 |\goth{a}(\p)|^4\dpp\\
&-
\frac{1}{2T_c} \int_{\set{R}^D}  \frac{1}{\cosh^2 \l(\frac{p^2-\mu}{2T_c}\r)} \l|K_{T_c}(p)\r|^2 |\goth{a}(\p)|^2\dpp.
\end{aligned}
\eeq
Here we used the auxiliary functions
\beq
\label{eq:g0defn}
\begin{aligned}
		g_0(z):=&\frac{\tanh(z/2)}{z}\\
		g_1(z):=&-g_0'(z) = z^{-1} g_0(z) -\frac{1}{2} z^{-1} \frac{1}{\cosh^2 (z/2)}\\		
		K_T(p): =& \frac{p^2-\mu}{\tanh\l(\frac{p^2-\mu}{2T}\r)}.
\end{aligned}
\eeq
\item
Moreover, if $\Gam$ is an approximate minimizer of $\curly{F}^{BCS}_{T}$ in the sense that
\beq
\label{eq:approximatemin0}
	\curly{F}^{BCS}_{T}(\Gam)-\curly{F}^{BCS}_{T}(\Gam_0) = \l(\frac{T_c-T}{T_c}\r)^2 \l(\min_{\goth{a}\in\ker(K_{T_c}+V)} \curly{E}^{GL}(\goth{a}) +\eps\r),
\eeq
for some $0<\eps\leq M$, then we can decompose its off-diagonal element $\ft{\al}$ as
\beq
\label{eq:aldecomposition}
	\ft{\al}(\p) = \sqrt{\frac{T_c-T}{T_c}} \goth{a}_0(\p) + \xi,
\eeq
where $\|\xi\|_2= O_M\l(T_c-T\r)$ and $\goth{a}_0\in\ker(K_{T_c}+V)$ is an approximate minimizer of the GL energy, i.e.
$$\curly{E}^{GL}(\goth{a}_0)\leq \min \curly{E}^{GL}+\eps+O_M\l(T_c-T\r).$$ (Here $O_M$ means that the implicit constant depends on $M$.)
\e{enumerate}
\e{thm}

The idea is that near $T_c$, where superconductivity is weak, the normal state $\Gam_0$ is the prime
  competitor for the development of a small off-diagonal component $\ft{\al}$ of the BCS minimizer.
   Theorem \ref{thm:mainTI} then says that the lowest-order deviation from the normal state is well-described by a GLn whose coefficients are given explicitly as integrals over microscopic quantities. 
   
   \be{rmk}
\label{rmk:mainTI}   
   \be{enumerate}[label=(\roman*)]
   \item We can equivalently rewrite the GL energy in terms of ``order parameters'' $\psi_1,\ldots,\psi_n$ as follows. We fix an orthonormal basis $\{a_j\}$ of $\ker(K_{T_c}+V)$ and decompose $\goth{a}\in\ker(K_{T_c}+V)$ as $\goth{a}(\p)=\sum_{j=1}^n\psi_j \ft a_j(\p)$. The basis coefficients $\psi_1,\ldots,\psi_n \in \set{C}$ are the $n$ order parameters, each one corresponds to a different ``type'' of superconductivity $\ft a_j$. The GL energy \eqref{eq:alternative} can then be rewritten in the equivalent form 
\beq
\label{eq:EGLdefn}
	\curly{E}^{GL}(\psi_1,\ldots,\psi_n) = \sum_{i,j,k,m} c_{ijkm} \ol{\psi_i} \ol{\psi_j} \psi_k \psi_m - \sum_{i,j} d_{ij} \ol{\psi_i} \psi_j.
\eeq
Here the ``GL coefficients'' $c_{ijkm}, d_{ij}$ are given by
\begin{align}
\label{eq:cdefn}
		c_{ijkm} &= \frac{1}{T_c^2} \int_{\set{R}^D} \frac{g_1((p^2-\mu)/T_c)}{p^2-\mu} \l|K_{T_c}(p)\r|^4 \ol{\ft{a}_i(\p)\ft{a}_j (\p)}\ft{a}_k(\p) \ft{a}_m(\p)\dpp\\
\label{eq:ddefn}
		d_{ij} &= \frac{1}{2T_c} \int_{\set{R}^D}  \frac{1}{\cosh^2\l(\frac{p^2-\mu}{2T_c}\r)} \l|K_{T_c}(p)\r|^2 \ol{\ft{a}_i(\p)}\ft{a}_j(\p)\dpp.
\end{align}
The minimum in \eqref{eq:thmmainTI} turns into the minimum over all $\psi_1,\ldots,\psi_n \in \set{C}$. In part II, we compute the integrals \eqref{eq:cdefn},\eqref{eq:ddefn} for special symmetry types and study the resulting minimization problem given by \eqref{eq:EGLdefn}.

\item If one assumes $n=1$, this result is a corollary of Theorem 1 in \cite{FrankHainzlSeiringerSolovej12}, which is obtained by restricting it to translation-invariant systems. (When comparing, note that \cite{FrankHainzlSeiringerSolovej12} rescale the BCS free energy to macroscopic units.) In this case, the microscopically derived GL theory is simply of the form \eqref{eq:mexican}.
\item Note that the error term in \eqref{eq:thmmainTI} is $O(T_c-T)$ higher than the order at which the GL energy enters. Such an error bound is probably optimal because the semiclassical expansion of Lemma \ref{lm:SC} will contribute terms at this order. It improves on the error term that one would obtain from Theorem 1 of \cite{FrankHainzlSeiringerSolovej12} in the case $n=1$. 
\e{enumerate}
   \e{rmk}
   
We note that writing $\min \curly{E}^{GL}$ in the above theorem is justified because

\be{prop}
\label{prop:GL}
The microscopically derived Ginzburg--Landau energy satisfies $\inf_{\set{C}^n} \curly{E}^{GL}>-\it$. Moreover, the infimum is attained. 
\e{prop}

When $T\geq T_c$, it was proved in \cite{HainzlHamzaSeiringerSolovej08} that the unique minimizer of $\Gam\mapsto\curly{F}^{BCS}_T(\Gam)$ is the normal state $\Gam_0$. In other words, the left-hand side in \eqref{eq:thmmainTI} vanishes identically for all $T\geq T_c$. Nonetheless, one can still ask if GL theory describes \emph{approximate} minimizers of the BCS free energy similarly to Theorem \ref{thm:mainTI} (ii) when $T-T_c$ is positive but small. Indeed, \emph{above} $T_c$ approximate minimizers must have \emph{small} GL order parameters (as one would expect):

\be{prop}
\label{prop:new}
Suppose $T>T_c$ and $\Gam$ satisfies 
\beqs
\curly{F}^{BCS}_{T}(\Gam)-\curly{F}^{BCS}_{T}(\Gam_0) =\eps \l(\frac{T-T_c}{T_c}\r)^2, 
\eeqs   
with $0<\eps\leq M$. Let $\{a_j\}$ be any choice of basis for $\ker(K_{T_c}+V)$. 

Then, there exist $\psi_1,\ldots,\psi_n\in\set{C}^n$ and $\xi\in L^2(\set{R}^D)$ such that $$\ft\Gam_{12}\equiv \ft{\al} = \sqrt{\frac{T-T_c}{T_c}} \sum_{j=1}^n \psi_j \ft{a}_j + \xi$$ with  $\|\xi\|_2= O_M\l(T-T_c\r)$ and
   \beq
   \label{eq:psiestimate}
\sum_{i=1}^n |\psi_i|^2\leq \frac{\eps}{\lam_{\mathrm{min}}} + O_M\l(T_c-T\r)
   \eeq
as $T\rightarrow T_c$. Here $\lam_{\mathrm{min}}>0$ is a system-dependent parameter. 
\e{prop}

%The GL coefficients in \eqref{eq:cdefn} and \eqref{eq:ddefn} can be related by Cauchy-Schwarz, but we will not use this in the following. %such as $c_{ijkm}\leq \sqrt{c_{ijij} c_{kmkm}}$.

\section{Part II: Examples with $d$-wave order parameters}

\subsection{Angular momentum sectors}
\label{sect:dwave}
In order to explicitly compute the GL coefficients given by formulae \eqref{eq:cdefn}, \eqref{eq:ddefn}, we make
 some assumptions on the potential $V$. First and foremost, we assume that \emph{$V$ is radially symmetric}. We can then decompose $L^2(\set{R}^3)$ into angular momentum sectors. We review here some basic facts about these and establish notation. For the \emph{spherical harmonics}, we use the definition
\beq
\label{eq:SH}
  Y_{l}^m(\vartheta,\vp) = \sqrt{\frac{(2l+1)}{4\pi}\frac{(l-m)!}{(l+m)!}}\ P_{l}^m(\cos\vartheta) e^{im\vp},
\eeq
where $P_l^m$ is the associated Legendre function, which we define with a factor of $(-1)^m$ relative to the Legendre polynomial $P_m$. While we will use the $Y_l^m$ in the proofs, it will be convenient to state the results in the basis of \emph{real-valued spherical harmonics} defined by
\beq
\label{eq:realSH}
 Y_{l,m}=
 \be{dcases}
    \frac{i}{\sqrt{2}} \l(Y_l^m-(-1)^{m}Y_l^{-m}\r),\quad &\text{ if } m<0\\
    Y_0^0,\quad &\text{ if } m=0\\
    \frac{1}{\sqrt{2}} \l(Y_l^m+(-1)^{m}Y_l^{-m}\r),\quad &\text{ if } m>0.
 \e{dcases}
\eeq
We let $\curly{S}_l=\spa\{Y_{l}^m\}_{m=-l,\ldots,l}=\spa\{Y_{l,m}\}_{m=-l,\ldots,l}$ and define
\beq
\label{eq:Hldefn}
\curly{H}_l=L^2(\set{R}_+;r^{2}\d r)\otimes \curly{S}_l,\qquad (r\equiv |\x|).
\eeq

We employ the usual physics terminology
\beq
\label{eq:wavedefn}
    \curly{H}_0 \equiv \{\text{$s$-waves}\},\quad 
    \curly{H}_1 \equiv \{\text{$p$-waves}\},\quad
    \curly{H}_2 \equiv \{\text{$d$-waves}\}.
\eeq
Note that $\curly{H}_0$ is just the set of spherically symmetric functions and $Y_{2,2}\propto \frac{x^2-y^2}{x^2+y^2}$ is the $d_{x^2-y^2}$-wave in this classification. In analogy to Fourier series, we have the orthogonal decomposition \cite{SteinWeiss}
\beq
\label{eq:L2decomposition}
    L^2(\set{R}^3)=\bigoplus_{l=0}^\it \curly{H}_l.
\eeq

Recall that $r\equiv |\x|$. The Laplacian in $3$-dimensional polar coordinates reads
\beq
\label{eq:radialLaplacian}
	\nabla^2 = \nabla^2_{\text{rad}} + \frac{\nabla^2_{S^{2}}}{r^2},
\eeq
where $\nabla^2_{\text{rad}}=r^{-2}\del_{r}(r^{2}\del_r)$ and $\nabla^2_{S^{2}}$ is the Laplace-Beltrami operator, which acts on spherical harmonics by
\beq
\label{eq:LBeltrami}
-\nabla^2_{\set{S}^2} Y_{l,m}= l(l+1) Y_{l,m}.
\eeq
Since $K_T$ commutes with the Laplacian and $V$ clearly leaves the decomposition \eqref{eq:L2decomposition} invariant, we observe that the eigenstates of $K_T+V$ can be labeled by $l$ (in physics terminology, $l$ is a ``good quantum number''). To make contact with unconventional superconductors, we will suppose we are in either of the two cases:
\be{itemize}
	\item  $\ker(K_{T_c}+V)=\spa\{\rho_2\}\otimes \curly{S}_2$,\quad \emph{``pure $d$-wave case"}
	\item $\ker(K_{T_c}+V)=\spa\{\rho_0\}\otimes\curly{S}_0+ \spa\{\rho_2\}\otimes\curly{S}_2$,\quad \emph{``mixed $(s+d)$-wave case"}.
\e{itemize}
Here $\rho_0,\rho_2\in L^2(\set{R}_+;r^2\d r)$ are radial functions. They are determined as the ground states of an appropriate $l$-dependent operator acting on radial functions. We assume that these radial ground states are non-degenerate for simplicity. This assumption is satisfied for the examples we give in part III, but may not be satisfied in general.

\subsection{Results}
\subsubsection{The pure $d$-wave case in three dimensions}

\be{thm}[Pure $d$-wave case, 3D]
\label{thm:pureTI}
Let $D=3$. Let $V$ be such that Theorem \ref{thm:mainTI} applies and such that $\ker(K_{T_c}+V)=\spa\{\rho_2\}\otimes \curly{S}_2$ for some $0\not\equiv\rho_2\in L^2(\set{R}_+;r^2\d r)$. Let $\{a_{2,m}\}_{m=-2,\ldots,2}$ be an orthonormal basis of the kernel such that 
\beq
\ft{a}_{2,m}(\p) = \vr(p) Y_{2,m}(\vartheta,\vp)
\eeq
 for an appropriate $\vr\in L^2(\set{R}_+;p^2\d p)$ (explicitly, $\vr$ is the Fourier-Bessel transform \eqref{eq:FB2} of $\vr$).
Let $\psi_{m}$ denote the GL order parameter corresponding to $\ft{a}_{2,m}$ for $-2\leq m\leq 2$. Then:

\be{enumerate}[label=(\roman*)]
\item The GL energy that arises from BCS theory as described in Theorem \ref{thm:mainTI} reads
\beq
\label{eq:pureGL}
\curly{E}_{\text{$d$-wave}}^{GL}(\psi_{-2},\ldots,\psi_2)
=\frac{5c}{14\pi}\l(\l(\sum_{m=-2}^2|\psi_{m}|^2-\tau\r)^2-\tau^2+ \frac{1}{2}	\l|\sum_{m=-2}^2 \psi_{m}^2 \r|^2	\r).
\eeq
where $\tau:= \frac{7\pi d}{5c}$ and
\beq
\begin{aligned}
\label{eq:cdefnn}
 c &= \int_0^\it 	f_4(p)\d p,\qquad\, d &= \int_0^\it  f_2(p)\d p.
\end{aligned}
\eeq
Here, we introduced the positive and radially symmetric functions 
\beq
\label{eq:f24defn}
\begin{aligned}
			f_4(p)=&\frac{p^2}{T_c^2} \frac{g_1\l(\frac{p^2-\mu}{T_c}\r)}{p^2-\mu} |K_{T_c}(p)\vr(p)|^4\\
			f_2(p)=&\frac{p^2}{2T_c} \frac{1}{\cosh^2\l(\frac{p^2-\mu}{2T_c}\r)} |K_{T_c}(p)\vr(p)|^2.
\end{aligned}
\eeq
See \eqref{eq:g0defn} for the definition of $g_1$ and $K_T(p)$.

\item We have $\min\curly{E}_{\text{$d$-wave}}^{GL}=-\frac{5c}{14\pi} \tau^2$. The set of minimizers is
 \beq
 \label{eq:minimizerconditionsTI}
    \curly{M}_{\text{$d$-wave}}=\l\{(\psi_{-2},\ldots,\psi_{2})\in\set{C}^5: \sum_{m=-2}^2 |\psi_{m}|^2 = \tau \text{ and } \sum_{m=-2}^2 \psi_{m}^2 =0\r\}.
 \eeq
 \e{enumerate}
\e{thm}

\be{rmk}
\be{enumerate}[label=(\roman*)]
\item The existence of $V$ such that the assumption on $\ker(K_{T_c}+V)$ holds for an open interval of parameter values follows from statement (i) of Theorem \ref{thm:appendix} by choosing $l_0=2$.
\item Observe that the minimization problem in \eqref{eq:pureGL} is trivial, i.e.\ (ii) is immediate.
\item Recall that we normalized the GL order parameters such that they are related to the Cooper pair wave function via \eqref{eq:aldecomposition}. For the special case \eqref{eq:minimizerconditionsTI}, we see that a minimizing vector will have absolute value $\sqrt{\tau}$. We can then reduce to the case where a minimizing vector lies on the \emph{unit} sphere by rescaling the order parameters. The advantage of this other normalization is that it allows to interpret the absolute value of the order parameters as \emph{relative} densities of superconducting charge carriers.
\e{enumerate}
\e{rmk}

We discuss what symmetry of $\curly{E}^{GL}$ one can expect. First of all, GL theory always has the global $U(1)$ gauge symmetry $\psi_j\mapsto e^{i\phi}\psi_j$ (this is due to the presence of the absolute value signs in \eqref{eq:EGLdefn}). Second, $SO(3)$ acts on spherical harmonics by pre-composition, i.e.\ for $g\in SO(3)$ and $\om\in \set{S}^2$,
\beqs
    g Y_{l,m}(\om):= Y_{l,m}(g^{-1}\om)=\sum_{m'} A^g_{m m'} Y_{l,m'}
\eeqs
where $A^g\in O(2l+1)$ is the analogue of the well-known Wigner $d$-matrix for real spherical harmonics \cite{Aubert13}. By changing the angular integration variable in \eqref{eq:cdefn} and \eqref{eq:ddefn} from $g\om$ to $\om$, it is easy to see that
  \beqs
    \curly{E}^{GL}((A^g)^{-1}\vec{\psi})=\curly{E}^{GL}(\vec{\psi}),
  \eeqs
  where we introduced $\vec{\psi}=(\psi_{-2},\ldots,\psi_2)$. Since $Y_{l,m}$ is reflection-symmetric for even $l$, we can extend the action to all of $O(3)$ and retain the invariance of $\curly{E}^{GL}$. This shows that we can expect $\curly{E}^{GL}$ to have symmetry groups $U(1)$ and $O(3)$. However:  

\be{cor}
    \label{cor:symmgroup}
    In the situation of Theorem \ref{thm:pureTI}:
\be{enumerate}[label=(\roman*)]
\item  For all $\phi\in[0,2\pi)$, $\curly{R}\in O(5)$ and $\vec{\psi}\in \set{C}^5$, 
  \beq
    \curly{E}^{GL}(e^{i\phi}\curly{R}\vec{\psi})=\curly{E}^{GL}(\vec{\psi})
  \eeq
Moreover, $O(5)$ acts transitively and faithfully on $\curly{M}_{\text{$d$-wave}}$.
\item $\curly{M}_{\text{$d$-wave}}$ is a $7$-dimensional manifold in $\set{R}^{10}$. 
\item Any minimizer of $\curly{E}^{GL}_{\text{$d$-wave}}$ has at least \emph{two} non-zero entries $\psi_j$.
    \e{enumerate}
\e{cor}

\be{rmk}
\be{enumerate}[label=(\roman*)]
\item Surprisingly, the emergent symmetry group $O(5)$ is considerably larger than the $O(3)$-symmetry discussed above. (Recall also that $A^g$ from above is in $O(5)$, so that the $O(3)$-symmetry is really contained in the $O(5)$-symmetry.) The particularly nice form of the $O(5)$ action is a consequence of choosing the real-valued spherical harmonics as a basis.

\item  We interpret faithfulness of the group action as saying that $\curly{M}_{\text{$d$-wave}}$ is ``truly'' invariant under the full $O(5)$. 
\item Transitivity means that the set of minimizers $\curly{M}_{\text{$d$-wave}}$ is a single orbit under the $O(5)$ symmetry. In other words, there exists a \emph{unique minimizer modulo symmetry}.
\item We interpret (iii) as a proof of non-trivial coupling between the \emph{real-valued} $d$-wave channels (it is of course a basis-dependent statement).
\e{enumerate}
\e{rmk} 

\be{proof} The invariance under multiplication by $e^{i\phi}$ is trivial. To see the $O(5)$ symmetry, we use real coordinates because they also provide an interesting change in perspective.
Writing $\vec{\psi}=\vec{x}+i\vec{y}$ with $\vec{x},\vec{y}\in\set{R}^5$, the GL energy becomes
\beq
\begin{aligned}
\curly{E}^{GL}(\vec{x}+i\vec{y})= \frac{5c}{14\pi}\l(\l(\vec{x}^2+\vec{y}^2-\tau\r)^2-\tau^2+ \frac{1}{2}	\l|\vec{x}^2-\vec{y}^2\r|^2+\l|\vec{x}\cdot \vec{y}\r|^2		\r).
\end{aligned}
\eeq
This is clearly invariant under the $O(5)$-action $\vec{x}+i\vec{y}\mapsto \curly{R}\vec{x}+i\curly{R}\vec{y}$. We can rewrite the set of minimizers as
\beq
\label{eq:linindep}
\curly{M}_{\text{$d$-wave}}=\l\{	(\vec{x},\vec{y})\in \set{R}^5\times \set{R}^5\, : \,	\vec{x}^2 =\vec{y}^2 = \frac{\tau}{2},\,	\vec{x}\cdot \vec{y} = 0	\r\}.
\eeq
Without loss of generality, we may set $\tau/2=1$, so that $\curly{M}_{\text{$d$-wave}}$ is just the set of pairs of orthonormal $\set{R}^5$-vectors. To see that the $O(5)$-action is transitive, consider the orbit of $(e_1,e_2)\in \curly{M}_{\text{$d$-wave}}$, namely $\setof{(\curly{R}e_1,\curly{R}e_2)}{\curly{R}\in O(5)}$. Since any two orthonormal vectors can appear as the first two columns of an orthogonal matrix, we have transitivity. To see that the action is faithful, note that for any two distinct $\curly{R},\tilde{\curly{R}}\in O(5)$, there exists $e_i$ such that $\curly{R}e_i\neq \tilde{\curly{R}}e_i$. 

For (ii), we employ the implicit function theorem and observe that the Jacobian associated with the functions $\vec{x}^2,\vec{y}^2,\vec{x}\cdot\vec{y}$ from \eqref{eq:linindep} has rank $3$. Finally, (iii) 
is immediate from \eqref{eq:minimizerconditionsTI}.
 \e{proof}

\subsubsection{The pure $d$-wave case in two dimensions}
Note that the two-dimen\-sional analogue of the space $\curly{S}_l$, namely the homogeneous polynomials of order $l$ on $\set{S}^1$, is spanned by $\cos(l\vp)$ and $\sin(l\vp)$. Thus assumption \eqref{eq:dwave2Dass} below is the two-dimensional analogue of the assumption $\ker(K_{T_c}+V)=\spa \{\rho_2\}\otimes\curly{S}_2$ in Theorem \ref{thm:mainTI} above.

\be{thm}[Pure $d$-wave case, 2D]
\label{thm:pureTI2D}
Let $D=2$. Let $V$ be such that Theorem \ref{thm:mainTI} applies and such that $\ker(K_{T_c}+V)=\spa\{a_{xy},a_{x^2-y^2}\}$ with
\beq
\label{eq:dwave2Dass}
\ft a_{x^2-y^2}(\p)=\vr(p)\frac{\cos(2\vp)}{\sqrt{\pi}},\qquad \ft a_{xy}(\p)=\vr(p)\frac{\sin(2\vp)}{\sqrt{\pi}},
\eeq
for an appropriate, normalized $0\not\equiv\vr\in L^2(\set{R}_+,p\d p)$.  
 Let $\psi_{x^2-y^2}$ and $\psi_{xy}$ denote the corresponding GL order parameters. Then:

\be{enumerate}[label=(\roman*)]
\item The GL energy that arises from BCS theory as described in Theorem \ref{thm:mainTI} reads
\beq
\begin{aligned}		
\label{eq:pureGL2D}
&\curly{E}_{\text{$d$-wave},2D}^{GL}(\psi_{x^2-y^2},\psi_{xy})\\
&=\frac{c}{2\pi} \l\{	 \l(|\psi_{x^2-y^2}|^2+|\psi_{xy}|^2-\frac{\pi d}{c}\r)^2-\frac{\pi^2 d^2}{c^2} + \frac12 \l|\psi_{x^2-y^2}^2+\psi_{xy}^2\r|^2	\r\}
\end{aligned}	
\eeq
where $c,d$ are defined in the same way as in Theorem \ref{thm:pureTI} with $f_2(p),f_4(p)$ replaced by $f_2(p)/p,f_4(p)/p$.

\item We have $\min\curly{E}_{\text{$d$-wave},2D}^{GL}=-\frac{\pi d^2}{2c}$. The set of minimizers is
 \beq
 \begin{aligned}
 \label{eq:minimizerconditionsTI2D}
    &\curly{M}_{\text{$d$-wave},2D}\\
    &=\l\{(\psi_{x^2-y^2},\psi_{xy})\in\set{C}^2: |\psi_{x^2-y^2}|^2+|\psi_{xy}|^2 
    =\frac{\pi d}{c},\; \psi_{x^2-y^2}^2+\psi_{xy}^2=0\r\}\\
    &=\l\{(\psi,\pm i \psi) \in\set{C}^2: |\psi|^2 =\frac{\pi d}{2c} \r\}
\end{aligned} 
 \eeq
 \e{enumerate}
\e{thm}

\be{rmk}
\be{enumerate}[label=(\roman*)]
\item 
Statement (i) directly implies the first equality in \eqref{eq:minimizerconditionsTI2D} and the second equality is elementary. Note that the result can be conveniently stated in terms of the complex-valued spherical harmonics as well.

\item  From the second equation in \eqref{eq:minimizerconditionsTI2D}, we see that the minimizers of the GL2 for a pure $d$-wave superconductor in two dimensions (in the cosine, sine basis) form a double cover of the minimizers of the usual ``mexican-hat'' GL1. 

\item A similar result holds for any pure angular momentum sector in two dimensions.
\e{enumerate}
\e{rmk}

%%%%%%%%%%%%%%%%%%%%%%%%%%%%%%%%%%%%%%%%%%%%%%%%%%%%%%%%%%%%

\subsubsection{The mixed $(s+d)$-wave case}
We write $\Re[z]$ for the real part of a complex number $z$.

\be{thm}[Mixed $(s+d)$-wave case, 3D]
\label{thm:mixedTI}
Let $D=3$. Let $V$ be such that Theorem \ref{thm:mainTI} applies and such that $\ker(K_{T_c}+V)=\spa\{\rho_0\}\otimes\curly{S}_0+\spa\{\vr_2\}\otimes \curly{S}_2$
for some $0\not\equiv\rho_0,\rho_2\in L^2(\set{R}_+;r^2\d r)$. As an orthonormal basis, take $a_{2,m}$ as in Theorem \ref{thm:pureTI} and $a_s$ with 
\beq
\ft{a}_{s}(\p) = \vr_s(p)\, Y_{0,0}(\vartheta,\vp).
\eeq
 Let $\psi_m,\, (m=-2,\ldots,2)$ and $\psi_s$ denote the GL order parameters corresponding to the respective basis functions. Then:

\be{enumerate}[label=(\roman*)]
\item The microscopically derived GL energy reads
\beq
\begin{aligned}
\label{eq:mixedGL}
    &\curly{E}_{\text{$(s+d)$-wave}}^{GL}(\psi_{s},\psi_{-2},\ldots,\psi_2)\\
    &=\curly{E}_{\text{$s$-wave}}^{GL}(\psi_{s})+\curly{E}^{GL}_{\text{$d$-wave}}(\psi_{-2},\ldots,\psi_{2})+\curly{E}^{GL}_{\text{coupling}}(\psi_{s},\psi_{-2},\ldots,\psi_2)\\
    \end{aligned}
\eeq
where $\curly{E}^{GL}_{\text{$d$-wave}}(\psi_{-2},\ldots,\psi_2)$ is given by \eqref{eq:pureGL}, 
\beq
\label{eq:pureswave}
\curly{E}^{GL}_{\text{$s$-wave}}(\psi_s)=\frac{c^{(4s)}}{4\pi}\l(\l(	|\psi_{s}|^2 -	\tau_s\right)^2-\tau_s^2\r),
\eeq
with $\tau_s=\frac{2 \pi d^{(2s)}}{c^{(4s)}}$, and
\beq
\label{eq:coupling}
\begin{aligned}
    &\curly{E}^{GL}_{\text{coupling}}(\psi_{s},\psi_{-2},\ldots,\psi_2)\\
    &=\frac{c^{(2s)}}{2\pi} \l( 2|\psi_s|^2\sum_{m=-2}^2 |\psi_{m}|^2+\Re\l[ \ol{\psi_s}^2\l(\sum_{m=-2}^2 \psi_{m}^2\r) \r] \r)\\
    &\quad+\frac{\sqrt{5}c^{(s)}}{7\pi} \l(   \Re\l[\ol{\psi_s} \l(2\psi_0|\psi_{0}|^2 + \sum_{m=\pm 1,2} |m|(-1)^{m+1} (2\psi_0|\psi_{m}|^2+\ol{\psi_0}\psi_{m}^2) \r)\r] \vphantom{\Re\l[\ol{\psi_s}    \l(\ol{\psi_{2}} \sum_{m=\pm 1} m \l(2|\psi_{m}|^2+\psi_{m}^2\r)\r) + 2\ol{\psi_{-2}}
    \l(\psi_{1}\psi_{-1}+2\Re[\ol{\psi_{1}}\psi_{-1}] \r)\r]} \r.\\
    &\quad\qquad\qquad\quad +\sqrt{3}\Re\l[\ol{\psi_s}
     \sum_{m=\pm 1} m \l(2\psi_{2}|\psi_{m}|^2+\ol{\psi_{2}}\psi_{m}^2\r)\r]\\
    &\quad\qquad\qquad\quad  \l.+ 2\sqrt{3}\Re\l[\ol{\psi_s}
    \l(\ol{\psi_{-2}}\psi_{1}\psi_{-1}+2\psi_{-2}\Re\l[\ol{\psi_{1}}\psi_{-1}\r] \r)\r] \vphantom{\Re\l[\ol{\psi_s\psi_{2}}
     \sum_{m=\pm 1} m \l(2|\psi_{m}|^2+\psi_{m}^2\r)\r]} \r).
    \end{aligned}
\eeq
 The coefficients $c,d$ are given by \eqref{eq:cdefnn}. Moreover, for $m=1,2,4$, we introduced
\beq
\label{eq:csdefn}
    c^{(ms)} = \int_0^\it  f_4(p) g_s(p)^m  \,\d p,\qquad
  d^{(2s)} = \int_0^\it  f_2(p) g_s(p)^2  \,\d p,
\eeq
with $f_2,f_4$ as in \eqref{eq:f24defn} and
\beq
\label{eq:gsdefn}
	g_s(p)=\l|\frac{\varrho_s(p)}{\varrho(p)}\r|.
\eeq

\item The following are equivalent:
\be{itemize}
    \item $d c^{(2s)} < \frac{5}{7}  cd^{(2s)}$,
\item for all sufficiently small $\eps>0$, and for any minimizer $(\psi_{-2},\ldots,\psi_2)$ of $\curly{E}^{GL}_{\text{$d$-wave}}$,
there exists $\psi_s$ with $|\psi_s|=\eps$ such that
\beq
\label{eq:thmmixedTI}
    \curly{E}^{GL}_{\text{$(s+d)$-wave}}(\psi_s,\psi_{-2},\ldots,\psi_2) < \curly{E}^{GL}_{\text{$d$-wave}}(\psi_{-2},\ldots,\psi_2)=\min \curly{E}^{GL}_{\text{$d$-wave}}.
\eeq
\e{itemize}

\item The following are equivalent:
\be{itemize}
    \item $d^{(2s)} c^{(2s)} \leq d c^{(4s)}$,
    \item for all sufficiently small $\eps>0$, and for any minimizer $\psi_{s}$ of $\curly{E}^{GL}_{\text{$s$-wave}}$,
    there exists $(\psi_{-2},\ldots,\psi_2)$
 with $|\psi_m|<\eps$ for $m=-2,\ldots,2$ such that
\beq
\label{eq:thmmixedTI2}
    \curly{E}^{GL}_{\text{$(s+d)$-wave}}(\psi_s,\psi_{-2},\ldots,\psi_2) < \curly{E}^{GL}_{\text{$s$-wave}}(\psi_{s})=\min \curly{E}^{GL}_{\text{$s$-wave}}.
\eeq
\e{itemize}
\e{enumerate}
\e{thm}

We see that $\curly{E}_{\text{$(s+d)$-wave}}^{GL}$ yields a much richer GL theory than $\curly{E}^{GL}_{\text{$d$-wave}}$. Especially the terms which depend on the relative phases of several GL order parameters make this a rather challenging minimization problem. Accordingly, we no longer have an explicit characterization of the set of minimizers. However, using (ii) and (iii) above, we immediately obtain

\be{cor}[Non-trivial coupling of $s$- and $d$-waves]
\label{cor:mixedcoupling}
In the situation of Theorem \ref{thm:mixedTI} suppose that $d c^{(2s)} < \frac{5}{7}  cd^{(2s)}$ and $d^{(2s)} c^{(2s)} \leq d c^{(4s)}$.
Then any minimizer $(\psi_{s},\psi_{-2},\ldots,\psi_2)$ of $\curly{E}^{GL}_{\text{$(s+d)$-wave}}$ must satisfy $\psi_{s}\neq 0$ and $\psi_m\neq0$ for some $-2\leq m\leq2$.	
\e{cor}

\be{rmk}
\label{rmk:mixed}
\be{enumerate}[label=(\roman*)]
\item The existence of $V$ such that the assumption on $\ker(K_{T_c}+V)$ in Theorem \ref{thm:mixedTI} holds for appropriate parameter values follows from statement (ii) of Theorem \ref{thm:appendix}.

\item Using the same method and the two-dimensional analogues of all quantities above, one can also compute the GL3 that arises for a two-dimen\-sional isotropic $(s+d)$-wave superconductor
\beq
\begin{aligned}
&4\pi \curly{E}^{GL}_{\text{$(s+d)$-wave},2D}(\psi_s,\psi_{x^2-y^2},\psi_{xy})\\
&=	3c|\psi_{x^2-y^2}|^4+3c|\psi_{xy}|^4+2c^{(4s)}|\psi_s|^4+2c\Re[\ol{\psi_{x^2-y^2}}^2 \psi_{xy}^2]\\
&\quad + 4c |\psi_{x^2-y^2}|^2 |\psi_{xy}|^2+ 4 c^{(2s)} \Re[\ol{\psi_s}^2 (\psi_{x^2-y^2}^2+\psi_{xy}^2)]\\
&\quad + 8 c^{(2s)}|\psi_s|^2 (|\psi_{x^2-y^2}|^2+|\psi_{xy}|^2)\\
&\quad	-4\pi d \l(|\psi_{x^2-y^2}|^2+|\psi_{xy}|^2\r) -4\pi d^{(2s)} |\psi_s|^2
\end{aligned}
\eeq
Its complexity lies somewhere between the GL theories in Theorems \ref{thm:pureTI2D} and \ref{thm:mixedTI}. Setting $\psi_{xy}=0$ (that is, we forbid the $d_{xy}$ channel ad hoc), we obtain the GL2 
\beq
\begin{aligned}
&4\pi \curly{E}^{GL}_{\text{$(s+d)$-wave},2D}(\psi_s,\psi_{x^2-y^2},0)\\
&= 3c|\psi_{x^2-y^2}|^4+2c^{(4s)}|\psi_s|^4+ 4 c^{(2s)} \Re[\ol{\psi_s}^2 \psi_{x^2-y^2}^2]\\
&\quad 8 2c^{(2s)} |\psi_s|^2 |\psi_{x^2-y^2}|^2 -4\pi d |\psi_{x^2-y^2}|^2 -4\pi d^{(2s)} |\psi_s|^2
\end{aligned}
\eeq
Compare this with $\curly{E}^{GL}_{\text{$d$-wave},2D}$ from Theorem \ref{thm:pureTI2D}. While one cannot complete the square because
the coefficients differ in a way that depends on the microscopic details, notice that the only phase-dependent term is of the form 
\beq
4 c^{(2s)}\Re[\ol{\psi_{s}}^2\psi_{x^2-y^2}^2]
\eeq
with $c^{(2s)}>0$. It is then clear that for minimizers, the $d_{x^2-y^2}$- and $s$-wave order parameters must have a relative phase of $\pm i$. \e{enumerate}
\e{rmk}

\section{Part III: Radial potentials with ground states of arbitrary angular momentum}
\label{sec:bessel}
In this part, $D=3$ and $\mu>0$. Recall that
\beq
\label{eq:KTdefnrecall} 
K_T(p)=\frac{p^2-\mu}{\tanh\l(\frac{p^2-\mu}{2T}\r)},
\eeq
and the operator $K_T$ is multiplication by the function $K_T(p)$ in Fourier space. Recall the definition \eqref{eq:Vdeltadefns}  of the Dirac delta potentials 
\beqs
V_{\lam,R}(\x)=-\lam\de(|\x|-R),
\eeqs
for $\lam,R>0$.

The following theorem says that, given a non-negative integer $l_0$, we can choose parameter values for $\mu,\lam,R$ from appropriate open intervals such that the zero-energy ground state sector of $K_{T_c}+V_{\lam,R}$ lies entirely within the angular momentum sector $\curly{H}_{l_0}$. 

\be{thm}
\be{enumerate}[label=(\roman*)]
\label{thm:appendix}
\item Let $l_0$ be a non-negative integer. For every $R>0$, there exist an open interval $I\subset \set{R}_+$ and $\lam_*>0$ such that for all $\mu\in I$ and all $\lam\in (0,\lam_*)$ there exists $T_c>0$ such that
\begin{align}
\label{eq:appendix0}
        \inf\textnormal{spec}(K_{T_c}+V_{\lam,R})&=0,\\
\label{eq:appendix1}
    \ker(K_{T_c}+V_{\lam,R}) &=\spa\{\rho_{l_0}\}\otimes \curly{S}_{l_0},\\
\label{eq:appendix2}
    \inf\textnormal{spec}(K_{T}+V)&<0,\quad \forall T<T_c.
\end{align}
Explicitly, the (non-normalized) radial part is
\beq
\label{eq:rholdefn}
\rho_{l_0}(r)= r^{-1/2} \int_0^\it p \frac{\curly{J}_{l_0+\frac{1}{2}}(rp)\curly{J}_{l_0+\frac{1}{2}}(Rp)}{K_{T_c}(p)}\d p.
\eeq
\item For every $R>0$, there exists $T_*>0$ such that for all $T_c<T_*$, there exist $\mu, \lam>0$ such that 
\begin{align}
    \label{eq:appendix3}
    \inf\textnormal{spec}(K_{T_c}+V_{\lam,R})&=0,\\
    \label{eq:appendix4}
	\ker(K_{T_c}+V_{\lam,R}) &= \spa\{\rho_{0}\}\otimes \curly{S}_{0} + \spa\{\rho_{2}\}\otimes \curly{S}_{2},\\
    \label{eq:appendix5}
	\inf\textnormal{spec}(K_{T}+V)&<0,\quad \forall T<T_c.
\end{align}
with $\rho_0,\rho_2$ as in \eqref{eq:rholdefn}.
\e{enumerate}
\e{thm}

\be{rmk}
\label{rmk:appendix}
\be{enumerate}[label=(\roman*)]
\item To be completely precise, in (i) there exists $T_0$ such that the analogue of \eqref{eq:appendix0}-\eqref{eq:appendix2} holds with $T_0$ in place of $T_c$. Then $T_c=T_0$ by definition \eqref{eq:Tcdefn}.
\item The parameter $R$ can be removed by rescaling $\mu,\lam$ and $T$ appropriately.
\item In statement (i), for given $\mu\in I$, $\lam\in(0,\lam_*)$ and $R>0$, $T_c$ is given as the unique solution to the implicit relation
\beq
1=\lam\int_0^\it \frac{pR}{K_{T_c}(p)} \curly{J}_{l_0+1/2}^2(pR) \d p.
\eeq
\item The fact that statement (i) holds for open intervals of $\mu$ and $\lam$ values is to be interpreted as saying that the occurrence of degenerate ground states for $K_{T_c}+V_{\lam,R}$ is ``not un-generic''. This may be surprising at first sight, because if one replaces $K_T+V$ by the Schr\"odinger operator $-\nabla^2+V$, the Perron-Frobenius Theorem (see e.g.\ \cite{LiebLoss}) implies that the ground state is always \emph{simple}. %The fact that the kinetic operator $K_T$ allows for a larger family of ground states than $-\nabla^2$ can be intuitively understood from the observation that its symbol \eqref{eq:KTdefnrecall} vanishes on a codimension-one set (the Fermi sphere $\{p^2=\mu\}$), whereas $p^2$ the symbol of $-\nabla^2$ vanishes only at the origin. 
\item The proof critically uses that $K_T(p)$ is small (for small enough $T$) on the set $\{\p:\p^2=\mu\}$. Note that this set would be empty for $\mu< 0$. 
\item It is interesting to compare Theorem \ref{thm:appendix} with Theorem 2.2 from \cite{FrankHainzlNabokoSeiringer07} which characterizes the critical temperature in the weak-coupling limit $\lam\rightarrow 0$ through an effective Hilbert-Schmidt operator $\curly{V}_\mu$ acting only on $L^2$ of the Fermi sphere. For radial potentials, \cite{FrankHainzlNabokoSeiringer07} shows that $\ker(K_{T_c}+V)\subset\curly{H}_{l_0}$ for all sufficiently small $\lambda$ iff $l_0$ is the unique minimizer of
\beq
\label{eq:FHNS}
  l\mapsto \frac{\sqrt{\mu}}{2\pi^2} \int V(\x) |j_l(\sqrt{\mu} |\x|)|^2 \d \x
\eeq
where $j_l(z)=\sqrt{\frac{\pi}{2z}}J_l(z)$ is the spherical Bessel function of the first kind. While our proof here will be independent of \cite{FrankHainzlNabokoSeiringer07}, one can take $V=V_{\lam,R}$ in \eqref{eq:FHNS} to see that the key fact needed to prove $\ker(K_{T_c}+V)\subset\curly{H}_{l_0}$ is that there is a point at which $j_{l_0}^2>\sup_{l\neq l_0}j_{l}^2$. This is the content of Theorem \ref{thm:bessel}.
 %We will explain why it would not hold if $K_{T}+V$ was defined on $L^2(\rt)$ in Subsection \ref{ssect:limitation}. We will also discuss there how the proof breaks down if one tried to generalize (ii) in the natural way, trying to produce potentials such that $\ker(K_T+V_{\lam,R}) = \curly{H}_{l_0} \cup \curly{H}_{l_1}$ for arbitrary even $l_0\neq l_1$.
\e{enumerate}
\e{rmk}

We conclude by discussing the conceivable extensions of Theorem \ref{thm:appendix}. Statement (i) also holds if $K_{T}+V$ is defined on all of $L^2(\rt)$ instead of just on $L^2_{\mathrm{symm}}(\rt)$, so there is nothing special about even functions in (i).

Statement (ii) can not be generalized as much: (a) it will not hold when odd functions are also considered and (b) it does not generalize to arbitrary pairs $(l_0,l_0+2)$ with $l_0$ even. The reason is that, as demonstrated within the proof of Theorem \ref{thm:appendix}, for small enough $T$, (ii) is equivalent to the existence of a point where $\curly{J}_{1/2}>\curly{J}_{l+1/2}$ for all even $l\geq 1$. The generalizations to more $l$-values described above require the analogous inequalities for Bessel functions. However, these facts will not hold in the cases above, as becomes plausible when considering Figure \ref{fig:bessel}. 

\section{Proofs for part I}
The strategy of the proof follows \cite{FrankHainzlSeiringerSolovej12}. 

We introduce the family of BCS states $\Gam_\Delta$ from which the trial state generating the upper bound will be chosen. The relative entropy identity \eqref{eq:REidentity} rewrites the difference of BCS free energies as terms involving $\Gam_\Delta$. 

The \emph{main simplification} of our proof as compared to \cite{FrankHainzlSeiringerSolovej12} is then in the ``semiclassical'' Theorem \ref{thm:SC}. While \cite{FrankHainzlSeiringerSolovej12} requires elaborate semiclassical analysis for analogous results, the proof in our technically simpler translation-invariant case reduces to an ordinary Taylor expansion.

 Afterwards, we discuss how one concludes Theorem \ref{thm:mainTI} by separately proving an upper and a lower bound. In the lower bound, the degeneracy requires modifying the arguments from \cite{FrankHainzlSeiringerSolovej12} slightly.

%\be{rmk}
%\label{rmk:step3}
%The lower bound, specifically step 3, is the place where the degeneracy requires modifying the proof of the lower bound in \cite{FrankHainzlSeiringerSolovej12} slightly. The key point is that $L^4$ norms do not behave nearly as well for sums of orthonormal functions as $L^2$ norms do. 
%\e{rmk}

\subsection{Relative entropy identity}
All integrals are over $\set{R}^D$ unless specified otherwise. We introduce the family of operators
\beq
\label{eq:Gamdeldefn}
	\ft\Gam_\Del := \frac{1}{1+\exp(\ft H_\Del/T)},\qquad \ft H_\Del:=\ttmatrix{\goth{h}}{\ft\Del}{\ol{\ft\Del}}{-\goth{h}}.
\eeq
Here $\Del$ is an even function on $\set{R}^D$ and we have introduced 
\beq
\label{eq:gothhdefn}
	\goth{h}(p)=p^2-\mu,	
\eeq
the energy of a single unpaired electron of momentum $\p$. Note that the choice $\ft \Delta\equiv0$ in \eqref{eq:Gamdeldefn} indeed yields the normal state $\Gam_0$ defined in \eqref{eq:normalstate}.

Recall that $\Gam$ is a BCS state iff $0\leq \ft\Gam \leq 1$ and $\ft\Gam$ is of the form \eqref{eq:BCSstatedefn}.

\be{prop}
$\Gam_\Del$ defined by \eqref{eq:Gamdeldefn} is a BCS state and
\beqs
 \ft\Gam_\Del(\p) = \ttmatrix{\ft{\gam}_\Del(\p)}{\ft{\al}_\Del(\p)}{\ol{\ft{\al}_\Del(\p)}}{1-\ft{\gam}_\Del(\p)}
\eeqs
with
\begin{align}
 			\label{eq:gamdelta}
			\ft{\gam}_\Delta(\p)&=\frac{1}{2} \left(1-(p^2-\mu)	\frac{\tanh(E_\Del(\p)/(2T))}{E_\Del(\p)}\right), \\
			\label{eq:aldelta}
 \		\ft{\al}_\Delta(\p)&=-\frac{\ft\Del(\p)}{2} \frac{\tanh(E_\Del(\p)/(2T))}{E_\Del(\p)},\\
 E_\Del(\p)&=\sqrt{\goth{h}(p)^2+|\Del(\p)|^2}.
 \end{align}
\e{prop}

\be{proof}
It is obvious from \eqref{eq:Gamdeldefn} that $0\leq \ft\Gam_\Del(\p)\leq 1$. Since $(\ft H_\Del)^2=E_\Del^2 I_2$ and since $\tanh(x)/x$ only depends on $x^2$, it follows that
\beqs
\begin{aligned}
		\ft\Gam_\Del &=\frac{1}{1+\exp(\ft H_\Del/T)} = \frac{1}{2} \l(	1- \tanh(\ft H_\Del/(2T))\r)\\
		& = \frac{1}{2} \l(	1- \frac{\ft H_\Del}{E_\Del} \tanh(E_\Del/(2T))\r),
		\end{aligned}
\eeqs
which yields \eqref{eq:gamdelta} and \eqref{eq:aldelta}.
\e{proof}

\noindent We now give an identity which rewrites the difference $\curly{F}_T^{BCS}(\Gam)-\curly{F}_T^{BCS}(\Gam_0)$ in terms of more manageable quantities involving $\Gam_\Del$, one of them is the relative entropy. 

\be{prop}[Relative Entropy Identity, \cite{FrankHainzlSeiringerSolovej12}]
\label{prop:REidentity}
Let $\Gam$ be an admissible BCS state and $\goth{a}\in H^1_{sym}(\set{R}^D)$. Set $\ft \Del= 2\ft{ V \goth{a}}$. It holds that
\beq
\begin{aligned}
		&\curly{F}^{BCS}_{T}(\Gam)-\curly{F}^{BCS}_{T}(\Gam_0)\\
		&= -\frac{T}{2} \Tr{\log\left( 1+ e^{-\ft H_\Delta/T}\right) 
		- \log\left( 1+ e^{-\ft H_0/T}\right) }\\
		\label{eq:REidentity}
		 &\quad + \frac{T}{2}\goth{H}(\Gam,\Gam_\Del) -  \int V |\goth{a}|^2 \dx + \int V |\al-\goth{a}|^2 \dx
		\end{aligned}
\eeq
where $\goth{H}(\Gam,\Gam_\Del)$ is the \emph{relative entropy} defined by
\beq
%\begin{aligned}
		\goth{H}(\Gam,\Gam_\Del)	:=\Tr{\phi(\ft\Gam,\ft \Gam_\Del)}.
%\end{aligned}
\eeq
Here we introduced
\beqs
\phi(x,y)=x(\log(x)-\log(y))+(1-x)(\log(1-x)-\log(1-y)), \quad\forall 0\leq x,y\leq 1.
\eeqs
\e{prop}

\be{proof}
This is a computation, see \cite{FrankHainzlSeiringerSolovej12} or \cite{FrankHainzlSeiringerSolovej14}.
\e{proof}

For the sake of comparability with \cite{FrankHainzlSeiringerSolovej12}, we note that in the translation-invariant case the $L^2$-trace per unit volume of a locally trace-class operator (which they denote by $\textnormal{Tr}$) is just the integral of its Fourier transform and so 
\beqs
		\textnormal{Tr}[\Gam]= \int_{\set{R}^D} \Tra{\set{C}^2}{\ft\Gam(\p)} \dpp.
\eeqs

\subsection{``Semiclassical'' expansion}
We prove Theorem \ref{thm:SC} by a Taylor expansion, which is sufficient because of the simplifications introduced by the translation-invariance. The analogous results in \cite{FrankHainzlSeiringerSolovej12} require many more pages of challenging semiclassical analysis. 

\subsubsection{The result and the key lemma}
Recall the definition of $g_1$ in \eqref{eq:g0defn}. The following is the main consequence of the Taylor expansion

\be{thm}
\label{thm:SC}
Let $\ft \Del=2h\ft{ V\goth{a}}$ for some $\goth{a}\in \ker(K_{T_c}+V)$. Define $h>0$ by
\beq
\label{eq:hdefn}
h=\sqrt{\frac{T_c-T}{T_c}}.
\eeq
Then, as $h\rightarrow 0$,
\beq
\label{eq:expansion}
\curly{F}^{BCS}_{T}(\Gam_\Del)-\curly{F}^{BCS}_{T}(\Gam_0)=h^4 E_2+O(h^6),
\eeq
where 
\beq
\label{eq:E2defn}
E_2=\frac{1}{16T_c^2}\int \frac{g_1(\goth{h}(p)/T_c)}{\goth{h}(p)} |t(\p)|^4	\dpp- \frac{1}{8T_c}\int	\frac{1}{\cosh^2(\goth{h}(p)/(2T_c))} |t(\p)|^2	\dpp
\eeq
with $t=2\ft{V\goth{a}}$.
\e{thm}

We emphasize that this is the place where the effective gap operator appears in the analysis. The choice $\goth{a}\in\ker(K_{T_c}+V)$ ensures that there are no $O(h^2)$ terms in the expansion \eqref{eq:expansion}.

The theorem follows from the key 

\be{lm}
\label{lm:SC}
Let $\Gam_\Del$ be given by \eqref{eq:Gamdeldefn} with $\ft\Del(p)=h t(p)$ for a function $t$ satisfying 
\beq
\label{eq:tcases}
t\in L^q(\set{R}^D) \textnormal{ with } 
\be{cases}
 q=\it &\textnormal{if } D=1,\\
 4<q<\it &\textnormal{if } D=2,\\
 q=6 &\textnormal{if } D=3. 
\e{cases}
\eeq
Then, as $h\rightarrow 0$,
\be{enumerate}[label=(\roman*)]
\item 
\beq
\begin{aligned}
\label{eq:semiclassicsTI}
&-\frac{T}{2} \Tr{\log\left( 1+ e^{-\ft H_\Delta/T}\right) 
		- \log\left( 1+ e^{-\ft H_0/T}\right) }\\
		&=h^2 E_1+h^4 E_2 +O(h^6)
\end{aligned}
\eeq
where $E_2$ is defined by \eqref{eq:E2defn} and (see \eqref{eq:g0defn} for $g_0$)
\beq
\label{eq:E1defn}
\begin{aligned}
E_1=-\frac{1}{4T_c} \int g_0(\goth{h}(p)/T_c) |t(\p)|^2\dpp,
\end{aligned}
\eeq

\item 
\beq
\|\al_\Del-\check\phi\|_{H^1} = O(h^3)
\eeq
with $\phi(\p)=-h\frac{t(p)}{2T_c} g_0(\goth{h}(\p)/T_c)$.
\e{enumerate}
\e{lm}

This may be compared to Theorems 2 and 3 in \cite{FrankHainzlSeiringerSolovej12}. 

To conclude Theorem \ref{thm:SC} from the key lemma, we need a regularity result for the translation-invariant operator. 

\be{prop}[]
\label{prop:regularity1}
Let $a\in H^1(\set{R}^D)$ satisfy $(K_{T_c}+V)a=0$. Then, $\ft{a}\in L^\it(\set{R}^D)$. Let $t:=\ft{Va}$ and $\jap{p}:=(1+p^2)^{1/2}$. Then, $\jap{p}^{-1} t\in L^2(\set{R}^D)$ and $t$ satisfies \eqref{eq:tcases}. 
\e{prop}

\be{proof}
Recall Assumption \ref{ass:V} on the potential $V$. When $V\in L^{p_V}(\set{R}^D)$, then the result follows from Proposition 2 in \cite{FrankHainzlSeiringerSolovej12}. For the potentials $V_{\lam,R}$ in $D=3$, the regularity properties can be read off directly from the explicit solution of the eigenvalue problem $(K_{T_c}+V_{\lam,R})a=0$, see \eqref{eq:Fourierrep} in the proof of Lemma \ref{lm:monotone} for its Fourier representation. Indeed, since $Y_{l,m}$ and the Bessel function of the first kind $\curly{J}_{l+1/2}$ are smooth and bounded with $\curly{J}_{l+1/2}(0)=0$ and since $E<2T$, we get $\ft{a}\in L^\it$. Moreover, 
\beqs
t(\p)\propto Y_{l,m}\l(\vt,\vp\r) \frac{\curly{J}_{l+1/2}(p R)}{\sqrt{p}}
\eeqs
and since $\curly{J}_{l+1/2}$ also decays like $p^{-1/2}$ for large $p$-values, the regularity properties of $t$ follow. In $D=1,2$, one can again solve the eigenvalue problem $(K_{T_c}+V_{\lambda,R})a=0$ explicitly and obtains the claimed regularity by similar considerations. The details are left to the reader.
\e{proof}

\be{proof}[Proof of Theorem \ref{thm:SC}]
First, note that $t=2\ft{ V\goth{a}}$ has all the regularity properties needed to apply (i), thanks to Proposition \ref{prop:regularity1}.
We invoke the relative entropy identity \eqref{eq:REidentity} and use Lemma \ref{lm:SC} to find
\beq
\begin{aligned}
\label{eq:semiclassicsTIpf}	
&\curly{F}^{BCS}_{T}(\Gam_\Del)-\curly{F}^{BCS}_{T}(\Gam_0)\\
&=h^2E_1 +h^4 E_2
	 -  h^2 \int V |\goth{a}|^2\d x + \int V |\al_\Del-h \goth{a}|^2\d x+O(h^6).
\end{aligned}
\eeq
Observe that 
\beq
\label{eq:g0KT}
g_0(\goth{h}(p)/T_c) = T_c K_{T_c}^{-1}(p).
\eeq
By Plancherel and the eigenvalue equation $(K_{T_c}+V)\goth{a}=0$, \eqref{eq:semiclassicsTIpf} becomes
\beqs
\curly{F}^{BCS}_{T}(\Gam_\Del)-\curly{F}^{BCS}_{T}(\Gam_0)=h^4 E_2+ \int V |\al_\Del-h \goth{a}|^2\d \x+O(h^6).
\eeqs
Thus, it remains to show
\beq
\label{eq:term31}
		\int V(\x) |\al_\Del(\x)-h \goth{a}(\x)|^2 \dx= O(h^6).
\eeq 
To see this, recall that $V$ is form-bounded with respect to $-\nabla^2$, so it suffices to prove that $\|\alpha_\Delta - h\goth{a}\|_{H^1}=O(h^3)$. Using the eigenvalue equation and \eqref{eq:g0KT},
\beqs
\ft{\goth{a}}(\p)=-K_{T_c}^{-1}(p)\ft{V\goth{a}}(\p)=-\frac{t(\p)}{2T_c} g_0(\goth{h}(p)/T_c)
\eeqs
and so \eqref{eq:term31} follows from Lemma \ref{lm:SC} (ii).
\e{proof}

\subsubsection{Proof of Lemma \ref{lm:SC}}
\dashuline{Proof of (i)}
We have
$$
\log\left( 1+ e^{-\ft H_\Delta(\p)/T} \right) = -\ft H_\Delta(\p)/(2T) + \log\cosh(\ft H_\Delta(\p)/(2T)).
$$
Observe that $\textnormal{Tr}_{\set{C}^2}\left[\ft H_\Delta(\p)\right]=0$, that $x\mapsto\cosh x$ is an even function and that $\ft H_\Delta(\p)^2=E_\Delta(\p)^2 I_2$. We find
\beqs
\begin{aligned}
\textnormal{Tr}_{\set{C}^2}\left[\log\left( 1+ e^{-\ft H_\Delta/T} \right) \right]
&= \textnormal{Tr}_{\set{C}^2}\left[\log\cosh(\ft H_\Delta(p)/(2T)) \right]\\
&= 2 \log\cosh(E_\Delta(p)/(2T)) \,.
\end{aligned}
\eeqs
This and a similar computation for $\Delta=0$ show that
\begin{align}
\nonumber
& -\frac{T}{2} \Tr{\log\left( 1+ e^{-\ft H_\Delta/T}\right) 
		- \log\left( 1+ e^{-\ft H_0/T}\right) } \\
\label{eq:term11}
		&= -T \int \l( \log\cosh(E_\Del/(2T))-\log\cosh(\goth{h}/(2T)) \r)\dpp.
\end{align}
We denote the function in \eqref{eq:term11} by
\beqs
f(h^2):=T(h^2)\l(\log\cosh\l(\frac{E(h^2)}{2T(h^2)}\r)-\log\cosh\l(\frac{E(0)}{2T(h^2)}\r)\r),
\eeqs
where we wrote $E(h^2)$ for $E_\Del$ and $T(h^2)=T_c(1-h^2)$. Note that $E'=|t|^2/(2E)$ and recall the definition \eqref{eq:g0defn} of $g_0$ and $g_1$. By an easy computation
\begin{align*}
	f(0)&=0,\\
	f'(0)&= -g_0(\goth{h}/T_c) \frac{|t|^2}{4T_c},\\
	\frac{1}{2} f''(0)&=	\frac{g_1(\goth{h}/T_c)}{\goth{h}} \frac{|t|^4}{16 T_c^2} - \frac{1}{\cosh^2(\goth{h}/(2T_c))}\frac{|t|^2}{8 T_c}.
\end{align*}
With this, we can expand \eqref{eq:term11} as follows
\begin{align}
\label{eq:term12}
		\frac{T}{2}&\Tr{\log\ft\Gam_\Del-\log\ft\Gam_0}\\
		\nonumber
		 = &-h^2	\int g_0(\goth{h}/T_c) \frac{|t|^2}{4T_c} \dpp\\		 
		&+ h^4 \l(\frac{1}{16 T_c^2}\int \frac{g_1(\goth{h}/T_c)}{\goth{h}} |t|^4 \dpp - \frac{1}{8 T_c}\int \frac{1}{\cosh^2
		\l(\frac{\goth{h}}{2T_c}\r)} |t|^2 \dpp\r)+O(h^6).
\end{align}
It remains to check that the $O(h^6)$ term is indeed finite. Using the Lagrange remainder in Taylor's formula, it suffices to show
\beq
\label{eq:Oh6finite}
\int \sup_{0<\delta<h^2} \frac{1}{3!} |f'''(\delta)| \dpp <\it.
\eeq
We will control this quantity in terms of appropriate integrals over $t$ which are finite by our assumptions on $t$. We introduce the function
\beq
\label{eq:g2defn}
	g_2(z):=g_1'(z) +\frac{2}{z} g_1(z) = \frac{1}{2z} \frac{1}{\cosh^2(z/2)} \tanh(z/2).
\eeq
By a straightforward computation
\begin{align*}
	\frac{1}{3!}f'''(\delta) = \frac{1}{8 T(\de)^3} &\l[	\frac{|t|^6}{12 E(\de)^2} \l(3\frac{g_1(E(\de)/T(\de))}{E(\de)/T(\de)}-g_2(E(\de)/T(\de))\r)\r.\\
	&- \frac{T_c}{2T(\de)} |t|^4 g_2(E(\de)/T(\de))\\
		&\l.  +T_c^2	|t|^2	\l(\frac{1}{\cosh^2\l(\frac{E(\de)}{2T(\de)}\r)}-\l(\frac{E(\de)}{T(\de)}\r)^2 g_2(E(\de)/T(\de))\r)\r].
\end{align*}
Note that, for $h^2$ small enough, $T_c/2\leq T(\delta)\leq 2T_c$ for all $0<\delta<h^2$. Using this and the fact that $\frac{1}{\cosh^2(z)}$ and $g_2(z)$ are monotone decreasing for $z>0$, we can estimate
\begin{align}
\nonumber
		\int &\sup_{0<\delta<h^2} \frac{1}{3!}|f'''(\delta)| \dpp\\
		\leq &C_1 \int |t|^6 \sup_{0<\delta<h^2} E(\de)^{-2} \l|3\frac{g_1(E(\de)/T(\de))}{E(\de)/T(\de)}-g_2(E(\de)/T(\de))\r|\dpp \\
\label{eq:term13}		
		&+ C_2 \int |t|^4 g_2(\goth{h}/(2T_c))\dpp\\
		\nonumber
		&+	C_3 \int |t|^2 	\l( \frac{1}{\cosh^2(\goth{h}/(4T_c))}	+g_2(\goth{h}/(2T_c))\sup_{0<\delta<h^2}	E(\de)^2 \r) \dpp.
\end{align}
Here $C_1,C_2,C_3$ denote constants which depend on $D,T_c$ and may change from line to line in the following. For definiteness, assume $D=3$. The arguments for $D=1,2$ are similar. Since $g_2(z)$ is a bounded function that decays exponentially for large $z$, we can use Cauchy-Schwarz and the fact that $\goth{h}(p)\sim C\jap{p}^{2}$ for large $p$ to conclude
\beqs
C_2 \int  |t|^4 g_2(\goth{h}/(2T_c)) \dpp \leq C_2 \int \l(	|t|^6 + \jap{p}^{-2} |t|^2	\r) \dpp
\eeqs
and the right-hand side is finite by Proposition \ref{prop:regularity1}. Using that $E(\de)^2=\goth{h}^2+\de |t|^2\leq \goth{h}^2+|t|^2$ for small enough $h$, the same argument applies to the $C_3$ term in \eqref{eq:term13}.

The $C_1$ term in \eqref{eq:term13} contains a factor $E(\delta)^{-2}$ which looks troubling because, as $\de\rightarrow 0$, it is of the form $\goth{h}^{-2}$ and thus singular on the sphere $\{\p: \p^2=\mu\}$ if $\mu>0$. For the radial integration, this singularity would not be integrable (and we have not even considered the factor $|t|^6$ yet). However, the singularity is canceled by the factor $3g_1(z)/z-g_2(z)$ with $z=E(\de)/T(\de)$ in \eqref{eq:term13}. To see this, recall the definition \eqref{eq:g0defn} of $g_1$ and \eqref{eq:g2defn} of $g_2$ and observe that $g_1(z)/z$ and $g_2(z)$ are both even functions. Using the power series representation for $\frac{1}{\cosh^2}$ and $\tanh$, it is elementary to check that in the expansion of $3g_1(z)/z-g_2(z)$ the coefficients of order $z^{-2}$ and $z^0$ vanish and so the lowest non-vanishing order is $z^2$. Therefore, the singularity is removed and since $g_1(z)/z$ and $g_2$ are bounded, we get
\beqs
		\sup_{0<\delta<h^2} E(\de)^{-2} \l|3\frac{g_1(E(\de)/T(\de))}{E(\de)/T(\de)}-g_2(E(\de)/T(\de))\r|\leq C<\it.
\eeqs
Since $\int |t|^6 \dpp<\it$ by our assumption on $t$, the $C_1$ term in \eqref{eq:term13} is finite and we have proved \eqref{eq:Oh6finite}.

\dashuline{Proof of (ii)}
From \eqref{eq:aldelta} we have
\beqs
		\ft{\al}_\Del(\p) = -h\frac{t(\p)}{2T} g_0(E_\Del(\p)/T).
\eeqs
Therefore
\beq
\label{eq:term32}
\|\al_\Del-\check \phi\|_{H^1}^2=  h^2 \int \jap{p}^2 |t|^2 |f(h^2)-f(0)|^2 \dpp,
\eeq
where we introduced the function
\beq
f(h^2):= \frac{g_0(E(h^2)/T(h^2))}{2T(h^2)}.
\eeq
Recall that $g_0'=-g_1$. Using this and the fact that for $h^2$ small enough, $T_c/2\leq T(\delta)\leq 2T_c$ for all $0<\delta<h^2$, Taylor's theorem with Lagrange remainder yields
\beqs
\begin{aligned}
&|f(h^2)-f(0)|\\
&\leq C h^2\sup_{0<\delta <h^2} \l(\l| g_0(E(\de)/T(\de))\r| + |g_1(E(\de)/T(\de))|	\l(\frac{|t|^2}{E(\de)}+E(\de)\r)	\r).
\end{aligned}
\eeqs
Note that $g_0(z)$ and $g_1(z)/z$ are monotone decreasing and so
\beqs
\begin{aligned}
|f(h^2)-f(0)|&\leq C h^2 \l|g_0(\goth{h}/(2T_c))\r| + C' \sup_{0<\delta<h^2} \l| \frac{g_1(\goth{h}/(2T_c))}{\goth{h}}\r|	\l(|t|^2+E(\de)^2\r)\\
&\leq 	C h^2\l(\l|g_0(\goth{h}/(2T_c))\r| +	\l|\frac{g_1(\goth{h}/(2T_c))}{\goth{h}}\r|	\l(|t|^2+\goth{h}^2\r)\r)\\
&\leq  C h^2\l(|t|^2 \goth{h}^{-3}+\jap{\goth{h}}^{-1}\r)
\end{aligned}
\eeqs
where in the second step we used that $E(\de)= \goth{h}^2+\de|t|^2\leq \goth{h}^2+|t|^2$ for small enough $h$ and in the third step we used $g_0(z)\leq C\jap{z}^{-1}$ as well as $g_1(z)/z\leq C \jap{z}^{-3}$. Assume $D=3$ for definiteness. We can bound \eqref{eq:term32} as follows
\beqs
	h^2 \int \jap{p}^2 |t|^2 |f(h^2)-f(0)|^2 \dpp
	\leq C h^6  \int \l(|t|^6	\jap{p}^{-10}+|t|^2\jap{p}^{-2}\r)\dpp = C h^6,
\eeqs
where the last equality holds by the assumption on $t$. This proves (ii).
\qed

\subsection{Proof of Theorem \ref{thm:mainTI}}
We follow the strategy in \cite{FrankHainzlSeiringerSolovej12}. That is, we prove theorem Theorem \ref{thm:mainTI} (i) by separately proving an upper and a lower bound on the left-hand side in \eqref{eq:thmmainTI}. The upper bound follows by choosing an appropriate trial state $\Gam_\Del$ and using the semiclassical expansion of the BCS free energy in the form of Theorem \ref{thm:SC}. For the lower bound, we show that the chosen trial states $\Gam_\Del$ indeed describe \emph{any} approximate minimizer $\Gam$ to lowest order in $h$ (this is precisely statement (ii) in Theorem \ref{thm:mainTI}) and conclude by using the semiclassical expansion once again.

\subsubsection{Upper bound}
Recall the definition of $h$ in \eqref{eq:hdefn}. In this section we prove
\beq
\label{eq:upperbound}
		\min_{\Gam}\curly{F}^{BCS}_{T}(\Gam)-\curly{F}^{BCS}_{T}(\Gam_0) \leq h^4\min_{\goth{a}\in \ker(K_{T_c}+V)} \curly{E}^{GL}(\goth{a}) + O(h^6),
\eeq
where $\curly{E}^{GL}$ is given by \eqref{eq:alternative}. 

We get this by using the trial state $\ft\Gam_\Del$, defined by \eqref{eq:Gamdeldefn} with the choice 
\beq
\label{eq:Delform}
\ft\Del=2h\ft{(V\check{\goth{a}})}
\eeq
where $\goth{a}\in \ker(K_{T_c}+V)$ minimizes $\curly{E}^{GL}$ (recall that minimizers exist by Proposition \ref{prop:GL}). Then, \eqref{eq:upperbound} follows from Theorem \ref{thm:SC} and the fact that evaluating the definition \eqref{eq:E2defn} of $E_2$ for the choice
\beqs
t(\p)=\ft{\Del}(\p)/h=-2 K_{T_c}(p) \goth{a}(\p)
\eeqs
produces the definition \eqref{eq:alternative} of $\curly{E}^{GL}(\goth{a})$. 
\qed

\subsubsection{Lower bound: Part A}
Following \cite{FrankHainzlSeiringerSolovej12}, we will prove the lower bound in \eqref{eq:thmmainTI} in conjunction with statement (ii) about approximate minimizers. We consider any BCS state $\Gam$ satisfying 
\beq
\label{eq:approximatemin}
\curly{F}^{BCS}_{T}(\Gam)-\curly{F}^{BCS}_{T}(\Gam_0)\leq O(h^4).
\eeq
Note that we may restrict to such $\Gam$ when minimizing $\curly{F}_T^{BCS}$ thanks to the upper bound \eqref{eq:upperbound} and that \eqref{eq:approximatemin} still includes the approximate minimizers considered in (ii). In \textbf{Part A}, we prove Proposition \ref{prop:xibded}, which says that the off-diagonal element $\al$ of such a $\Gam$ will be close to a minimizer of  $\curly{E}^{GL}$. In \textbf{Part B}, we will use this to get $\curly{F}_T^{BCS}(\Gam)-\curly{F}_T^{BCS}(\Gam_\Del)\geq O(h^6)$ for $\Del$ of the form \eqref{eq:Delform} and hence
\beqs
\curly{F}_T^{BCS}(\Gam)-\curly{F}_T^{BCS}(\Gam_0)\geq \curly{F}_T^{BCS}(\Gam_\Del)-\curly{F}_T^{BCS}(\Gam_0)+O(h^6).
\eeqs
Since we know $\curly{F}_T^{BCS}(\Gam_\Del)-\curly{F}_T^{BCS}(\Gam_0)=h^4\curly{E}^{GL}(\psi_1,\ldots,\psi_n)+O(h^6)$ from Theorem \ref{thm:SC}, this will imply both the lower bound in \eqref{eq:thmmainTI} and statement (ii) about approximate minimizers.

In the remainder of this section, we will prove:
\be{prop}
\label{prop:xibded}
Suppose $\Gam$ satisfies \eqref{eq:approximatemin} and let $P$ denote the orthogonal projection onto $\ker(K_{T_c}+V)$ and let $P^\perp=1-P$. Then, $\|P\al\|_2=O(h)$ and $\|P^\perp\al\|_2=O(h^2)$.
\e{prop}

%We define $\psi_j:=h^{-1} \scp{\al}{a_j}$ and can thus write
%\beqs
%		P\al=h\sum_{j=1}^n \psi_j a_j.
%\eeqs
%Note that orthonormality of the $\{a_j\}_j$ implies $\|P\al\|_2^2=  h^2\sum_j |\psi_j|^2$. Since all norms on $\set{C}^n$ are equivalent, the claim $\|P\al\|_2=O(h)$ is equivalent to $|\psi|_\it=O(1)$ where $|\psi|_\it:=\max_j |\psi_j|$.

This implies statement (ii) in Theorem \ref{thm:mainTI} with $\goth{a}_0\equiv h^{-1} P\al$. The proof of Proposition \ref{prop:xibded} will use the following lemma, which bounds the relative entropy $\goth{H}(\Gam,\Gam_\Del)$ from below in terms of a weighted Hilbert-Schmidt norm. The result without the second ``bonus'' term on the right-hand side first appeared in \cite{HainzlLewinSeiringer08}, the improved version is due to \cite{FrankHainzlSeiringerSolovej12}.

\be{lm}[Lemma 1 in \cite{FrankHainzlSeiringerSolovej12}]
\label{lm:lbonentropy}
For any $0\leq \Gam\leq 1$ and $\Gam^{(H)}=(1+\exp(H))^{-1}$, it holds that
\beq
\begin{aligned}
\label{eq:lbonentropy}
 \goth{H}(\Gam,\Gam^{(H)})\geq &\Tr{\frac{\ft H}{\tanh(\ft H/2)}(\ft\Gam-\ft\Gam^{(H)})^2}\\
  &+ \frac{4}{3}\Tr{(\ft\Gam(1-\ft\Gam)-\ft\Gam^{(H)}(1-\ft\Gam^{(H)}))^2}.
\e{aligned}
\eeq
\e{lm}

\be{proof}
By the identity (5.7) in \cite{FrankHainzlSeiringerSolovej12} and Klein's inequality for $2\times 2$ matrices, \eqref{eq:lbonentropy} even holds pointwise in $\p$.
\e{proof}

Here is a quick outline of the proof of Proposition \ref{prop:xibded}: Following \cite{FrankHainzlSeiringerSolovej12}, we rewrite $\curly{F}_T^{BCS}(\Gam)-\curly{F}_T^{BCS}(\Gam_0)$ by invoking the relative entropy identity \eqref{eq:REidentity}. Then, we bound the right hand side from below by $\scp{\al}{(K_{T}+V)\al}$, which is therefore negative due to \eqref{eq:approximatemin}. Since $K_{T_c}+V\geq 0$ with a spectral gap above zero, this will allow us to conclude that the part of $\al$ lying outside of $\ker(K_{T_c}+V)$ must be small, more precisely that $\|\al-P\al\|_2=O(h^2)$. To get that $\|P\al\|_2$ itself is $O(h)$, we use the second ``bonus'' term on the right-hand side of Lemma \ref{lm:lbonentropy}.

\be{proof}[Proof of Proposition \ref{prop:xibded}]
\dashuline{Step 1:} We first apply the relative entropy identity \eqref{eq:REidentity} with the choice $\goth{a}=0$ to get
\beq
\label{eq:proppf1}
		O(h^4)\geq \curly{F}_T^{BCS}(\Gam)-\curly{F}_T^{BCS}(\Gam_0)=\frac{T}{2} \goth{H}(\Gam,\Gam_0) + \int V |\al|^2 \dx.
\eeq
Next, we use Lemma \ref{lm:lbonentropy}. To evaluate the resulting expression, note that
\beqs
\begin{aligned}
		\frac{\ft H_0}{\tanh(\ft H_0/(2T))}&=K_T I_{2\times 2},\\
		\ft\Gam(1-\ft\Gam)-\ft\Gam_0(1-\ft\Gam_0) &= \l(\ft{\gam}(1-\ft{\gam})-\ft{\gam}_0(1-\ft{\gam}_0)-|\ft{\al}|^2\r) I_{2\times 2},
\end{aligned}
\eeqs
are diagonal matrices. We obtain
\begin{align*}
	\frac{T}{2}\goth{H}(\Gam,\Gam_0)&\geq \int \l( K_T(\cdot)(\ft{\gam}-\ft{\gam}_0)^2 \dpp + K_T(\cdot) |\ft{\al}|^2\r.\\
	&\quad \l.\vphantom{\int}+ \frac{4T}{3} \l(\ft{\gam}(1-\ft{\gam})-\ft{\gam}_0(1-\ft{\gam}_0)-|\ft{\al}|^2\r)^2 \r)\dpp.
\end{align*}
We estimate the first term using $K_T(p)\geq 2T$ and find the lower bound
\beqs
\int \l( K_T(\cdot)|\ft{\al}|^2 +2T  (\ft{\gam}-\ft{\gam}_0)^2+\frac{4T}{3}\l(\ft{\gam}(1-\ft{\gam})-\ft{\gam}_0(1-\ft{\gam}_0)-|\ft{\al}|^2\r)^2 \r) \dpp.
\eeqs
By
\beqs
		(x(1-x)-y(1-y))^2\leq(x-y)^2,\quad \forall 0\leq x,y\leq 1
\eeqs
and the triangle inequality, we get the pointwise estimate
\beqs
 \l(2(\ft{\gam}-\ft{\gam}_0)^2+\frac{4}{3}\l(\ft{\gam}(1-\ft{\gam})-\ft{\gam}_0(1-\ft{\gam}_0)-|\ft{\al}|^2\r)^2\r) \geq \frac{4}{5} |\ft{\al}|^4.
\eeqs
Going back to \eqref{eq:proppf1}, we have shown that
\beq
\label{eq:pf0}
		\frac{4T}{5}\|\ft{\al}\|_4^4 + \scp{\al}{(K_T+V)\al} \leq O(h^4).
\eeq

\dashuline{Step 2:} Next, we replace $K_T$ by $K_{T_c}$ in \eqref{eq:pf0} to make use of the spectral gap of $K_{T_c}+V$. This is an easy version of what is Step 2 of Part A in \cite{FrankHainzlSeiringerSolovej12}, which is more involved because it also removes the dependence on the external fields $A,W$. For us, it suffices to observe that
\beqs
	\frac{\d}{\d T} K_T(p) = \frac{1}{2T^2} \frac{\goth{h}(p)^2}{\sinh^2(\goth{h}(p)/(2T))}
\eeqs
is uniformly bounded in $p$ for all $h$ small enough such that $T>T_c/2$. By the mean-value theorem, $\|K_T-K_{T_c}\|_\it\leq O(h^2)$. Using this on \eqref{eq:pf0}, we find
\beq
\label{eq:gapcon0}
		\frac{4T}{5}\|\ft{\al}\|_4^4 + \scp{\al}{(K_{T_c}+V)\al} \leq O(h^2) \|\al\|_2^2+ O(h^4).
\eeq
Let $\kappa>0$ denote the size of the spectral gap of $K_{T_c}+V$ above energy zero. We write $\al=P\al+P^\perp\al$. Using $(K_{T_c}+V)P\al=0$, we obtain
\beq
\label{eq:gapcon}
	\frac{4T}{5}\|\ft{\al}\|_4^4 + \kappa \|P^\perp\al\|_2^2 \leq O(h^2) \|\al\|_2^2 +O(h^4).
\eeq
For the moment we drop the first term on the left-hand side of \eqref{eq:gapcon} and use orthogonality to get
\beqs
  \|P^\perp\al\|_2^2 \leq O(h^2) (\|P\al\|_2^2+\|P^\perp\al\|_2^2)+O(h^4)
\eeqs
which yields
\beq
\label{eq:xibound}
\|P^\perp\al\|_2^2\leq O(h^2)\|P\al\|_2^2+O(h^4).
\eeq
Thus, both claims will follow, once we show $\|P \al\|_2=O(h)$.

\dashuline{Step 3:} Here the degeneracy requires a slight modification. We now drop the second term on the left-hand side of \eqref{eq:gapcon} to get
\beq
\label{eq:pf1}
	\|\ft{\al}\|_4 \leq O(h^{1/2}) \|\ft{\al}\|_2^{1/2}+O(h).
\eeq
By orthogonality and \eqref{eq:xibound},
\beq
\label{eq:pf2}
	\|\ft{\al}\|_4 \leq O(h^{1/2}) \|\ft{P\al}\|_2^{1/2}+O(h),
\eeq
On the right-hand side of \eqref{eq:pf1} however, the replacement of $\ft \al$ by $\ft{P \al}$ requires more work. By the triangle inequality for $\|\cdot\|_4$ and \eqref{eq:xibound}
\begin{align*}
  \|\ft{\al}\|_4&\geq \|\ft{P\al}\|_4 - \| \ft{P^\perp\al}\|_4 \geq \|\ft{P\al}\|_4 - \|\ft{P^\perp\al}\|_2^{1/2} \|\ft{P^\perp\al}\|_\it^{1/2} \\
	&\geq \|\ft{P\al}\|_4 - O(h^{1/2}) \|\ft{P\al}\|_2^{1/2} \|\ft{P^\perp\al}\|_\it^{1/2}.
\end{align*}
We use $\ft{P^\perp\al}=\ft{\al}-\ft{P\al}$ and $|\ft{\al}|^2\leq \ft{\gam}(1-\ft{\gam})\leq 1/4$ pointwise to find $\|\ft{P^\perp\al}\|_\it\leq \frac{1}{4} + \|\ft {P\al}\|_\it$. It is slightly more convenient to conclude the argument by choosing an orthonormal basis $\{a_j\}$ for $\ker(K_{T_c}+V)$. This allows us to write
\beq
\label{eq:decomp}
P\al = h\sum_{j=1}^n \psi_j a_j
\eeq
By Proposition \ref{prop:regularity1}, $\|\ft{a_j}\|_\it\leq C$ for all $j$ and therefore $\|\ft{P\al}\|_{\it}\leq O(h) |\psi|_{\it}$. We have shown
\beqs
	\|\ft{\al}\|_4\geq \|\ft{P\al}\|_4 - O(h^{1/2}) \|\ft{P\al}\|_2^{1/2} (1+h |\psi|_\it)^{1/2}.
\eeqs
Combining this with \eqref{eq:pf2}, we obtain
\beq
\label{eq:pf3}
	\|\ft{P\al}\|_4 \leq O(h^{1/2}) \|\ft{P\al}\|_2^{1/2} (1+h |\psi|_\it)^{1/2}+O(h).
\eeq
It remains to bound $\|\ft{P\al}\|_4$ from below in terms of $\|\ft{P\al}\|_2$. Let $R>0$. We split the integration domain into $\{p\leq R\}$ and $\{p> R\}$. Applying H\"older's inequality to the former yields
\beq
\label{eq:pf4}
	\|\ft{P\al}\|_2^2\leq C R^{D/2} \|\ft{P\al}\|_4^2 + h^2|\psi|_\it^2 \sum_{i,j} \int\limits_{\{p>R\}} |\ft{a}_i||\ft{a}_j|\dpp
\eeq
where $C>0$ denotes a constant independent of $h,R$. Note that for all $i,j$, Cauchy-Schwarz implies $|\ft{a}_i||\ft{a}_j|\in L^1(\set{R}^D)$ and so for $R_0>0$ large enough,
\beqs
\sum_{i,j} \int\limits_{\{p>R_0\}} |\ft{a}_i||\ft{a}_j|\dpp < \frac{1}{2}.
\eeqs
We recall \eqref{eq:pf3} to find
\beqs
		\|\ft{P\al}\|_2^2\leq O(h) \|\ft{P\al}\|_2 (1+h |\psi|_\it)+ \frac{1}{2} h^2 |\psi|_\it^2 +O(h^2).
\eeqs
Since the $\{a_j\}$ in \eqref{eq:decomp} are orthonormal, $h|\psi|_\it\leq \|\ft{P\al}\|_2\leq \sqrt{n} h|\psi|_\it$. This implies
\beq
\label{eq:almostresult}
		 \l(\frac{1}{2}+O(h)\r)|\psi|_\it^2 \leq O(1)|\psi|_\it+O(1).
\eeq
Let $h$ be small enough such that the $1/2+O(h)$ term exceeds $1/4$. We conclude that $|\psi|_\it\leq O(1)$. Since $\|\ft{P\al}\|_2\leq \sqrt{n} h|\psi|_\it$, it follows that $\|\ft{P\al}\|_2\leq O(h)$ as claimed.
\end{proof}

\subsubsection{Lower bound: Part B}
We use once more the relative entropy identity \eqref{eq:REidentity}. Together with Lemma \ref{lm:SC} (i) and the eigenvalue equation, we get
\beq
\begin{aligned}
\label{eq:LBTI1}
&\curly{F}^{BCS}_{T}(\Gam)-\curly{F}^{BCS}_{T}(\Gam_0)\\
= &h^4\curly{E}^{GL}(P\al)+\frac{T}{2} \goth{H}(\Gam,\Gam_\Del) + \int V |\al-P\al|^2\dx+O(h^6).
\end{aligned}
\eeq
We see that to prove the lower bound it remains to show
\beq
\label{eq:LBTI2}
		\frac{T}{2} \goth{H}(\Gam,\Gam_\Del) + \int V |\al-P\al|^2\dx = \frac{T}{2} \goth{H}(\Gam,\Gam_\Del) + \int V |P^\perp \al|^2\dx\geq O(h^6).
\eeq
By Lemma \ref{lm:lbonentropy} and the fact that $x\mapsto x/\tanh(x)$ is a monotone function that depends only on $x^2$, we have
\beq
\begin{aligned}
	\frac{T}{2} \goth{H}(\Gam,\Gam_\Del) 
	&\geq \frac{1}{2}\Tr{(\ft\Gam-\ft\Gam_\Del)\frac{\ft H_\Del}{\tanh\l(\ft H_\Del/(2T)\r)}(\ft\Gam-\ft\Gam_\Del)}\\
	&=\frac{1}{2}\Tr{(\ft\Gam-\ft\Gam_\Del)\frac{E_\Del}{\tanh\l(E_\Del/(2T)\r)}(\ft\Gam-\ft\Gam_\Del)}\\
	&\geq \frac{1}{2}\Tr{(\ft\Gam-\ft\Gam_\Del) K_T (\ft\Gam-\ft\Gam_\Del)}.
\end{aligned}
\eeq
Since $K_T\geq 0$, we have for every fixed (i.e.\ $h$-independent) $0<\eps<1$,	
\beqs
\begin{aligned}
	&\frac{1}{2}\Tr{(\ft\Gam-\ft\Gam_\Del) K_T(\ft\Gam-\ft\Gam_\Del)}\\
	&\geq\int K_T |\ft{\al}-\ft{\al}_\Del|^2\dpp\\
	&\geq \int K_T | \ft{P^\perp \al}|^2 \dpp -2\Re \int	K_T \ol{\ft{P^\perp \al}} \l(\ft{P\al}-\ft{\al}_\Del	\r)\dpp\\
	&\geq (1-\eps)\int K_T | \ft{P^\perp \al}|^2 \dpp - C_\eps\int K_T |\ft{P\al}-\ft{\al}_\Del|^2\dpp\\
	&\geq (1-\eps)\int K_T | \ft{P^\perp \al}|^2 \dpp + O(h^6).
\end{aligned}
\eeqs
In the last step, we used Lemma \ref{lm:SC} (ii) and $K_T(p)\leq C \jap{p}^2$ to get
\beq
		\int K_T |\ft{P\al}-\ft{\al}_\Del|^2\dpp = O(h^6).
\eeq
Using these estimates on \eqref{eq:LBTI2} and setting $\xi:=P^\perp \al$, we see that it remains to show that there exists an $h$-independent choice of $0<\eps<1$ such that
\beq
\label{eq:oh6}
   \scp{\xi}{((1-\eps)K_T+V)\xi} \geq  O(h^6).
\eeq
Recall from step 2 of the proof of Proposition \ref{prop:xibded} that $\|K_T-K_{T_c}\|_\it\leq O(h^2)$. Since also $\|\xi\|_2=O(h^2)$ by Proposition \ref{prop:xibded}, we get
$$
\scp{\xi}{((1-\eps)K_T+V)\xi}= \scp{\xi}{((1-\eps)K_{T_c}+V)\xi} + O(h^6).
$$
We claim that there exists a constant $c>0$ such that
\beq
\label{eq:claim}
\scp{\xi}{(K_{T_c}+V) \xi}\geq c \scp{\xi}{K_{T_c} \xi}.
\eeq
Choosing $\eps$ sufficiently small will then give $\scp{\xi}{(1-\eps)K_{T_c}+V)\xi}\geq 0$. Thus, it remains to prove \eqref{eq:claim}. Since $V_-$ is infinitesimally form-bounded with respect to $K_{T_c}$, we have for any $\de>0$
\beq
(1-\de) K_{T_c}\leq K_{T_c}-V_- +C_\de\leq K_{T_c} +V +C_\de
\eeq
or
\beq
K_{T_c}\leq C_1 (K_{T_c} +V)+C_2.
\eeq
Now, on $\ker(K_{T_c}+V)^\perp$, it also holds that $K_{T_c}+V-\kappa\geq 0$ where $\kappa>0$ denotes the gap size. Thus, for all $\lam>0$,
\beq
K_{T_c}\leq (C_1+\lam) (K_{T_c} +V)+C_2 -\lam\kappa,\quad \textnormal{on} \ker(K_{T_c}+V)^\perp
\eeq
and choosing $\lam=C_2/\kappa$, we see that \eqref{eq:claim} follows. This proves (i).

Statement (ii) was proved along the way: Any approximate minimizer satisfies \eqref{eq:approximatemin} and hence Proposition \ref{prop:xibded} implies that its off-diagonal part can be split into $\al=P\al+\xi$ with $\|\xi\|=O(h^2)$. Since $P$ is the projection onto $\ker(K_{T_c}+V)$, $P\alpha\in\ker(K_{T_c}+V)$. Moreover, $\goth{a}_0\equiv h^{-1} P\alpha$ approximately minimizes the GL energy because the proof of the lower bound shows that for all $\Gam$ satisfying \eqref{eq:approximatemin} (not just for actual minimizers), 
\beqs
\curly{F}^{BCS}_T(\Gam)-\curly{F}^{BCS}_T(\Gam_0)\geq h^4 \curly{E}^{GL}(\goth{a}_0)+O(h^6).
\eeqs
This finishes the proof of Theorem \ref{thm:mainTI}.
\qed

\subsection{Proofs of Propositions \ref{prop:formbounded}, \ref{prop:GL} and \ref{prop:new}}
\be{proof}[Proof of Proposition \ref{prop:formbounded}]
For the $L^{p_V}$ potentials, this is a standard argument combining H\"older's inequality and Sobolev's inequality. 

Consider the potentials \eqref{eq:Vdeltadefns}, i.e.\ $V(\x)=-\lam\de(|\x|-R)$ with $\lam,R>0$. Let $f\in H^1(\set{R}^D)$. 
We first consider the case $D=1$. Then
$$
\scp{f}{Vf}=-\lam |f(R)|^2.
$$
We apply the simplest Sobolev inequality
\beq
\label{eq:simplesobolev}
2\sup_{x\in\set{R}}|u(x)| \leq \int_{-\it}^\it |u'(x)| \d x, \quad \forall u\in W^{1,1}(\set{R}),
\eeq
(which follows from the fundamental theorem of calculus) with the choice $u(s)=f(s)^2$. By \eqref{eq:simplesobolev} and Cauchy-Schwarz, we get
$$
|f(R)|^2\leq \int_{-\it}^\it |f(x) f'(x)| \d x\leq \eps \|f'\|_2^2 +\frac{1}{4\eps} \|f\|_2^2
$$
for any $\eps>0$. This proves the claimed infinitesimal form-boundedness of $V$ when $D=1$.

Let now $D=2,3$. We have
\beqs
\scp{f}{Vf}= -\lam \int_{\set{S}^{D-1}} R^{D-1} |f(R\om)|^2 \d\sigma(\om),
\eeqs
where $\d\sigma$ is the usual surface measure on $\set{S}^{D-1}$. Observe that the inequality \eqref{eq:simplesobolev} implies
\beqs
2\sup_{s>0}|u(s)| \leq \int_0^\it |u'(s)| \d s, \quad \forall u\in W^{1,1}_0(\set{R}_+).
\eeqs
We use this with the choice $u(s)=s^{D-1}f(s\om)^2$, pointwise in $\om\in\set{S}^2$, and find
\beq
\label{eq:parent}
\begin{aligned}
&\int_{\set{S}^{D-1}} R^{D-1} |f(R\om)|^2 \d\sigma(\om)\\
&\leq \int_{\set{S}^{D-1}} \int_0^\it \l(\frac{D-1}{2}s^{D-2} |f(s\om)|^2+ s^{D-1} |f(s\om)\del_s f(s\om)|\r)  \d s\d\sigma(\om).
\end{aligned}
\eeq
Consider the first term in the parentheses. We split the integration domain into $s>1$ and $s\leq 1$ and estimate $s^{D-2}<s^{D-1}$ in the first region. By applying H\"older's inequality in the second region, we get
\beqs
\begin{aligned}
\int_{\set{S}^{D-1}} \int_0^\it s^{D-2} |f(s\om)|^2 \d s\d\sigma(\om) 
&< \|f\|_2^2 + \int_{\set{S}^{D-1}} \int_0^1 s^{D-2} |f(s\om)|^2 \d s\d\sigma(\om)\\
&\leq 
\|f\|_2^2 + \l(\int_0^1 s^{D-8/3}\d s\r)^{3/5}\|f\|_5^2\\
&= \|f\|_2^2 + C\|f\|_5^2
\end{aligned}
\eeqs
where $C$ is a finite constant, since $D-8/3>-1$. The $L^5$ norm is infinitesimally form-bounded with respect to $-\nabla^2$ by the usual argument via Sobolev's inequality.

We come to the second term in \eqref{eq:parent} in parentheses. By Cauchy-Schwarz, for every $\eps>0$, it is bounded by
\beqs
\eps \int_{\set{S}^{D-1}} \int_0^\it s^{D-1} |\del_s f(s\om)|^2  \d s\d\sigma(\om)+\frac{1}{4\eps} \|f\|_2^2.
\eeqs
The first term is the quadratic form corresponding to (the negative of) the radial part of the Laplacian, see \eqref{eq:radialLaplacian}. It differs from the full Laplacian by a multiple of the Laplace-Beltrami operator $-\nabla^2_{\set{S}^{D-1}}$, i.e.\ a nonnegative operator. This implies infinitesimal form-boundedness when $D=2,3$. 
\e{proof}

\be{proof}[Proof of Proposition \ref{prop:GL}]
Recall \eqref{eq:alternative} 
\beqs
\begin{aligned}
\curly{E}^{GL}(\goth{a})
=&\frac{1}{T_c} \int_{\set{R}^D} \frac{g_1((p^2-\mu)/T_c)}{(p^2-\mu)/T_c} \l|K_{T_c}(p)\r|^4 |\goth{a}(\p)|^4\dpp\\
&-
\frac{1}{2T_c} \int_{\set{R}^D}  \frac{1}{\cosh^2\l(\frac{p^2-\mu}{2T_c}\r)} \l|K_{T_c}(p)\r|^2 |\goth{a}(\p)|^2\dpp,
\end{aligned}
\eeqs
We denote the quartic term by $A(\goth{a})$ and the quadratic term by $-B(\goth{a})$. Note that $A,B>0$ whenever $\goth{a}$ is not identically zero. 

We use the basis representation of the GL energy mentioned in Remark \ref{rmk:mainTI} (i). That is, we fix a basis $\{a_j\}$ of $\ker(K_{T_c}+V)$ and write $\goth{a}(\p)=\sum_{j=1}^n\psi_j \ft a_j(\p)$ with $(\psi_1,\ldots,\psi_n)\in\set{C}^n$. Then we write $$(\psi_1,\ldots,\psi_n)=L\om, \qquad L\geq 0,\, \om\in S(\set{C}^n),$$ where  $S(\set{C}^n)$ is the unit sphere in $\set{C}^n$. It follows that
\beqs
\begin{aligned}
\inf_{(\psi_1,\ldots,\psi_n)\in\set{C}^n} \curly{E}^{GL}(\psi_1,\ldots,\psi_n)
 &= \inf_{\om\in S(\set{C}^n)} \inf_{L\geq 0} \curly{E}^{GL}(L\om)\\ 
 &= \inf_{\om\in S(\set{C}^n)} \inf_{L\geq 0} \l(L^4 A(\om)-L^2 B(\om)\r)\\
&= \inf_{\om\in S(\set{C}^n)} \frac{-B(\om)^2}{4A(\om)}
\end{aligned}
\eeqs
and since $A,B$ are continuous functions which never vanish on the compact set $S(\set{C}^n)$, the last infimum is finite and attained.
\e{proof}

\be{proof}[Proof of Proposition \ref{prop:new}]
The same argument that proves Theorem \ref{thm:mainTI} (ii) applies for $T>T_c$ and yields the same result with the sign of the $|\goth{a}|^2$ term in the GL energy \eqref{eq:alternative} flipped. Consequently, the unique minimizer of the GL energy is $\goth{a}=0$. To see coercivity of the GL energy around this minimizer, we drop the quartic term and rewrite the the quadratic term as in the proof of Proposition \ref{prop:GL} above. We get
\beqs
\curly{E}^{GL}(\psi_1,\ldots,\psi_n)\geq \eps  \lam_{min}\sum_{j=1}^n |\psi_j|^2 
\eeqs
with
$$
\lam_{min}:=\min_{\om\in S(\set{C}^n)} \frac{1}{2T_c} \int_{\set{R}^D}  \frac{1}{\cosh^2 \l(\frac{p^2-\mu}{2T_c}\r)} \l|K_{T_c}(p)\r|^2 \l|\sum_{j=1}^n\om_j\ft a_j(\p)\r|^2\dpp.
$$
Note that $\lam_{min}>0$, since it is the minimum of a positive, continuous function over a compact set.
\e{proof}

\section{Proofs for part II}

\subsection{Setting}
We use the formulation of GL theory from Remark \ref{rmk:mainTI}(i). We compute the GL coefficients $c_{ijkm}$ and $d_{ij}$ given by formulae \eqref{eq:cdefn} and \eqref{eq:ddefn}. They determine the GL energy $\curly{E}_{\text{$d$-wave}}^{GL}:\set{C}^5\rightarrow \set{R}$ via
\beqs
\curly{E}^{GL}\l(\tilde{\psi}_{-2},\ldots,\tilde{\psi}_2\r) = \sum_{i,j,k,m=-2}^2 c_{ijkm} \ol{\tilde{\psi}_i \tilde{\psi}_j} \tilde{\psi}_k \tilde{\psi}_m - \sum_{i,j=-2}^2 d_{ij} \ol{\tilde{\psi}_i} \tilde{\psi}_j
\eeqs
It remains to pick a convenient basis to compute \eqref{eq:cdefn} and \eqref{eq:ddefn}. Since the Fourier transform maps $\curly{H}_l$ to itself in a bijective fashion, see e.g.\ \cite{SteinWeiss}, we can choose
\beq
\label{eq:basischoiced}
\ft{a}_{m}(\p) = \vr(p)\, Y_m^2(\vartheta,\vp),\qquad \p\equiv (p,\vt,\vp),
\eeq
for an appropriate radial function $\vr$. We will denote the GL order parameter corresponding to $\ft{a}_m$ (in the sense of \eqref{eq:content}) by $\tilde{\psi}_m$ with $-2\leq m\leq 2$. (Note that we use the \emph{ordinary} spherical harmonics $Y_2^m$ \eqref{eq:SH} as a basis because it is more convenient to do computations, but our final result is phrased in terms the basis of \emph{real} spherical harmonics \eqref{eq:realSH}.) 

With the choice \eqref{eq:basischoiced}, equations \eqref{eq:cdefn},\eqref{eq:ddefn} for the GL coefficients read
\begin{align}
\label{eq:cTI}
		c_{ijkm} &= \int p^{-2}f_4(p)  \ol{ Y_2^i(\vartheta,\vp)Y_2^j(\vartheta,\vp)}		Y_2^k(\vartheta,\vp)Y_2^m(\vartheta,\vp)\dpp\\
		\label{eq:dTI}
		d_{ij} &= -\int p^{-2} f_2(p)\ol{Y_2^i(\vartheta,\vp)} Y_2^j(\vartheta,\vp)(p)\dpp,
\end{align}
where $i,j,k,m=-2,\ldots,2$ and we used the functions $f_2,f_4$ defined in \eqref{eq:f24defn}. Note that $f_2,f_4$ are positive (since $g_1$ defined by \eqref{eq:g0defn} satisfies $\frac{g_1(z)}{z}>0$) and radially symmetric.

%%%%%%%%%%%%%%%%%%%%%%%%%%%%%%%%%%%%%%%%%%
\subsection{Proof of Theorem \ref{thm:pureTI}}
While the radial integrals in \eqref{eq:cTI},\eqref{eq:dTI} depend on the details of the microscopic potential $V$ through $\vr$, the integration over the angular variables can be performed explicitly. Since the spherical harmonics form an orthonormal family with respect to surface measure on $\set{S}^2$, we immediately get
\beqs
	d_{ij}=d\delta_{ij}
\eeqs
where $d>0$ is the result of the radial integration in \eqref{eq:dTI}, i.e.\
\beq
\label{eq:deq}
		d= \int_0^\it  f_2(p)  \,\d p
\eeq
and this is the second relation claimed in \eqref{eq:cdefnn}. \\

Next, we consider \eqref{eq:cTI}. Firstly, note that $c_{ijkm}$ is always proportional to the result of the radial integration in \eqref{eq:cTI}, i.e.\
\beq
\label{eq:ceq}
		c= \int_0^\it  f_4(p)  \,\d p
\eeq
and this is the first relation claimed in \eqref{eq:cdefnn}.

\be{table}[t]
\be{center}
\be{tabular}{|c|c|c|c|c|}
 $i$ & $j$ & $k$&	$m$& $c_{ijkm}\cdot28 \pi$\\
 \hline
 $2$ & $2$ & $2$&	$2$& $10 c$\\
 \hline
 $2$ & $1$ & $2$&	$1$& $5c$\\
 \hline
 $1$ & $1$ & $1$&	$1$& $10c$\\
 $0$ & $2$ & $0$&	$2$& $5c$\\
 $1$ & $1$ & $0$&	$2$& $0$\\
 \hline
 $0$ & $1$ & $0$&	$1$& $5c$\\
 $-1$ & $2$ & $-1$&	$2$& $5c$\\
 $0$ & $1$ & $-1$&	$2$& $0$\\
 \hline
 $0$ & $0$ & $0$&	$0$& $15c$\\
 $1$ & $-1$ & $1$&	$-1$& $10c$\\
  $2$ & $-2$ & $2$&	$-2$& $10c$\\
  $0$ & $0$ & $2$&	$-2$& $5c$\\
  $0$ & $0$ & $1$&	$-1$& $-5c$\\
  $1$ & $-1$ & $2$&	$-2$& $-5c$\\
 \hline
 \e{tabular}
\e{center}
\caption{Non-trivial equivalence classes of Ginzburg--Landau coefficients in the pure $d$-wave case. $c$ is defined as the result of the radial integration \eqref{eq:ceq}. Notice that the case $i+j=0$ behaves rather differently. This is due to the fact that the ``pair permutation'' and ``pair sign-flip'' symmetries fall together in this case. We keep the factor $5$ to ensure better comparability with Table \ref{table:GLmixed} later on.}
\label{table:GLpure}
\e{table}

It remains to compute the angular part of the integral in \eqref{eq:cTI}. We express the product of two spherical harmonics of angular momentum $l=2$ as a linear combination of spherical harmonics of angular momentum ranging from $l=0$ to $l=4$. The general relation involves the well-tabulated Clebsch-Gordan coefficients, which we denote by $\langle l_1,l_2;m_1,m_2\vert L;M\rangle$, and can be found in textbooks on quantum mechanics (see e.g.\ \cite{CCT} p.\ 1046):
\beq
\label{eq:clebsch}
\begin{aligned}
		Y_{l_1}^{m_1}(\vartheta,\vp)Y_{l_2}^{m_2}(\vartheta,\vp)		
		 = \sum_{L=|l_1-l_2|}^{l_1+l_2} &\sqrt{\frac{(2l_1+1)(2l_2+1)}{4\pi(2L+1)}} \langle l_1,l_2;0,0\vert L;0\rangle\\
		  &\times\langle l_1,l_2;m_1,m_2\vert L;m_1+m_2\rangle Y_{L}^{m_1+m_2}(\vartheta,\vp).
\end{aligned}
\eeq
Physically, this corresponds to expressing a pair of particles, uncorrelated in the angular variable, in terms of a wave function for the composite system. Since the total angular momentum of the composite system is not determined uniquely by the product wavefunction on the left-hand side, the sum over $L$ appears on the right. However, the total $z$-component of the angular momentum is determined to be $m_1+m_2$. This ``selection rule'' will greatly restrict which $c_{ijkm}$ may be non-zero. 

Now, we can use the orthonormality of the spherical harmonics to compute the angular integrals and find
 \beq
 \label{eq:coeff3}
 \begin{aligned}
 		c_{ijkm} =  \sum_{L=0,2,4} \frac{25c}{4\pi(2L+1)} &\langle 2,2;0,0\vert L;0\rangle^2
		\langle 2,2;i,j\vert L;i+j\rangle\\ &\times \langle 2,2;k,m\vert L;k+m\rangle	\delta_{i+j,k+m},
\end{aligned}
 \eeq
where we used that the Clebsch-Gordan coefficients are real-valued and that $\langle l_1,l_2;0,0\vert L;0\rangle=0$ unless $L$ is even \cite{CCT}. Note that the selection rule from above yielded the necessary relation $i+j=k+m$ for $c_{ijkm}\neq 0$.

 There are further symmetries: Considering the original expression \eqref{eq:cTI}, we  that $c_{ijkm}=c_{jikm}=c_{ijmk}$. Since \eqref{eq:coeff3} shows $c_{ijkm}\in\set{R}$,  \eqref{eq:cTI} also implies that $c_{ijkm}=c_{kmij}$. We subsume these relations as ``pair permutation'' symmetry. Physically, they correspond to the exchange of Cooper pairs. Moreover, as can be seen from reference tables for Clebsch-Gordan coefficients, we have $c_{ijkm}=c_{(-i)(-j)km}$, to which we will refer as ``pair sign-flip'' symmetry. Physically, it is a consequence of the invariance of our system under reflection in the $xy$-plane.

 It thus suffices to look up \eqref{eq:coeff3} in a reference table for Clebsch-Gordan coefficients once for each member of a ``pair permutation''and ``pair sign-flip'' equivalence class, ignoring those tuples $(i,j,k,m)$ which do not satisfy the selection rule $i+j=k+m$. The result is presented in Table \ref{table:GLpure}. By counting the number of elements of each equivalence class, we find
\beq
\label{eq:GLpureold}
\begin{aligned}	
&\curly{E}_{\text{$d$-wave}}^{GL}\l(\tilde{\psi}_{-2},\ldots,\tilde{\psi}_2\r)\\
	  &=\frac{5c}{14\pi}\Bigg(\l(\sum_{m=-2}^2|\tilde{\psi}_m|^2-\tau\r)^2 -\tau^2+ \frac{1}{2}|\tilde{\psi}_0|^4 + 2\sum_{m=1,2} | \tilde{\psi}_m|^2| \tilde{\psi}_{-m}|^2\\ \nonumber
		 &\quad - 2 \Re\l(	\ol{\tilde{\psi}_0}^2 \tilde{\psi}_1\tilde{\psi}_{-1}\r) + 2 \Re\l(	\ol{\tilde{\psi}_0}^2 \tilde{\psi}_2 \tilde{\psi}_{-2}\r) - 4 \Re\l(	\ol{\tilde{\psi}_1\tilde{\psi}_{-1}} \tilde{\psi}_{2}\tilde{\psi}_{-2}\r)\Bigg),
\end{aligned}
\eeq
where $\tau=\frac{7\pi d}{5c}$. Notice that this expression contains a second complete square:
\beq
\begin{aligned}
	&\curly{E}_{\text{$d$-wave}}^{GL}\l(\tilde{\psi}_{-2},\ldots,\tilde{\psi}_2\r)\\
	  &= \frac{5c}{14\pi}\l(\l(\sum_{m=-2}^2|\tilde{\psi}_m|^2-\tau\r)^2-\tau^2 + \frac{1}{2}	\l|\tilde{\psi}_0^2 - 2 \tilde{\psi}_1\tilde{\psi}_{-1}+2 \tilde{\psi}_2\tilde{\psi}_{-2}\r|^2	\r)
\end{aligned}
\eeq 

To conclude Theorem \ref{thm:pureTI}, it remains to make the basis change to the real-valued spherical harmonics, i.e.\ to invert \eqref{eq:realSH}.
 On the level of the GL order parameters, this yields the $SU(5)$ transformation
\beq
\label{eq:su5trafo}
\be{aligned}
\tilde{\psi}_0 = \psi_0,\qquad &\tilde{\psi}_{-1} = \frac{-\psi_{1}+i\psi_{-1}}{\sqrt{2}},
\qquad \tilde{\psi}_{1} = \frac{\psi_{1}+i\psi_{-1}}{\sqrt{2}},\\
        &\tilde{\psi}_{-2} = \frac{\psi_{2}-i\psi_{-2}}{\sqrt{2}},
        \qquad \tilde{\psi}_{2} = \frac{\psi_{2}+i\psi_{-2}}{\sqrt{2}}.
\e{aligned}
\eeq

\subsection{Proof of Theorem \ref{thm:pureTI2D}}
The situation is as in three dimensions, only simpler. The $d_{ij}$ GL coefficients are again diagonal by orthogonality and they come with a factor $d$ defined in the same way as in Theorem \ref{thm:pureTI} but with $f_2(p)$ replaced $f_2(p)/p$ since $D=2$ (of course the definition of $\vr$ has changed as well). For the $c_{ijkm}$ coefficients, instead of considering Clebsch-Gordan coefficients, it suffices to compute
\beq
\frac{c}{\pi^2}\int_0^{2\pi} \cos(2\vp)^{k} \sin(2\vp)^{4-k} \d \vp
\eeq
for all $0\leq k\leq 4$. Here, the GL coefficient $c$ is defined in the same way as in Theorem \ref{thm:pureTI}. We omit the details.

\subsection{Proof of Theorem \ref{thm:mixedTI}}

We compute $\curly{E}^{GL}_{\text{$(s+d)$-wave}}$ by using the formulae \eqref{eq:cdefn} and \eqref{eq:ddefn} for the GL coefficients as in the previous section. We already computed most of the GL coefficients, namely all the ones that couple $d$-waves to $d$-waves.

By orthonormality of the spherical harmonics, $d_{ij}$ is still diagonal. For $i,j\neq s$, $d$ is as in \eqref{eq:deq}. Notice however that $d$ depends on $\varrho$ through $f_2$.
When $i=j=s$, we have to replace $\varrho$ by $\varrho_s$, which is conveniently described as multiplication by $g_s=\l|\frac{\vr_s}{\vr}\r|$. We conclude that
\beqs
	d_{ij}=
	\be{cases}
            d^{(2s)} \quad  &\textnormal{if $i=j=s$},\\
			d\delta_{ij}  \quad  &\textnormal{otherwise}.
	\e{cases}
\eeqs
with $d^{(2s)}$ as defined in \eqref{eq:csdefn}.\\

\be{table}[t]
\be{center}
\be{tabular}{|c|c|c|c|c|}
 $i$ & $j$ & $k$&	$m$& $c_{ijkm}\cdot28\pi$\\
 \hline
 $s$ & $2$ & $0$&	$2$& $-2\sqrt{5}c^{(1s)}$\\
 $s$ & $2$ & $s$&	$2$& $7c^{(2s)}$\\
 $s$ & $2$ & $1$&	$1$& $\sqrt{30}c^{(1s)}$\\
 \hline
 $s$ & $1$ & $0$&	$1$& $\sqrt{5}c^{(1s)}$\\
 $s$ & $1$ & $s$&	$1$& $7c^{(2s)}$\\
 $s$ & $1$ & $-1$&	$2$& $-\sqrt{30}c^{(1s)}$\\
 \hline
 $s$ & $0$ & $0$&	$0$& $2\sqrt{5}c^{(1s)}$\\
 $s$ & $s$ & $0$&	$0$& $7c^{(2s)}$\\
 $s$ & $0$ & $s$&	$0$& $7c^{(2s)}$\\
 $s$ & $s$ & $s$&	$0$& $0$\\
 $s$ & $s$ & $s$&	$s$& $7c^{(4s)}$\\
  $s$ & $0$ & $2$&	$-2$& $-2\sqrt{5}c^{(1s)}$\\
  $s$ & $s$ & $2$&	$-2$& $7c^{(2s)}$\\
  $s$ & $0$ & $1$&	$-1$& $-\sqrt{5}c^{(1s)}$\\
  $s$ & $s$ & $1$&	$-1$& $-7c^{(2s)}$\\
 \hline
 \e{tabular}
\caption{Equivalence classes of new Ginzburg--Landau coefficients in the mixed $(s+d)$-wave case. $c^{(1s)},c^{(2s)},c^{(4s)}$ are defined in \eqref{eq:csdefn}.}
\label{table:GLmixed}
\e{center}
\e{table}

We turn to the quartic GL coefficients $c_{ijkm}$. Note that the ``pair permutation'' and ``pair sign-flip'' symmetries described in the
proof of Theorem \ref{thm:pureTI} still hold. In addition to the results listed in Table \ref{table:GLpure}, we now have equivalence classes of $c_{ijkm}$ where some indices are equal to $s$. Since the corresponding $\ft{a}_{s}$ carry zero momentum in the $z$-direction, the selection rule dictates that
$c_{ijkm}$ can only be non-zero if the $s$ replaces a $0$-index. 

We thus consider all equivalence classes of GL coefficients that
can be obtained by replacing a $0$ in Table \ref{table:GLpure} by $s$.
We compute their values again via \eqref{eq:clebsch} (some follow
immediately from the fact that $Y_0^0=1/\sqrt{4\pi}$). The results are presented in Table \ref{table:GLmixed}.

Just as for $d_{ij}$, the $c^{(1s)},c^{(2s)},c^{(4s)}$ are the result of a radial integration where for each index equal to $s$, $f_4$ is multiplied by a
factor $g_s$. This yields the expressions \eqref{eq:csdefn} for $c^{(1s)},c^{(2s)},c^{(4s)}$.
Note that according to Table \ref{table:GLmixed}, $c_{sss0}=0$ and thus it is not necessary to define $c^{(3s)}$.

Armed with Table \ref{table:GLmixed}, it remains to count the number of GL coefficients in each equivalence class.
After some algebra, we obtain
\beq
\begin{aligned}
\label{eq:GLdoldbasis}
&\curly{E}_{\text{$(s+d)$-wave}}^{GL}\l(\tilde{\psi}_{s},\tilde{\psi}_{-2},\ldots,\tilde{\psi}_2\r)\\
=&\curly{E}^{GL}_{\text{$d$-wave}}(\tilde{\psi}_{-2},\ldots,\tilde{\psi}_2)+\curly{E}^{GL}_{\text{$s$-wave}}(\tilde{\psi}_{s})
+\curly{E}^{GL}_{\text{coupling}}(\tilde{\psi}_s,\tilde{\psi}_{-2},\ldots,\tilde{\psi}_2).
\end{aligned}
\eeq
where
\beqs
\begin{aligned}
&\curly{E}^{GL}_{\text{coupling}}(\tilde{\psi}_s,\tilde{\psi}_{-2},\ldots,\tilde{\psi}_2)\\
= &\frac{\sqrt{5}c^{(1s)}}{7\pi}	\l(2	\Re\l[\ol{\tilde{\psi}_{s}}\tilde{\psi}_0	\l(\sum_{m=0,\pm 1} |\tilde{\psi}_m|^2 -2\sum_{m=\pm 2} |\tilde{\psi}_m|^2\r)\r]\vphantom{\Re\l[\ol{\tilde{\psi}_{1}}^2 \tilde{\psi}_{s}\tilde{\psi}_{2}\r]}\r.\\
 &\qquad\qquad\qquad\quad - 2\sum_{m=1,2} m \Re\l[\ol{\tilde{\psi}_{s}}\ol{\tilde{\psi}_{0}}\tilde{\psi}_{m}\tilde{\psi}_{-m}\r] \\	
&\l.\qquad\qquad\qquad\quad + \sqrt{6} \sum_{\sigma=\pm 1}\l(		\Re\l[\ol{\tilde{\psi}_{\sigma}}^2 \tilde{\psi}_{s}\tilde{\psi}_{2\sigma}\r]		-2\Re\l[\ol{\tilde{\psi}_{s}}\ol{\tilde{\psi}_{\sigma}}\tilde{\psi}_{-\sigma}\tilde{\psi}_{2\sigma}\r]\r) \r)\\
 	&\quad +\frac{c^{(2s)}}{2\pi}\l(	2	|\tilde{\psi}_{s}|^2 \sum_{m=-2}^2 |\tilde{\psi}_m|^2 +\Re\l[\ol{\tilde{\psi}_{s}}^2\l(\tilde{\psi}_0^2-2\tilde{\psi}_1\tilde{\psi}_{-1}+2\tilde{\psi}_{2}\tilde{\psi}_{-2}			\r)\r] \r)
\end{aligned}
\eeqs
where $\curly{E}^{GL}_{\text{$d$-wave}}(\tilde{\psi}_{-2},\ldots,\tilde{\psi}_2)$ is given by \eqref{eq:GLpureold} and
\beqs
\curly{E}^{GL}_{\text{$s$-wave}}(\tilde{\psi}_s)=\frac{c^{(4s)}}{4\pi}\l(\l(	|\tilde{\psi}_{s}|^2 -	\tau_s\right)^2-\tau_s^2\r)
\eeqs
with $\tau_s=\frac{2 \pi d^{(2s)}}{c^{(4s)}}$. Statement (i) in Theorem \ref{thm:pureTI}, which gives the expression for
$\curly{E}_{\text{$(s+d)$-wave}}^{GL}$, now follows by transforming into the basis of real spherical harmonics via \eqref{eq:su5trafo}.\\

To prove (ii), we use the GL energy expressed in the basis of real spherical harmonics.
 Let $\eps>0$ and take $(\psi_{-2},\ldots,\psi_2)\in\curly{M}_{\text{$d$-wave}}$, the set of minimizers of
  $\curly{E}_{\text{$d$-wave}}$ described by \eqref{eq:minimizerconditionsTI}. Set $\psi_s=\eps \om$ with $|\om|=1$ and note that
  \beqs
  \begin{aligned}
    \curly{E}_{\text{$(s+d)$-wave}}^{GL}\l(\psi_{s},\psi_{-2},\ldots,\psi_2\r)
    = &\inf\curly{E}^{GL}_{\text{$d$-wave}} + \eps \Re[\ol{\om} z]\\
     &+ \eps^2 \l(\frac{\tau  c^{(2s)}}{\pi}-\frac{\tau_s c^{(4s)}}{2\pi}\r) + \frac{c^{4s}}{4\pi}\eps^4.
\end{aligned}  
  \eeqs
  for some $z\in\set{C}$, which is independent of $\eps$ and $w$. Consider first the case that $(\psi_{-2},\ldots,\psi_2)\in\curly{M}_{\text{$d$-wave}}$ is such that $z\neq0$. Then, we can choose $\om$ such that $Re[\ol{\om} z] <0$ and we obtain \eqref{eq:thmmixedTI} for sufficiently small $\eps$. Thus, suppose that $z=0$, which is e.g.\ the case for $(0,\tau/\sqrt{2},0,i\tau/\sqrt{2},0)\in\curly{M}_{\text{$d$-wave}}$. It is then clear that \eqref{eq:thmmixedTI} holds iff $\frac{\tau c^{(2s)}}{\pi}<\frac{\tau_s c^{(4s)}}{2\pi}$, or equivalently $d c^{(2s)} < \frac{5}{7}  cd^{(2s)}$. This proves (ii).\\

For statement (iii), let $\psi_s$ be a minimizer of $\curly{E}^{GL}_{\text{$s$-wave}}$, i.e.\ $|\psi_s|^2=\tau_s$. Now let
$\eps>0$ and let $(\psi_{-2},\ldots,\psi_2)$ have entries of the form $\psi_m=\eps \psi_m'$ with $|\psi_m'|<1$. We have
    \begin{align*}
    &\curly{E}_{\text{$(s+d)$-wave}}^{GL}\l(\psi_{s},\psi_{-2},\ldots,\psi_2\r)\\
    &=\min \curly{E}^{GL}_{\text{$s$-wave}}+\eps^2 \l(\l(-d+ \frac{c^{(2s)}\tau_s}{\pi}\r)\sum_m|\psi_m'|^2
    +\frac{c^{(2s)}}{2\pi} \Re\l[\ol{\psi_s^2}\sum_{m=-2}^2 \l(\psi_m'\r)^2\r]\r)\\
    &\quad+O(\eps^3)
   \end{align*}
as $\eps\rightarrow 0$. The real part is clearly minimal when we choose $\Arg(\psi_m')=\Arg(\psi_s)+\pi/2$ for all $m$ with $\psi_m'\neq 0$.
This choice yields
\begin{align*}
    &\curly{E}_{\text{$(s+d)$-wave}}^{GL}\l(\psi_{s},\psi_{-2},\ldots,\psi_2\r)\\
    &=\min \curly{E}^{GL}_{\text{$s$-wave}}+\eps^2  \sum_m|\psi_m'|^2 \l(-d + \frac{c^{(2s)}\tau_s}{\pi}\r)+O(\eps^3).
\end{align*}
When the term in parentheses is strictly negative, which is equivalent to $d^{(2s)} c^{(2s)}< d c^{(4s)}$, we see that
$\curly{E}_{\text{$(s+d)$-wave}}^{GL}<\min \curly{E}^{GL}_{\text{$s$-wave}}$ for sufficiently small $\eps$.
Vice-versa, when the term in parentheses is strictly positive, $\curly{E}_{\text{$(s+d)$-wave}}^{GL}>\min \curly{E}^{GL}_{\text{$s$-wave}}$
for all small $\eps>0$.

To conclude statement (iii), it remains to consider the case $d^{(2s)} c^{(2s)}=d c^{(4s)}$, when the $O(\eps^2)$-term vanishes.
The leading correction is now given by the $O(\eps^3)$-term
and by choosing $\psi_m=0$ for $m\neq0$, we find
\beqs
    \curly{E}_{\text{$(s+d)$-wave}}^{GL}\l(\psi_{s},0,0,\psi_0,0,0\r) = \min \curly{E}^{GL}_{\text{$s$-wave}}
    +\eps^3\frac{2\sqrt{5}c^{(s)}}{7\pi} |\psi_0'|^2\Re[\ol{\psi_s}\psi_0']+O(\eps^4).
\eeqs
Letting $\Arg(\psi_0')=\Arg(\psi_s)+\pi$ shows that $\curly{E}_{\text{$(s+d)$-wave}}^{GL}\l(\psi_{s},0,0,\psi_0,0,0\r)
<\min \curly{E}^{GL}_{\text{$s$-wave}}$ in this case as well.
 This proves statement (iv).
\qed

\section{Proofs for part III}
\label{sect:proofsIII}
The proof of Theorem \ref{thm:appendix} is based on three steps.
\be{itemize}
	\item In Lemma \ref{lm:monotone}, we solve the eigenvalue problem for $K_T+V_{\lam,R}$ explicitly in each angular momentum sector $\curly{H}_l$. The key result is the ``eigenvalue condition'' \eqref{eq:monotone} which gives a formula for the eigenvalue (or energy) $E$ in terms of the other parameters $l,\mu,T$ and $\lam$. We will see that one can solve this for $\lam$ and one obtains an integral formula which is \emph{monotone} in $E$. Therefore, instead of showing that $E$ is minimal for $l=l_0$, one can equivalently show that $\lam$ is minimal for $l=l_0$. 
	\item In Lemma \ref{lm:j}, we show how, by adapting the parameters $\mu,T$ of the ``weight function'' $p/K_T(p)$, one can conclude that $\int_0^\it  \frac{p}{K_T(p)} f(p)\d k$ is positive, if one assumes that $f$ is strictly positive on an interval. 
	\item By Theorem \ref{thm:bessel}, for any half-integer Bessel function of the first kind $\curly{J}_{l_0+1/2}$, there exists an open interval around its first maximum on which it is strictly larger than (the absolute value of) all other half-integer Bessel functions. 
\e{itemize}
The idea is then to use the eigenvalue condition \eqref{eq:monotone} to rephrase the question whether some state in $\curly{H}_{l_0}$ has lower energy than all states in $\curly{H}_l$ as the more tangible question whether the quantity
\beqs
\int_0^\it \l(\curly{J}^2_{l_0+1/2}(p)-\curly{J}^2_{l+1/2}(p)\r) \frac{p}{K_T(p)} \d p
\eeqs
is positive. By Theorem \ref{thm:bessel} there is an interval of $p$-values on which the integrand is positive and by Lemma \ref{lm:j} there are intervals of $\mu$- and $T$-values such that the entire integral is positive. 

\subsection{Solving the eigenvalue problem}
For any radial $V$, we can block diagonalize $K_T+V$ by using the orthogonal decomposition of $L^2(\set{R}^3)$ into angular momentum sectors \eqref{eq:L2decomposition}, namely $L^2(\set{R}^3)=\bigoplus_{l=0}^\it \curly{H}_l$ with $\curly{H}_l$ defined in \eqref{eq:Hldefn}. It is well-known \cite{SteinWeiss} that the Fourier transform leaves each $\curly{H}_l$ invariant. Consequently, if we have $\al\in H^1(\set{R}^3)$ satisfying the eigenvalue equation
\beq
\label{eq:evalue}
\l(K_T+V\r)\al=E\al,
\eeq
then we can decompose it as $\al=\sum_l \al_l$ with $\al_l\in \curly{H}_l$ mutually orthogonal. Taking the Fourier transform of \eqref{eq:evalue} and using the fact that $V\al_l\in \curly{H}_l$ since $V$ is radial, we get from orthogonality
\beqs
 K_T(p)\ft \al_l(\p) + \ft{V\al_l}(\p)=E\ft \al_l(\p),
\eeqs
for every $l\geq0$ and a.e.\ $\p\in\set{R}^3$. Thus, we can study each component $\al_l$ separately. When $V_{\lam,R}$ is the specific radial potential \eqref{eq:Vdeltadefns}, we can say even more.

\be{lm}
\label{lm:monotone}
 Let $V_{\lam,R}$ be as in \eqref{eq:Vdeltadefns} and let $l$ be a non-negative integer. We write $\curly{J}_{l+\frac{1}{2}}$ for the Bessel function of the first kind of order $l+1/2$. Let $E<2T$ if $\mu\geq 0$ and $E<\frac{|\mu|}{\tanh(|\mu|/(2T))}$ if $\mu<0$.
Then
\beq
\label{eq:nonempty}
\ker\l(K_T+V_{\lam,R}-E\r)\cap \curly{H}_l\neq \emptyset
\eeq
is equivalent to the ``eigenvalue condition''
\beq
\label{eq:monotone}
    1 =  \lam \int_0^\it \frac{p R}{K_T(p)-E} \curly{J}^2_{l+\frac{1}{2}}(pR)\d p.
\eeq
Moreover, if \eqref{eq:nonempty} holds, then $\ker\l(K_T+V_{\lam,R}-E\r)=\spa\{\rho_l\}\otimes \curly{S}_l$ with
\beq
\label{eq:efct}
\rho_l(r)= r^{-1/2} \int_0^\it p \frac{\curly{J}_{l+\frac{1}{2}}(rp)\curly{J}_{l+\frac{1}{2}}(Rp)}{K_T(p)-E}\d p.
\eeq
\e{lm}

 Since $|\curly{J}_{l+\frac{1}{2}}(p)|\leq C p^{-1/2}$, the numerator in \eqref{eq:monotone} and \eqref{eq:efct} poses no threat for convergence of the integral. 
 
    %\item By the Sobolev trace theorem, $\al_{l}$ is well-defined a.e.\ on the sphere where $|x|=R$ and so $V_{\lam,R}\al_{l,m}$ makes sense.
%    \item

\be{proof}
By the definition of $\curly{H}_l$, we have
\beqs
    \al_l(\x) = \sum_{m=-l}^l \al_{l,m}(r) Y_l^m(\vt,\vp),\qquad \x\equiv (r,\vt,\vp).
\eeqs
We suppose $\al_l$ satisfies $\l(K_T+V_{\lam,R}\r)\al_{l}=E\al_{l}$.  Recall that the Fourier transform not only leaves each $\curly{H}_l$ invariant, it also reduces to the Fourier-Bessel transform $\curly{F}_{l}$ on it \cite{SteinWeiss}. That is, a function of the form $f(\x)=g(r) Y_{l}^m(\vt,\vp)$ has Fourier transform given by
\beq
\label{eq:FB1}
    \ft{f} (\p) = i^{-l} \l(\curly{F}_{l} g\r)(p)\, Y_{l}^m(\vt,\vp), \qquad \p\equiv (p,\vt,\vp),
\eeq
where the Fourier-Bessel transform reads
\beq
\label{eq:FB2}
    \curly{F}_{l} g(p) = \int_0^\it
    s^{3/2} p^{-1/2} \curly{J}_{l+\frac{1}{2}}(sp)  g(s)  \d s.
\eeq
We apply the Fourier transform to the eigenvalue equation. By \eqref{eq:FB1} and orthogonality of the spherical harmonics,
\beqs
    (K_T(p)-E) \curly{F}_l \al_{l,m}(p) + \curly{F}_l(V_{\lam,R}\al_{l,m})(p) = 0
\eeqs
for all $m$ and a.e.\ $\p\in\set{R}^3$. The assumption on $E$ is such that $K_T(p)-E>0$ and therefore
\beq
\label{eq:evalue1}
    \curly{F}_l \al_{l,m}(p)=-\frac{\curly{F}_l(V_{\lam,R}\al_{l,m})(p)}{K_T(p)-E}.
\eeq
So far we only used that the potential is radial. Since $V_{\lam,R}=-\lam \de(|\cdot|-R)$,
\begin{align*}
    -\curly{F}_l(V_{\lam,R}\al_{l,m})(p)
    &=-\al_{l,m}(R) (\curly{F}_l V_{\lam,R})(p) \\
    &=\lam\al_{l,m}(R) R^{3/2}p^{-1/2} \curly{J}_{l+\frac{1}{2}}(Rp).
\end{align*}
Plugging this back into \eqref{eq:evalue1}, we find the following explicit expression for the solution to the eigenvalue problem:
\beq
\label{eq:Fourierrep}
    \curly{F}_l \al_{l,m}(p)=\lam\al_{l,m}(R) R^{3/2}p^{-1/2}  \frac{\curly{J}_{l+\frac{1}{2}}(Rp)}{K_T(p)-E}
\eeq
Now we apply $\curly{F}_l^{-1}$ which, by unitarity of the Fourier transform, has the operator kernel $r^{-1/2}k^{3/2} \curly{J}_{l+\frac{1}{2}}(rk)$ when evaluated at $r>0$. For all $r>0$, we have
\beqs
\al_{l,m}(r)=\al_{l,m}(R) \lam R^{3/2} r^{-1/2} \int_0^\it p \frac{\curly{J}_{l+\frac{1}{2}}(rp)\curly{J}_{l+\frac{1}{2}}(Rp)}{K_T(p)-E}\d p
\eeqs
Note that we may assume that for some $m$, $\al_{l,m}(R)\neq0$, since otherwise $\al_l\equiv0$. Evaluating the above expression for that particular $m$ at $r=R$ gives \eqref{eq:monotone}. We write $\al_{l,m}(R)=c_{l,m}  \lam^{-1}R^{-3/2}$ and absorb $c_{l,m}$ into the angular part $\curly{S}_l$ to get \eqref{eq:efct}. Clearly the argument works in reverse, proving the claimed equivalence.
\e{proof}

\subsection{Choosing $\mu$ and $T$}
From now on, let $\mu>0$. The following lemma concerns the quantity 
\beqs
	\int_0^\it \frac{p}{K_T(p)} f(p) \d p.
\eeqs
Suppose we know that $f>\eps$ on some interval $I$, while $f$ may be negative outside of $I$. Our goal in this section is to choose the right values of $\mu$ and $T$ such that the above integral is then also positive. 

The basic idea is to view $p/K_T(p)$ as a weight function which is centered at the point $p=\sqrt{\mu}$, where it takes a value proportional to $T^{-1}$. By making $T$ small enough, we can ensure that the neighborhood of the point $p=\sqrt{\mu}$ dominates in the above integral. By choosing $\sqrt{\mu}\in I$ and $T$ sufficiently small, the integral will pick up mostly points where $f$ is positive and will therefore yield a positive value itself. This is spelled out in the following lemma.

We will eventually apply this lemma with $f=\curly{J}_{l_0+1/2}^2-\curly{J}_{l+1/2}^2$ and positivity of the above integral will translate via \eqref{eq:monotone} to the statement that the angular momentum sector $\curly{H}_{l_0}$ has lower energy than $\curly{H}_{l}$.

\be{lm}
\label{lm:j}
Let $f:\set{R}_+\rightarrow \set{R}$ be a continuous function satisfying $|f(p)|\leq C_f(1+p^2)^{-1/2}$ for some $C_f>0$. Suppose there exists $\eps>0$ and an interval $(a,b)$ such that $f>\eps$ on $(a,b)$. Then, 
\be{enumerate}[label=(\roman*)]
\item
for every $\de>0$ small enough, there exists $T_*>0$ and an interval $I$ such that for every $\mu\in I$ and $T\in (0,T_*)$,
\beq
    \int_0^\it \frac{p}{K_T(p)}f(p)\d p>0.
\eeq
\item letting $\de:=\frac{b^2-a^2}{4}$, one can choose
\beq
\label{eq:T*defn}
I:=(a^2+\de,b^2-\de),\qquad T_*:= \frac{\de}{2} \exp\l(-{\frac{2C_f  \l(\sqrt{1+2b^2}+\frac{1}{2b}\r)}{\eps\de}}\r).
\eeq
\e{enumerate}
\e{lm}

\be{proof}
Let $\mu\in (a^2+\de,b^2-\de)$. Since $\frac{p}{K_T(p)}>0$ and $|\tanh|\leq 1$, we can estimate
\beq
    \int_0^\it \frac{p f(p)}{K_{T}(p)}\d p
    \geq -C_f \int_{[0,a)\cup (b,\it)} \frac{p}{(1+p^2)^{1/2}|p^2-\mu|}\d p +\eps \int_a^b \frac{p}{K_T(p)} \d p.
\eeq
In the first integral, we estimate pointwise
\beqs
|p^2-\mu|^{-1}\leq \de^{-1} (\chi_{\{p\leq 2b\}}+2p^{-2}\chi_{\{p>2b\}})
\eeqs
 with $\chi_A$ denoting the characteristic function of a set $A$. This gives
\beq
 -C_f \int_{[0,a)\cup (b,\it)} \frac{p}{(1+p^2)^{1/2}|p^2-\mu|}\d p 
\geq -C_f \de^{-1} \l(\sqrt{1+2b^2}+\frac{1}{2b}\r)
\eeq
In the second integral, we change variables and use $\mu\in (a^2+\de,b^2-\de)$ with $\tanh(u)/u>0$ to get
\beqs
\int_a^b \frac{p}{K_T(p)} \d p=\frac{1}{2} \int_{\frac{a^2-\mu}{2T}}^\frac{b^2-\mu}{2T} \frac{\tanh(u)}{u} \d u> \int_{-\frac{\de}{2T}}^{\frac{\de}{2T}} \frac{\tanh(u)}{u} \d u> \log\l(\frac{\de}{2T}\r),
\eeqs
where in the last step we also used that $\tanh x\geq 1/2$ for $x\geq 1$. Combining everything, we get
\beq
\int_0^\it \frac{p f(p)}{K_{T}(p)}\d p\geq -C_f \de^{-1} \l(\sqrt{1+2b^2}+\frac{1}{2b}\r)+\eps \log\l(\frac{\de}{2T}\r).
\eeq
The claim follows from some algebra.
\e{proof}

\subsection{Proof of Theorem \ref{thm:appendix}}
\label{ssect:concluding}
\be{proof}[Proof of (i)]
By rescaling the parameters $\mu,\lam$ and $T$, we may assume that $R=1$. We fix a non-negative integer $l_0$ and invoke Theorem \ref{thm:bessel} to get $\eps>0$ and an interval $(a,b)$ on which $\curly{J}_{l_0+1/2}^2-\curly{J}_{l+1/2}^2>\eps$ for all $l\neq l_0$. Then we apply Lemma \ref{lm:j} to 
\beqs
f:=\curly{J}_{l_0+1/2}^2-\curly{J}_{l+1/2}^2,
\eeqs
which satisfies
\beq
\label{eq:fbound}
|f(p)|\leq 2^{1/2}(1+p^2)^{-1/2}.
\eeq
and so $C_f=2^{1/2}$ in Lemma \ref{lm:j}. To prove \eqref{eq:fbound}, ones uses statement (ii) in Lemma \ref{lm:last} to get $\curly{J}^2_{\nu}(p)\leq M_{\nu}^2(p) \leq \frac{1}{p}$ for all $\nu$. Together with $|\curly{J}_\nu|\leq 1$ from (9.1.60) in \cite{AbramowitzStegun}, this implies $|f(p)|\leq \min\{1,p^{-1}\}$ and hence \eqref{eq:fbound}. 

Note that $T_*$ and $I$ defined in Lemma \ref{lm:j} (ii) work for all $l\neq l_0$, because they depend on $f$ only through $(a,b)$, which is uniform in $f$ by Theorem \ref{thm:bessel}, and through $C_f=2^{3/2}$. Hence, Lemma \ref{lm:j} provides $T_*>0$ and an interval $I$ such that for all $\mu\in I$, all $T<T_*$ and all $l\neq l_0$ we have
\beq
\label{eq:cc}
\int_0^\it \frac{p}{K_T(p)}\l(\curly{J}_{l_0+1/2}^2(p)-\curly{J}_{l+1/2}^2(p)\r) \d p >0
\eeq
For every non-negative integer $l$, we define the function
\beq
\label{eq:lamldefn}
\lam_{l}(T,\mu):=\l(\int_0^\it \frac{p}{K_T(p)}\curly{J}_{l+1/2}^2(p) \d p\r)^{-1}
\eeq
which is chosen such that $\lam$ satisfies the eigenvalue condition \eqref{eq:monotone} with $E=0$. We write
\beqs
E_l(T,\mu,\lam):= \inf \mathrm{spec} \l(K_T+V_{\lam,1}\r)\big\vert_{\curly{H}_l}.
\eeqs
 With these definitions, Lemma \ref{lm:monotone} says
\beq
\label{eq:Eldefn}
E_l(T,\mu,\lam_l(T,\mu))=0
\eeq
At the heart of our proof is the following monotonicity argument. For all $\mu\in I$, all $T<T_*$ and all $l\neq l_0$, we have
\beq
\label{eq:vararg}
0=E_{l}(T,\mu,\lam_{l}(T,\mu))<E_{l}(T,\mu,\lam_{l_0}(T,\mu)),
\eeq
where the inequality holds by the variational principle applied to the operator $(K_T+V_{\lam,1})\big\vert_{\curly{H}_l}$ and the observation that \eqref{eq:cc} is equivalent to $\lam_{l_0}(T,\mu)<\lam_{l}(T,\mu)$. (The inequality is strict because $\scp{\al}{V_{\lam,1}\al}=-\lam \al(R)^2$ is either strictly monotone decreasing in $\lam$ or identically zero and in the latter case the energy has to be at least $2T$.)

This would already prove \eqref{eq:appendix0} and \eqref{eq:appendix1} under the condition that one fixes $T<T_*$ and determines $\lam$ through \eqref{eq:cc}. We find it physically more appealing to fix $\lam$ small enough and determine $T$ instead. To this end, we observe that $T\mapsto \lam_{l_0}(T,\mu)$ is monotone increasing, because $T\mapsto K_T(p)$ is monotone increasing for every $p>0$. Therefore, for every $\mu\in I$, we have the monotone increasing inverse function
\begin{align*}
(0,\lam_{l_0}(T_*,\mu))&\rightarrow (0,T_*)\\
\lam&\mapsto T(\lam,\mu)
\end{align*}
satisfying $\lam_{l_0}(T(\lam,\mu),\mu)=\lam$. To remove the $\mu$-dependence from the maximal value for $\lam$, we set
\beq
\label{eq:lam*defn}
\lam_*:=\min_{\mu\in \ol{I}} \lam(T_*,\mu)
\eeq
and note that $\lam_*>0$ since the integral in \eqref{eq:lamldefn} is continuous in $\mu$ by dominated convergence. For $\lam<\lam_*$, \eqref{eq:Eldefn} and \eqref{eq:vararg} become
\beqs
E_{l_0}(T(\lam,\mu),\mu,\lam)=0,\qquad E_{l}(T(\lam,\mu),\mu,\lam)>0,\quad \forall l\neq l_0.
\eeqs
This proves that for all $\mu\in I$ and all $\lam<\lam_*$, there exists $T_0<T_*$ (namely $T_0:=T(\lam,\mu)$) such that \eqref{eq:appendix0} holds (modulo restoring the $R$ parameter). Moreover, \eqref{eq:appendix1} is a direct consequence of the explicit characterization of $\ker(K_{T_c}+V)$ in Lemma \ref{lm:monotone}. Finally, \eqref{eq:appendix2} follows via the variational principle from the observation that $T\mapsto K_T(p)$ is strictly increasing for all $p>0$ and so $T\mapsto E_{l_0}(T,\mu,\lam)$ is strictly increasing as well, as long as it stays below $2T$.
\e{proof}

\be{proof}[Proof of (ii)]
Consider the function
\beqs
	\de_T:\mu \mapsto\int_0^\it \frac{p}{K_T(p)} \l(\curly{J}^2_{1/2}(p)-\curly{J}^2_{5/2}(p)\r)\d p.
\eeqs
\dashuline{Claim:} \emph{There exists $T_{**}>0$ such that for all $0<T<T_{**}$ there exists $\mu_T>0$ such that $\de_T(\mu_T)=0$.
Moreover, $\sqrt{\mu_T}\rightarrow z_{1/2}$ as $T\rightarrow 0$, where $z_{1/2}=\min\{z>0: \curly{J}_{1/2}^2(z)=\curly{J}_{5/2}^2(z)\}$.}

The claim follows essentially from the intermediate value theorem. Before we give the details, we explain how one may conclude statement (ii) from the claim. Let $0<T<T_{**}$. By definition \eqref{eq:lamldefn}, $\de_T(\mu_T)=0$ implies $\lam_0(T,\mu_T)=\lam_2(T,\mu_T)$. By Lemma \ref{lm:monotone} and using the notation \eqref{eq:Eldefn}, 
\beq
E_0(T,\mu_T,\lam_0(T,\mu_T))=E_2(T,\mu_T,\lam_0(T,\mu_T))=0.
\eeq
This implies $\subset$ in \eqref{eq:appendix4} according to Lemma \ref{lm:monotone}. Equation \eqref{eq:appendix5} follows by the same monotonicity argument as in the proof of statement (i) above.

In order to prove \eqref{eq:appendix3} with the choices $\mu\equiv \mu_T$ and $\lam\equiv\lam_0(T,\mu_T)$ and the remaining $\supset$ in \eqref{eq:appendix4}, we shall show that there exists $T_*\in(0,T_{**}]$ such that for all $0<T<T_*$, 
\beq
\label{eq:remains}
E_l(T,\mu_T,\lam_0(T,\mu_T))>0,\quad \forall l\geq 4, l \textnormal{ is even.}
\eeq
By Theorem \ref{thm:bessel} (ii) (with $l_0=1$) and Lemma \ref{zl}, there exists an open interval containing $z_{1/2}$ such that
\beqs
\curly{J}_{5/2}^2 - \sup_{\substack{l\geq 4\\ l\ \text{even}}}\curly{J}_{l+1/2}^2>\eps'\quad \text{on this interval}. 
\eeqs
As in part (i), Lemma \ref{lm:j} provides $T_{**}>0$ and an interval $I'$ containing $z_{1/2}^2$ such that for all $\mu\in I'$, all $T<T_{**}$ and all even $l\geq 4$ we have
\beq
\label{eq:cc2}
\int_0^\it \frac{p}{K_T(p)}\l(\curly{J}_{5/2}^2(p)-\curly{J}_{l+1/2}^2(p)\r) \d p >0 \,.
\eeq
Since the second part of the claim gives $\mu_T\rightarrow z_{1/2}^2$ as $T\rightarrow 0$, we may assume, after decreasing $T_{**}$ to $T_*$ if necessary, that $\mu_T \in I'$ for all $0<T<T_{*}$. Therefore \eqref{eq:cc2} implies that $\lambda_0(T,\mu_T) = \lambda_2(T,\mu_T)<\lambda_l(T,\mu_T)$ for all $T<T_{*}$ and all even $l\geq 4$. By the same variational argument as in \eqref{eq:vararg}, this implies \eqref{eq:remains}.

We now prove the claim. The reader may find it helpful to consider Figure \ref{fig:bessel}. Since $\mu\mapsto K_T(p)$ is continuous for every $p$, $\mu\mapsto \de_T$ is also continuous by dominated convergence. 
Let $x_l$ ($l=0,2$) denote the first maximum of $\curly{J}_{l+1/2}$. It is well-known that $x_0<x_2$ \cite{PalmaiApagyi} and that $\curly{J}_{1/2}^2(x_0)> \curly{J}_{5/2}^2(x_0)$ and $\curly{J}_{5/2}^2(x_2)>\curly{J}_{1/2}^2(x_2)$ (which is also a very special case of our Theorem \ref{thm:bessel} (i)). By continuity these inequalities hold also in neighborhoods of $x_0$ and $x_2$. Therefore Lemma \ref{lm:j} provides open intervals $I_l\subset\set{R}_+$ ($l=0,2$), containing $x_l$, and a $T_{**}>0$ such that for all $T<T_{**}$, we have $\de_T>0$ on $I_0$ and $\de_T<0$ on $I_2$. By the intermediate value theorem, for any $T<T_{**}$ there is a $\mu_T\in [\sup I_0,\inf I_2]$ with $\de_T(\mu_T)=0$. This proves the first part of the claim.

We are left with showing that $\sqrt{\mu_T}\rightarrow z_{1/2}$ as $T\to 0$. Since $\mu_T\in[\sup I_0,\inf I_2]$ is bounded, it has a limit point as $T\rightarrow 0$. We argue by contradiction and assume that there is a limit point $\tilde z$ different from $z_{1/2}$. By Lemma \ref{zl}, $z_{1/2}$ is also the position of the first critical point of $\curly{J}_{l+3/2}$. By the interlacing properties of the zeros of Bessel functions and their derivatives, see e.g.\ \cite{PalmaiApagyi}, $z_{1/2}\in(x_0,x_2)$ and there is no other point $z\in(x_0,x_2)$ at which $\curly{J}_{1/2}^2(z)=\curly{J}_{5/2}^2(z)$. Therefore $\curly{J}^2_{1/2}-\curly{J}^2_{5/2}$ is either strictly positive or strictly negative at $\tilde z$ and, by continuity, also in an open interval containing $\tilde z$. Lemma \ref{lm:j} provides an open interval $\tilde I$ containing $\tilde z^2$ and a $\tilde T>0$ such that $\delta_T(\mu)$ is either strictly positive or strictly negative for all $T<\tilde T$ and $\mu\in\tilde I$. Since $\tilde z$ is a limit point of $\sqrt{\mu_T}$, there is a sequence $T_m\to 0$ with $\mu_{T_m}\to\tilde z^2$. In particular, $\mu_{T_m}\in\tilde I$ and $T_m<\tilde T$ for all sufficiently large $m$. Thus, $\delta_{T_m}(\mu_{T_m})$ is either strictly positive or strictly negative for all sufficiently large $m$. This, however, contradicts the construction of $\mu_T$, according to which $\delta_T(\mu_T)=0$ for all $T<T_{**}$. Thus, we have shown that $\sqrt{\mu_T}\rightarrow z_{1/2}$.
\e{proof}

\be{appendix}
\section{Properties of Bessel functions}
While one might expect the following fact about Bessel functions to be known, it appears to be new:

\emph{At its first maximum, a half-integer Bessel function is strictly larger than (the absolute value of) all other half-integer Bessel functions.}

The precise statement is in Theorem \ref{thm:bessel} (i) below. It extends to families of Bessel functions $\{\curly{J}_{\nu+k}\}_{k\in\set{Z}_+}$ with $\nu\in [0,1]$, in particular to the family of integer Bessel functions. We acknowledge a helpful discussion on mathoverflow.net \cite{MObessel} that led to Lemma \ref{lm:last}. 

Let $l_0$ be a non-negative integer. We recall that the Bessel function $\curly{J}_{l_0+1/2}$ (of the first kind, of order $l_0+1/2$) vanishes at the origin and then increases to its first maximum, whose location we denote as usual by $j_{l_0+1/2,1}'$. The following theorem says that at $j_{l_0+1/2,1}'$, $\curly{J}_{l_0+1/2}^2$ is strictly larger than any other $\curly{J}_{l+1/2}^2$ with $l$ a non-negative integer different from $l_0$.

\be{thm}
\label{thm:bessel}
Let $\set{Z}_+$ denote the set of non-negative integers and let $l_0\in\set{Z}_+$. Recall that $j_{l_0+1/2,1}'$ denotes the position of the first maximum of $\curly{J}_{l_0+1/2}$.
\be{enumerate}[label=(\roman*)]
\item
There exist $\eps>0$ and an open interval $I$ containing $j_{l_0+1/2,1}'$ such that
\beq
\label{eq:bessel}
    \curly{J}_{l_0+1/2}^2 - \sup_{l\in \set{Z}_+\setminus \{l_0\}}\curly{J}_{l+1/2}^2> \eps\quad \text{on } I.
\eeq

\item 
If $l_0\geq 1$, then $\curly J_{l_0-1/2}(j_{l_0+1/2,1}') = \curly J_{l_0+3/2}(j_{l_0+1/2,1}')$ and there exist $\eps'>0$ and an open interval $I'$ containing $j_{l_0+1/2,1}'$ such that
\beq
\min\{ \curly{J}_{l_0-1/2}^2,\curly{J}_{l_0+3/2}^2\} - \sup_{\substack{l\geq l_0+3 \\ l-l_0 \ \mathrm{odd}}}\curly{J}_{l+1/2}^2>\eps'\quad \text{on } I'.
\eeq
\e{enumerate}
\e{thm}

\be{rmk}
Statement (i) is the key result and implies Theorem \ref{thm:appendix} (i). Statement (ii) is used to prove Theorem \ref{thm:appendix} (ii).
%The reader might wonder if it is possible to use a similar continuity argument together with more general estimates to obtain $\ker(K_T+V_{\lam,R}) = \curly{H}_{l_0} \cup \curly{H}_{l_1}$ for arbitrary even $l_0\neq l_1$. However, this is not possible: As one can see from Figure \ref{fig:bessel} that near the first point of intersection of, say, $\curly{J}^2_{5/2}$ and $\curly{J}^2_{9/2}$, the function $\curly{J}^2_{1/2}$ greatly exceeds both of them. This means the method would not be able produce $\ker(K_T+V_{\lam,R}) = \curly{H}_{2} \cup \curly{H}_{4}$, because $\curly{H}_{0}$ would have a lower energy.
\e{rmk}

 \begin{figure}[t]
\begin{center}
\includegraphics[height=7cm]{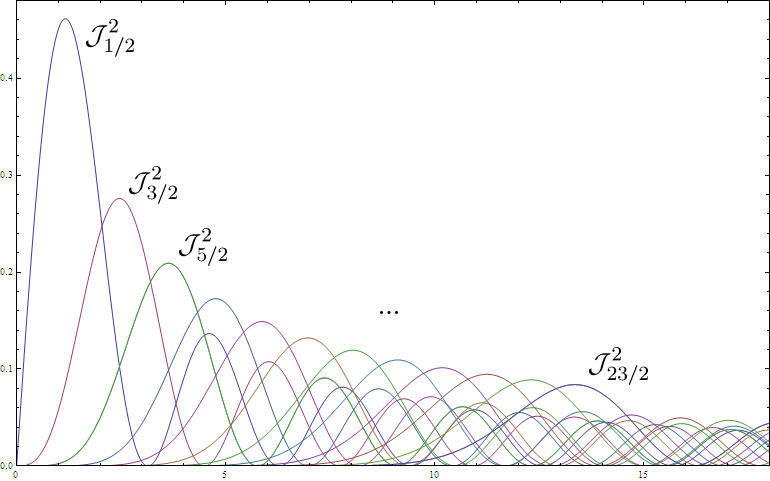}

\caption{A plot of the squared Bessel functions $\curly{J}^2_{1/2},\curly{J}^2_{3/2},\curly{J}^2_{5/2},\ldots,\curly{J}^2_{23/2}$. Observe that
 in an open interval around its maximum, each function is the largest one among all the shown ones (in particular it is the largest among all the $\curly{J}_{l+1/2}^2$ according to Lemma \ref{lm:last-2}).}
\label{fig:bessel}
\end{center}
\end{figure}

The proof of (i) in Theorem \ref{thm:bessel} is split into three Lemmata, each treating one of the following three regimes of $l$:
\beqs
\begin{aligned}
    \curly{L}_>:&=\setof{l\in\set{Z}_+}{l>l_0},\\
    \curly{L}_\lesssim:&=\setof{l\in\set{Z}_+}{l< l_0,\,j_{l+1/2,1}\geq j_{l_0+1/2,1}'},\\
    \curly{L}_\ll:&=\setof{l\in\set{Z}_+}{l< l_0,\,j_{l+1/2,1} < j_{l_0+1/2,1}'}.
\end{aligned}
\eeqs
Here, as usual, $j_{l+1/2,1}$ denotes the first positive zero of $\curly{J}_{l+1/2}$. The most cumbersome regime is $\curly{L}_\ll$. The proof there is based on a combination of some hands-on elementary estimates and bounds on the zeros of Bessel functions and their derivatives, which we could not find in the usual reference books
\cite{AbramowitzStegun},\cite{Watson}. The first regime $\curly{L}_>$ is the easiest

\be{lm}
\label{lm:last-2}
There exist $\eps_1>0$ and an open interval $I_1$ containing $j_{l_0+1/2,1}'$ such that
\beq
    \curly{J}_{l_0+1/2}^2 - \sup_{l>l_0}\curly{J}_{l+1/2}^2 \geq \eps_1\quad \text{on } I_1
   \eeq
\e{lm}

\be{proof}
According to \cite{Landau00}, the function
\beqs
    \nu\mapsto \max_y|\curly{J}_\nu(y)|
\eeqs
is strictly decreasing. Therefore
\beqs
\eps_1:=\frac{1}{2}\l(\curly{J}_{l_0+1/2}^2(j_{l_0+1/2,1}') - \max_y \curly{J}_{l_0+3/2}^2(y)\r)
\eeqs
is strictly positive. By continuity, there exists an open interval $I_1$ containing $j_{l_0+1/2,1}'$ such that for all $x\in I_1$, 
\beqs
|\curly{J}_{l_0+1/2}^2(j_{l_0+1/2,1}')-\curly{J}_{l_0+1/2}^2(x)|< \eps_1.
\eeqs
For $x\in I_1$, we have
\begin{align*}
\curly{J}_{l_0+1/2}^2(x)- \sup_{l>l_0} \curly{J}_{l+1/2}^2(x)
&> -\eps_1 + \curly{J}_{l_0+1/2}^2(j_{l_0+1/2,1}')- \sup_{l>l_0} \max_y \curly{J}_{l+1/2}^2(y)\\
&\geq -\eps_1 + \curly{J}_{l_0+1/2}^2(j_{l_0+1/2,1}') - \max_y \curly{J}_{l_0+3/2}^2(y)\\
&=\eps_1.\qedhere
\end{align*} 
\e{proof}

\be{lm}
\label{lm:last-1}
There exist $\eps_2>0$ and an open interval $I_2$ containing $j_{l_0+1/2,1}'$ such that
\beqs
    \curly{J}_{l_0+1/2}^2 - \sup_{l\in \curly{L}_{\lesssim}}\curly{J}_{l+1/2}^2 \geq \eps_2\quad \text{on } I_2.
\eeqs
\e{lm}

\be{proof}
Since the supremum of finitely many continuous functions is itself continuous, it suffices to prove $\curly{J}_{l_0+1/2}^2(j_{l_0+1/2,1}')> \curly{J}_{l+1/2}^2(j_{l_0+1/2,1}')$ for every $l\in \curly{L}_{\lesssim}$.
We define the sequence $\{a_l\}_{l\in \curly{L}_{\lesssim}}$ by
\beq
\label{eq:a_ldefn}
    \curly{J}_{l+1/2}(j_{l_0+1/2,1}')= a_l \curly{J}_{l_0+1/2}(j_{l_0+1/2,1}').
\eeq
With this definition, the recurrence relation for Bessel functions from (9.1.27) in \cite{AbramowitzStegun} appears in the form of a second-order difference equation
\beq
\label{eq:secondorderdiff}
    a_{l-1} = 2\frac{l+1/2}{x_0} a_l-a_{l+1}
\eeq
with initial conditions $a_{l_0}=1$ and $a_{l_0-1}=(l_0+1/2)/j_{l_0+1/2,1}'$. It is well-known that the latter quantity is strictly less than one, see eq.\ (3) on p.\ 486 of \cite{Watson}. Moreover, $a_l\geq 0$ for all $l\in \curly{L}_{\lesssim}$, because $j_{l+1/2,1}\geq j_{l_0+1/2,1}'$ and all Bessel functions are positive before they first become zero. An easy induction lets us conclude from \eqref{eq:secondorderdiff} that $a_l<a_{l+1}<1$ for all $l\in \curly{L}_{\lesssim}$. In particular, $a_l\leq a_{l_0-1}=(l_0+1/2)/j_{l_0+1/2,1}'<1$. Recalling the definition \eqref{eq:a_ldefn} of $a_l$, this proves the claim.
\e{proof}

We finally come to the  regime $\curly{L}_{\ll}$. As a tool, we will use the ``modulus'' function defined by
\beqs
M_\nu:= \sqrt{\curly{J}_\nu^2+\curly{Y}_\nu^2},
\eeqs
where $\curly{Y}_\nu$ is the Bessel function of the second kind. The first two statements of the following Lemma are known facts about the modulus function. Statement (iii) is the key result to derive (iv). 

\be{lm}
\label{lm:last}
\be{enumerate}[label=(\roman*)]
\item The map $\nu \mapsto M_\nu(x)$ is strictly increasing for all $x>0$.
\item For all $x>\nu$,
\beqs
    M_\nu^2(x) < \frac{2}{\pi} \frac{1}{\sqrt{x^2-\nu^2}}.
\eeqs
\item
If $l_0\geq 11$, there exists $l_1<l_0$ such that we have both, 
\be{enumerate}
	\item[(iii.a)] $\curly{J}^2_{l_0+1/2}(j_{l_0+1/2,1}')>M_{l_1+1/2}^2(j_{l_0+1/2,1}')$
	\item[(iii.b)]  $j_{l_1+1/2,1}>j_{l_0+1/2,1}'$
\e{enumerate} 
\item There exist $\eps_3>0$ and an open interval $I_3$ containing $j_{l_0+1/2,1}'$ such that
\beqs
    \curly{J}_{l_0+1/2}^2 - \sup_{l\in \curly{L}_{\ll}}\curly{J}_{l+1/2}^2 \geq \eps_3\quad \text{on } I_3.
\eeqs
\e{enumerate}
\e{lm}

The intuition why such $l_1$ as in (iii) should exist is based on a heuristic argument of which we learned through \cite{MObessel}, involving asymptotic formulae for the relevant expression. To turn this into a rigorous proof, we need to replace the asymptotics by bounds that hold for all $l_0$ (or at least for all $l_0\geq 11$). \cite{GatteschiLaforgia} contains results which are sufficient for our purposes when combined with a number of elementary estimates.

\be{proof}
Statement (i) is a direct consequence of Nicholson' formula, see p.\ 444 in \cite{Watson}, and the fact that $K_0>0$.  
Statement (ii) is formula (1) on p.\ 447 of \cite{Watson}.

We come to statement (iii). For convenience, we write $m=l+1/2$, so in particular $m_0=l_0+1/2$. We also abbreviate $x_0 = j_{l_0+1/2,1}'$. The basic idea (inspired by asymptotics) is to choose
\beqs
	m_1= m_0-c m_0^{1/3}
\eeqs
with $c$ small enough to have (iii.a) hold but large enough to have (iii.b) hold. By (i), (iii.a) is implied by
\beq
\label{eq:modulus2}
     \frac{2}{\pi}  \frac{1}{\sqrt{x_0^2-m_1^2}} < \curly{J}_{m_0}^2(x_0).
\eeq
By \cite{GatteschiLaforgia}, we have the lower bound
\beq
\label{eq:x0LB}
x_0 >m_0 \exp\l(2^{-1/3}a_1'm_0^{-2/3} -1.06 m_0^{-4/3}\r)
\eeq
for all $m_0\geq 11.5$. Here, $a_1'$ is the absolute value of the first zero of the derivative of the Airy function, with a numerical value of about $1.018793$. From $m_0\geq 11.5$, we can conclude that the argument of the exponential in \eqref{eq:x0LB} is greater than $0.6 m_0^{-2/3}$. Thus, by the elementary estimate $e^y \geq 1+y$, \eqref{eq:x0LB} implies the more manageable lower bound
\beqs
x_0>m_0+0.6 m_0^{1/3}
\eeqs
Setting $m_1=m_0-c m_0^{1/3}$ with $c$ to be determined and using the above bound on $x_0$, as well as $m_0\geq 11.5$, we see that \eqref{eq:modulus2} is implied by
\beq
\label{eq:modulus3}
       \frac{2}{\pi}  \frac{1}{\sqrt{1.26+2c-0.19c^2}} <  \l(  m_0^{1/3} \max_x |\curly{J}_{m_0}(x)| \r)^{2}
\eeq
According to \cite{Landau00}, $\nu\mapsto \nu^{1/3} \max_x |\curly{J}_{\nu}(x)|$ is an increasing function and so we can estimate the right-hand side in \eqref{eq:modulus3} from below by $\nu^{2/3} \max_x \curly{J}_{\nu}(x)^2$ for any $1/2\leq \nu\leq m_0$. Unfortunately, the numerical value one obtains for the ``worst case'' $\nu=1/2$ is not good enough to also get (iii.b). Instead, we assume that $c\leq 1$ and use $m_0\geq 11.5$ to get $m_0-c m_0^{-1/3}\geq 8.5$ and so
\beqs
 \l(  m_0^{1/3} \max_x |\curly{J}_{m_0}(x)| \r)^{2} > \l(  (8.5)^{1/3} \max_x |\curly{J}_{8.5}(x)| \r)^{2}>0.42
\eeqs
where the last inequality can be read off from a plot, for example. Therefore, \eqref{eq:modulus3} holds if we can find $c\leq 1$ that satisfies
\beq
\label{eq:margin}
 \frac{2}{\pi}  \frac{1}{\sqrt{1.26+2c-0.19c^2}}<0.42
\eeq
and it is easily seen that this holds for $c\in [0.5,1]$. 

Now, we want to ensure that $c$ is also small enough to have (iii.b) hold, i.e.\ $j_{m_1}>x_0$. To this end, we invoke two more facts: 
\be{itemize}
\item 
the upper bound 
\beq
\label{eq:x0UB}
 x_0 < m_0 + 0.89 m_0^{1/3}.
\eeq
This is a consequence of the bound
\beqs
    x_0 < m_0 \exp\l(2^{-1/3}a_1' m_0^{-2/3}\r)
\eeqs
from \cite{GatteschiLaforgia}, where again $a_1'\approx 1.018793$, by noting that $m_0\geq 11.5$ implies that the argument of the exponential, call it $y$, satisfies $y<1.59$. On $[0,1.59]$, we can estimate $\exp(y)<1+1.09 y$, as one can verify e.g.\ by plotting and this yields \eqref{eq:x0UB}.

\item the lower bound
\beq
\label{eq:jmLB}
    j_{m_1}> m_1 +1.85 m_1^{1/3}
\eeq
which we obtained from the optimal lower bound proved in \cite{QuWong99} by rounding down. This is better than the bound one can derive from a corresponding result of \cite{GatteschiLaforgia} as we did above. 
\e{itemize}

From \eqref{eq:x0UB} and \eqref{eq:jmLB}, we see that $j_{m_1}>x_0$ will follow from
\beq
\label{eq:modulus4}
	(m_0-c m_0^{1/3})+1.85 (m_0-c m_0^{1/3})^{1/3}	>	m_0+0.89 m_0^{1/3}, 
\eeq
Since $c\leq 1$ and $m_0\geq 11.5$, we have $1-c m_0^{-2/3}>0.8$ and so \eqref{eq:modulus4} is implied by
\beqs
	c < (0.8)^3 * 1.85 - 0.89 =  0.827.
\eeqs
So any choice of $c\in[0.5,0.8]$ will ensure that (iii.a) and (iii.b) hold.

We prove statement (iv). By continuity, it suffices to prove $\curly{J}_{l_0+1/2}^2(x_0)>\curly{J}_{l}^2(x_0)$ for all $l\in\curly{L}_{\ll}$ (which we recall means $l<l_0$ with $j_{l+1/2}\leq x_0$). Assume first that $l_0\geq 11$. Choosing $l_1$ as in statement (iii), (iii.a) states
\beq
\label{eq:pwisebessel2}
\curly{J}_{l_0+1/2}^2(x_0)>M_{l_1+1/2}^2(x_0)
\eeq
and (iii.b) implies that $l_1\in\curly{L}_{\lesssim}$. By the monotonicity of $\nu\mapsto j_\nu$, it holds that $l\in \curly{L}_{\ll}$ implies $l<l_1$. Thus, the definition of $M_\nu$ and statement (i) imply
\beq
\label{eq:chain}
	\curly{J}_{l+1/2}^2\leq M_{l+1/2}^2 \leq M_{l_1+1/2}^2.
\eeq
Together with \eqref{eq:pwisebessel2}, this implies (iv) for $l_0\geq 11$.
Since for $l_0=0,1$ there are no $l\ll l_0$, we may assume $l_0\geq 2$. 
For $2\leq l_0\leq 10$, one can then check by hand that \eqref{eq:pwisebessel2} holds with the choice $l_1=l_0-2$. 
Since  $l_0-1\in\curly{L}_{\lesssim}$, we get that $l\in\curly{L}_\ll$ implies $l\leq l_1$ and so \eqref{eq:chain} applies for all such $l$. 
\e{proof}

\begin{lm}\label{zl}
 For any positive integer $l$, 
\beq
\label{eq:zldefn}
\begin{aligned}
\min\{z>0: \curly{J}_{l-1/2}^2(z)=\curly{J}_{l+3/2}^2(z)\}
= j_{l+1/2,1}'
\end{aligned}
\eeq
and $\curly{J}_{l-1/2},\curly{J}_{l+1/2},\curly{J}_{l+3/2}$ are positive on $(0,j_{l+1/2,1}']$.
\end{lm}

\begin{proof}
We recall the recurrence relation from (9.1.27) in \cite{AbramowitzStegun}, which says that for all $\nu,z>0$,
\beqs
\curly{J}_{\nu-1}(z)-\curly{J}_{\nu+1}(z)=2\curly{J}_{\nu}'(z).
\eeqs
Applying this with $\nu=l+1/2$, $z=j_{l+1/2,1}'$ we obtain $\curly{J}_{l-1/2}(j_{l+1/2,1}')=\curly{J}_{l+3/2}(j_{l+1/2,1}')$ and hence $\leq$ in \eqref{eq:zldefn}. Notice that by the interlacing properties of zeros and extrema of Bessel functions, see e.g.\ \cite{PalmaiApagyi}, $j_{l+1/2,1}'$ is to the left of the first positive zeros of $\curly{J}_{l-1/2},\curly{J}_{l+1/2},\curly{J}_{l+3/2}$. Since Bessel functions are positive before they reach their first positive zero, we conclude that $\curly{J}_{l-1/2},\curly{J}_{l+1/2},\curly{J}_{l+3/2}$ are positive on $(0,j_{l+1/2,1}']$. In particular, $\curly{J}_{l-1/2},\curly{J}_{l+3/2}$ are positive at the left side of \eqref{eq:zldefn}, call it $z_l$, and so we can take square roots to get $\curly{J}_{l-1/2}(z_l)=\curly{J}_{l+3/2}(z_l)$. By the recurrence relation from above, $\curly{J}_{l+3/2}'(z_l)=0$ implying $z_l\geq j_{l+1/2,1}'$, as claimed.
\end{proof}

It remains to give the

\be{proof}[Proof of Theorem \ref{thm:bessel}]
Statement (i) is a direct consequence of Lemmata \ref{lm:last-2} to \ref{lm:last}. 

For statement (ii) we first observe that for any positive integer $l$,
\begin{equation}
\label{eq:besselproof1}
\curly{J}_{l-1/2}^2>\curly{J}_{l+3/2}^2,\quad \text{on } (0,j_{l+1/2,1}').
\end{equation}
In fact, by standard asymptotics, this inequality holds near zero and, according to Lemma \ref{zl}, $j_{l+1/2,1}'$ is the first point of intersection of $\curly{J}_{l-1/2}^2$ and $\curly{J}_{l+3/2}^2$. Therefore the inequality holds on all of $(0,j_{l+1/2,1}')$, as claimed.

We now use the fact that $j_{l+1/2,1}'$ is increasing in $l$ \cite{PalmaiApagyi}. Choose $I'$ to be an open interval containing $j_{l_0+1/2,1}'$ whose closure is contained in $(0,j_{l_0+5/2,1}')$. Then by \eqref{eq:besselproof1} (with $l=l_0+2$) and continuity there is an $\epsilon'>0$ such that
$$
\curly{J}_{l_0+3/2}^2\geq \curly{J}_{l_0+7/2}^2 + \epsilon' \quad \text{on } I'.
$$
Applying \eqref{eq:besselproof1} successively with $l=l_0+4,l_0+6,\ldots$, we conclude that
$$
\curly{J}_{l_0+3/2}^2\geq \sup_{\substack{l\geq l_0+3\\l-l_0 \ \text{odd}}}\curly{J}_{l+1/2}^2 + \epsilon' \quad \text{on } I',
$$
which is one part of the claim. Finally, we want to prove the same inequality with $\curly{J}_{l_0-1/2}^2$ on the left side (with possibly smaller $\epsilon'$ and $I'$). Clearly, \eqref{eq:besselproof1} implies that this is true on $I'\cap (0,j_{l_0+1/2,1}']$. Now use continuity to find $\delta>0$ such that $\curly{J}_{l_0-1/2}^2\geq \curly{J}_{l_0+3/2}^2-\epsilon'/2$ on $[j_{l_0+1/2,1}',j_{l_0+1/2,1}'+\delta]$. Thus,
$$
\curly{J}_{l_0-1/2}^2\geq \curly{J}_{l_0+7/2}^2 + \epsilon'' \quad \text{on } I''
$$
with $\epsilon''=\epsilon'/2$ and $I''=I'\cap (0,j_{l_0+1/2,1}'+\delta)$. As before, \eqref{eq:besselproof1} now implies the inequality in part (ii). This completes the proof.
\e{proof}
\e{appendix}

%\section{Open Problems}
%We close with a short list of open problems that suggest themselves for further study given our results. 
%
%\be{enumerate}
%\item Extending Theorem \ref{thm:mainTI} to systems with weak external fields that vary on the macroscopic scale, the situation originally considered in \cite{FrankHainzlSeiringerSolovej12}. Again the $\|\cdot\|_4$ appearing in the proof appears to be the main difficulty.
%
%\item Extending Theorem \ref{thm:mainTI} to antisymmetric $\al$ in order to describe e.g.\ superfluid helium-3 which has a $p$-wave order parameter. 
%See the discussion in Section \ref{sect:pwave} for the first steps in this direction.
%
%\item Computing the GL energies for certain cylindrical interactions (which occur for actual unconventional superconductors) and consequently describing their minimizers.

%\e{enumerate}

\section*{Acknowledgements}
The authors would like to thank Egor Babaev, Christian Hainzl, Edwin Langmann and Robert Seiringer for helpful discussions. R.L.F.\ was supported by the U.S. National Science Foundation through grants PHY-1347399 and DMS-1363432

\end{document}